\documentclass[a4paper,12pt]{amsart}
\usepackage[lmargin=2.4cm,rmargin=2.25cm,bmargin=3cm,tmargin=2.5cm]{geometry}

\usepackage[utf8]{inputenc}
\usepackage{amsmath}
\usepackage{amssymb}
\usepackage{amsthm}
\usepackage{textcomp}
\usepackage{enumerate}
\usepackage[sort,numbers]{natbib}
\usepackage{array}
\usepackage{verbatim}
\usepackage{hyperref}
\usepackage{graphicx}
\usepackage{booktabs}
\usepackage{algorithm}
\usepackage{algpseudocode}
\usepackage[table]{xcolor}
\usepackage{caption}
\usepackage[labelformat=simple]{subcaption}

\hypersetup{
  colorlinks,
  linkcolor={red!50!black},
  citecolor={blue!50!black},
  urlcolor={blue!80!black}
}

\usepackage{tcolorbox}
\usepackage{soul}
\usepackage[foot]{amsaddr}

\usepackage{todonotes}

\theoremstyle{plain}
\newtheorem{theorem}{Theorem}

\newtheorem{proposition}[theorem]{Proposition}
\newtheorem{lemma}[theorem]{Lemma}

\newtheorem{definition}[theorem]{Definition}
\newtheorem{assumption}{Assumption}

\theoremstyle{remark}
\newtheorem{remark}[theorem]{Remark}

\newcommand{\ud}{\mathrm{d}} 
\newcommand{\R}{\mathbb{R}}  
\newcommand{\Z}{\mathbb{Z}}  
\newcommand{\N}{\mathbb{N}}  
\renewcommand{\P}{\mathbb{P}} 
\newcommand{\E}{\mathbb{E}}  
\newcommand{\uarg}{\;\cdot\;} 
\newcommand{\charfun}[1]{\mathbf{1}\left(#1\right)} 
\newcommand{\tv}{\mathrm{tv}}
\newcommand{\Sp}{\mathbb{S}}  
\newcommand{\X}{\mathbb{X}}  
\newcommand{\var}{\mathrm{var}} 
\newcommand{\cov}{\mathrm{cov}} 

\newcommand{\W}[1]{\mathbb{W}_{#1}}
\newcommand{\angles}[1]{\langle #1 \rangle}
\newcommand{\bigangles}[1]{\big\langle #1 \big\rangle}
\newcommand{\braces}[1]{\lbrace #1 \rbrace}
\newcommand{\bigbraces}[1]{\big\lbrace #1 \big\rbrace}
\newcommand{\Bigbraces}[1]{\Big\lbrace #1 \Big\rbrace}
\newcommand{\biggbraces}[1]{\bigg\lbrace #1 \bigg\rbrace}

\newcommand{\floor}[1]{\lfloor #1 \rfloor}

\newcommand{\C}{^c}

\newsavebox{\topleftimage}
\newsavebox{\bottomleftimage}
\newsavebox{\bigimage}

\newcounter{experiment}
\newcommand{\experiment}{
  \refstepcounter{experiment}
  \subsubsection*{Experiment \theexperiment}
}

\title[I-SIR with adaptive number of proposals]{Iterated sampling importance resampling with adaptive number of proposals}

\author{Pietari Laitinen$^1$}
\author{Matti Vihola$^1$}
\address{$^1$ Department of Mathematics and Statistics, University of Jyväskylä, Finland}

\begin{document}

\begin{abstract}
   Iterated sampling importance resampling (i-SIR) is a Markov chain Monte Carlo (MCMC) algorithm which is based on $N$ independent proposals. As $N$ grows, its samples become nearly independent, but with an increased computational cost. We discuss a method which finds an approximately optimal number of proposals $N$ in terms of the asymptotic efficiency. The optimal $N$ depends on both the mixing properties of the i-SIR chain and the (parallel) computing costs. Our method for finding an appropriate $N$ is based on an approximate asymptotic variance of the i-SIR, which has similar properties as the i-SIR asymptotic variance, and a generalised i-SIR transition having fractional `number of proposals.' These lead to an adaptive i-SIR algorithm, which tunes the number of proposals automatically during sampling. Our experiments demonstrate that our approximate efficiency and the adaptive i-SIR algorithm have promising empirical behaviour. We also present new theoretical results regarding the i-SIR, such as the convexity of asymptotic variance in the number of proposals, which can be of independent interest.
\end{abstract}
   
\maketitle

\section{Introduction}\label{sec:introduction}

Importance sampling is a ubiquitous method in Bayesian statistics which can be used to estimate expectations of the form
$$
   \pi(f) = \E_\pi[f(X)] = \int_{\X} f(x) \pi(x) \ud x,
$$
where $\pi$ is a probability density on a state space $\X$,\footnote{Usually, $\X\subset\R^d$ and the integral is with respect to the Lebesgue measure, but the method works with any underlying space $\X$ equipped with a $\sigma$-finite dominating measure `$\ud x$'; see Appendix \ref{app:setup} for details.} and $f:\X\to\R$ is a (measurable) test function. In importance sampling, $N$ independent samples are drawn from a `proposal' distribution with density $q$ which dominates $\pi$,\footnote{That is, $q(x)=0$ implies $\pi(x)=0$.} and the samples $X_i$ are associated with importance weights $W_i \propto \pi(X_i)/q(X_i)$.

The iterated sampling importance resampling (i-SIR) algorithm, which was proposed in \citep{tjelmeland} and generalised to a conditional sequential Monte Carlo algorithm in \citep{andrieu-doucet-holenstein}, is a Markov chain Monte Carlo (MCMC) method, which uses the same ingredients as importance sampling, but is iterated $n$ times with fixed (moderate) $N$, and leads to an (approximate and correlated) sample $X_1,\ldots,X_n\sim \pi$. 

In contrast with importance sampling, the i-SIR produces an unweighted sample, which can be more convenient in some applications. 
The i-SIR can also be used as `burn-in' to (self-normalised) importance sampling in order to mitigate its bias; this aspect was investigated in \citep{cardoso-samsonov-thin-moulines-olsson}. The i-SIR can also be readily used within a component-wise MCMC sampler; for instance, in a latent variable model, i-SIR can serve as a Markov update for latent variables given hyperparameters. In this context, there is another strategy to use importance sampling within MCMC, known as the `grouped independent Metropolis--Hastings' (GIMH) \citep{andrieu-roberts}. It is known to suffer from problematic mixing behaviour, which does not improve as $N$ increases, contrary to the i-SIR \citep[]{andrieu-vihola-pseudo,andrieu-lee-vihola}. 

We present an adaptive i-SIR algorithm, which automatically selects the number of proposals $N$, which minimises the cost-weighted approximate asymptotic variance \citep{glynn-whitt}. We start by presenting the i-SIR and its continuous-number-of-proposals relaxation (Section \ref{sec:isir}) and summarise some of its theoretical properties (Section \ref{sec:isir-properties}). In particular, Theorem \ref{thm:isir-asvar-properties} guarantees that the asymptotic variance of i-SIR is strictly decreasing and convex with respect to the number of proposals. We then discuss its efficiency and present an approximate efficiency in Section \ref{sec:efficiency}, which our adaptive algorithm discussed in Section \ref{section:adaptation} aims to minimise. We prove a strong law of large numbers for our adaptive i-SIR algorithm (Theorem \ref{thm:adaptive-isir-slln}), which guarantees consistency of estimates based on its output.

Sections \ref{section:inner-product-convex-sequence}--\ref{sec:approximate} are devoted to technical results required by our theoretical results. Section \ref{section:inner-product-convex-sequence} is about the behaviour of stationary covariances and acceptance rates, which lead to the convexity of the asymptotic variance in Section \ref{sec:asymptotic-variance}. Section \ref{sec:approximate} discusses i-SIR in discrete state spaces, and its alternative approximation.

Our empirical findings in Section \ref{sec:experiments} suggest that our approximate efficiency criterion can be quite robust, but also expose some problematic cases. Our adaptive algorithm shows promising behaviour in our experiments. Section \ref{sec:discussion} concludes with a discussion, including possible directions in which our algorithm could be extended. 

\section{The i-SIR algorithm}
\label{sec:isir}

The i-SIR algorithm defines a Markov chain Monte Carlo method targeting a probability density $\pi(x)$. It uses samples from a proposal probability density $q(x)$, whose support satisfies $\mathrm{supp}(q)\supset \mathrm{supp}(\pi)=:\Sp$.\footnote{The support of a density function $p:\X\to [0,\infty)$ is defined $\mathrm{supp}(p)=\braces{x\in \X:p(x)>0}$.} Then, the (unnormalised) importance weights $w(x) = c\pi(x)/q(x) \in [0,\infty)$ are well-defined for all $x$, where we use the convention $0/0=0$. The value of the constant $c>0$ need not be known; typically it stems from that we can evaluate only the unnormalised target density $\pi_u(x) = c\pi(x)$, in which case $c^{-1} = (\int \pi_u(x) \ud x)^{-1}$ is the unknown normalising constant.

Algorithm \ref{alg:isir} summarises the iteration $k$ of i-SIR with $N-1$ proposals, that is, the Markov transition $X_{k-1}\to X_k$. In Algorithm \ref{alg:isir} and hereafter, we use the notation $i{:}j$ to denote a range of integers, that is, $i{:}j = (i,i+1,\ldots,j-1,j)$ to mean integers from $i$ to $j$ and $I\sim \mathrm{Categorical}(w^{1:N})$ means a categorical distribution with weights $w^1,\ldots,w^N$, that is, $\P(I=k) = w^k$.
\begin{algorithm}
   \caption{$\textsc{iSIR}(X_{k-1}, N)$}
   \label{alg:isir}
   \begin{algorithmic}[1]
      \State Draw independent $Y_k^{2:N} \sim q$ and set $Y_k^1 = X_{k-1}$
      \State Calculate weights $W_k^i = w(Y_k^i)/\sum_{j=1}^N w(Y_k^j)$
      \State Sample $I_k \sim \mathrm{Categorical}(W_k^{1:N})$
      \State Return $X_k = Y_k^{I_k}$
   \end{algorithmic}
\end{algorithm}

The i-SIR algorithm does not involve an explicit `accept/reject' decision; instead, sampling $I_k=1$ leads to a `rejection' because $X_k = X_{k-1}$. In the case $N=2$, the i-SIR coincides with independence sampler having Barker's acceptance rate \cite{barker}. We focus on the i-SIR of Algorithm \ref{alg:isir}, but note that the algorithm can be extended to cover dependent proposals \citep[e.g.][]{mendes-scharth-kohn,karppinen-vihola} with similar flavour to some other multiple-try Metropolis algorithms \citep[see, e.g., the review][]{martino-review}. 

It is not hard to check that the i-SIR algorithm has the following transition probability:
\begin{equation}\label{def:isir-transition-probability}
  P_N(z_1,A) = \int \bigg( \prod_{n=2}^N q( z_n) \bigg) \sum_{i=1}^N \frac{w(z_i)}{\sum_{j=1}^N w(z_j)} \charfun{z_i\in A} \ud z_2 \cdots \ud z_N,
\end{equation}
for $N\ge 2$, and we define $P_1(z_1,A) = \delta_{z_1}(A)$.

The following `continuous relaxation' of the i-SIR algorithm plays a key role in our adaptive i-SIR algorithm described in Section \ref{section:adaptation}. It is defined for a `fractional number of proposals' $\lambda \ge 1$ (although later we will focus on $\lambda\ge 2$ only). Namely, let $\lambda\ge 1$, then
$$
  P_\lambda(x,A) = \beta(\lambda) P_{\lfloor \lambda \rfloor}(x,A) + \big(1-\beta(\lambda)\big) P_{\lfloor \lambda \rfloor+1}(x,A),
$$
with the interpolation function $\beta(\lambda) = 1-(\lambda-\lfloor \lambda \rfloor) \in (0,1]$.  Algorithm \ref{alg:continuous-isir} presents a particular way of sampling from $P_\lambda$, which will be useful later.

\begin{algorithm}
   \caption{$\textsc{iSIR}(X_{k-1}, \lambda)$}
   \label{alg:continuous-isir}
   \begin{algorithmic}[1]
      \State Draw independent $Y_k^{2:(\lfloor \lambda \rfloor+1)} \sim q$ and set $Y_k^1 = X_{k-1}$
      \State With probability $\lambda-\lfloor \lambda \rfloor$ set $N_k=\lfloor \lambda \rfloor+1$, else set $N_k=\lfloor \lambda \rfloor$,
       and calculate weights $W_k^i = w(Y_k^i)/\sum_{j=1}^{N_k} w(Y_k^j)$
      \State Sample $I_k \sim \mathrm{Categorical}(W_k^{1:N_k})$
      \State Return $X_k = Y_k^{I_k}$
   \end{algorithmic}
\end{algorithm}

For Algorithm \ref{alg:continuous-isir}, we shall denote the rejection probability at state $x$ by $\epsilon(\lambda,x)$, that is, the probability of $I_k=1$, when $X_{k-1}=x$.

\section{Properties of the i-SIR}
\label{sec:isir-properties}
It is well-known that the i-SIR algorithm $P_N$ is reversible with respect to $\pi$ \cite[e.g.][]{tjelmeland,andrieu-lee-vihola}, and therefore the generalised i-SIR $P_\lambda$, as a convex combination of $P_N$ and $P_{N+1}$ is $\pi$-reversible. We assume the following to hold throughout Sections \ref{sec:isir-properties}--\ref{section:adaptation}:

\begin{assumption}
  \label{a:bounded-weights}
  The weight function is bounded:
  $$
  \hat{w} = \sup_{x\in\mathbb{X}} w(x) < \infty.
  $$
\end{assumption}

This assumption requires, roughly speaking, the target density $\pi$ to have lighter tails than the proposal $q$, which can often be guaranteed by a suitable choice of $q$. It is known that i-SIR is slowly mixing (subgeometric) essentially if \ref{a:bounded-weights} does not hold, and if \ref{a:bounded-weights} holds, then the i-SIR is guaranteed to be well-behaved (uniformly ergodic) \citep[cf.][]{andrieu-lee-vihola}. In fact, for large $\lambda$, $P_\lambda$ corresponds asymptotically to independent sampling from $\pi$. Hereafter, we use the same symbol $\pi$ to denote the probability measure corresponding to the density $\pi$, that is, $\pi(A) = \int_A \pi(x) \ud x$.
\begin{proposition}
   \label{prop:isir-minorisation}
   The i-SIR transition probability $P_\lambda$ satisfies the following strong minorisation condition:
   $$
      P_\lambda(x,A) \ge \frac{\lambda - 1}{2\hat{w} + \lambda - 1} \pi(A).
   $$
\end{proposition}
This result is from \citep{andrieu-lee-vihola} for integer $\lambda$ and detailed in Lemma \ref{lemma:i-sir-probability-inequality} in Appendix \ref{app:ergodicity-auxiliaries} for non-integer $\lambda$.

Our main focus will be on the asymptotic variance of the i-SIR chain, defined for finite-variance functionals. Hereafter, we denote the integral of a (measurable) function $f:\mathbb{X}\to\R$ with respect to $\pi$ by $\pi(f) = \int f(x) \pi(x) \ud x$ and by $L^2(\pi)$ the set of all finite variance functionals with respect to $\pi$, that is, such that $\pi(f^2)  <\infty$, and also define $\var_\pi(f) = \pi(f^2) - \pi(f)^2$.

\begin{definition}
Fix $\lambda\in(1,\infty)$ and assume that $(X_k)_{k\ge 0}$ is the stationary i-SIR chain following $P_\lambda$, that is, $X_0 \sim \pi$ and $X_k\mid (X_{k-1}=x) \sim P_\lambda(x,\uarg)$. The asymptotic variance of $f\in L^2(\pi)$ is:
$$
   \var(P_\lambda,f) = \var\big(f(X_0)\big) + 2\sum_{k=1}^\infty \cov\big(f(X_0), f(X_k)\big).
$$
\end{definition}
Thanks to the uniform ergodicity of $P_\lambda$ ensured by Proposition \ref{prop:isir-minorisation},  $\var(P_\lambda,f)$ is always well-defined and finite. Moreover, a central limit theorem holds \citep[e.g.][Theorem 9]{jones-survey}:
\begin{theorem}
   \label{thm:clt}
   For $\lambda\in(1,\infty)$, let $X_k$ be the i-SIR chain with an arbitrary initial state $X_0=x$. Then, for any $f\in L^2(\pi)$:
$$
\frac{1}{\sqrt{n}} \sum_{k=1}^n \big[f(X_k) - \pi(f)\big] \xrightarrow{d} N\big(0, \var(P_\lambda, f)\big), \qquad \text{as $n\to \infty$}.
$$
\end{theorem}

Theorem \ref{thm:clt} shows that $\var(P_\lambda, f)$ quantifies the statistical (per sample) efficiency of estimating $\pi(f)$ with a Markov chain following $P_\lambda$. While we do not know $\var(P_\lambda, f)$ in general, we know that it has certain properties:
\begin{theorem}
   \label{thm:isir-asvar-properties}
   Let $\lambda\in(1,\infty)$ and $f\in L^2(\pi)$ such that $\var_\pi(f)>0$. Then,
   \begin{enumerate}[(i)]
      \item \label{item:asvar-upper-bound} $\var_\pi(f) \le \var(P_\lambda, f) \le \frac{4 \hat{w} + \lambda - 1}{\lambda - 1} \var_\pi(f)$.
      \item \label{item:asvar-convexity} $\lambda\mapsto \var(P_\lambda, f)$ is strictly convex and decreasing on $(1,\infty)$.
   \end{enumerate}
\end{theorem}
Theorem \ref{thm:isir-asvar-properties} \eqref{item:asvar-upper-bound} is similar to the bound in \citep{andrieu-lee-vihola}, and follows from Proposition \ref{prop:isir-minorisation}; see Lemma \ref{lemma:lower-and-upper-bound-for-asymptotic-variance} in Appendix \ref{app:upper-lower-bounds}. Theorem \ref{thm:isir-asvar-properties} \eqref{item:asvar-convexity} is to our knowledge new and given in Theorem \ref{theorem:i-sir-asymptotic-variance}.

\section{Overall efficiency and approximation based on average acceptance rate}
\label{sec:efficiency}

The asymptotic variance measures the statistical efficiency of $P_\lambda$, but because sampling $P_\lambda$ involves simulation of $O(\lambda)$ samples and evaluation of the weight function at each simulated sample, the number of computer operations is $O(\lambda)$. This affects the overall efficiency of $P_\lambda$. However, thanks to the independence of the samples, the sampling and evaluation of the weight can be performed in parallel, so number of computer operations might not be directly useful.

We will assume that $c(\lambda)$ is a function that defines the `cost' of (time spent on) one iteration of $P_\lambda$. In particular, we will focus on a simplified cost of the form $c(\lambda) = a + b\lambda$, which assumes an initial (`parallelisation overhead') cost $a>0$ and a term which increases in terms of $\lambda$ (`parallel computing cost'), with rate $b>0$; see \citep{gustafson} for other ways to characterise parallel computing speedups.

When we take into account the $\lambda$-dependent cost $c(\lambda)$, we would ultimately want to minimise the so-called inverse relative efficiency $L(\lambda) = c(\lambda) \var(P_\lambda, f)$ in terms of a test function $f$ of interest \citep[cf.][]{glynn-whitt}. However, this is difficult in practice, for two reasons. The asymptotic variance $\var(P_\lambda, f)$ is unknown, and notoriously difficult to estimate. Furthermore, we often have multiple test functions $f_1,\ldots,f_m$ which we want to estimate, which will generally have different asymptotic variances.

The upper bound of the asymptotic variance \eqref{item:asvar-upper-bound} might be a possible proxy to optimise. However, while the \emph{existence} of $\hat{w}<\infty$ in \ref{a:bounded-weights} might be known, its value is rarely known. It is also unclear how well the upper bound matches the actual asymptotic variance. (We present some numerical comparisons in Section \ref{sec:experiments} regarding this.)

Instead of the asymptotic variance upper bound, we consider the following approximation of the i-SIR transition inspired by 
Proposition \ref{prop:isir-minorisation}:

\begin{definition}
   \label{def:approximate-isir}
   The approximate i-SIR transition $\tilde{P}_\lambda$ of $P_\lambda$ is defined as
   $$
   \tilde{P}_\lambda(x,A) = \big(1-\epsilon(\lambda)\big) \pi(A) + \epsilon(\lambda) \delta_x(A),
   $$
   where the average rejection probability $\epsilon(\lambda)$ is
   $$ 
   \epsilon(\lambda) = \int \epsilon(\lambda,x)\pi(x) \ud x.
   $$
\end{definition}

In words, $\tilde{P}_\lambda$ corresponds to independent sampling with probability $1-\epsilon(\lambda)$ and `rejection' with probability $\epsilon(\lambda)$. Proposition \ref{prop:isir-minorisation} guarantees that we can write $P_\lambda(x,A) = (1-\delta(\lambda))\pi(A) + \delta(\lambda) Q_\lambda(x,A)$, where $\delta(\lambda) \to 0$ as $\lambda\to \infty$, and the `residual' kernel $Q_\lambda(x,A)$ is $\pi$-reversible. This suggests that at least for $\lambda$ large, $\tilde{P}_\lambda \approx P_\lambda$. 

The approximate transition probability $\tilde{P}_\lambda$ is simple enough to admit an explicit form of asymptotic variance: 
\begin{proposition}
\label{proposition:approximate-i-sir-asymptotic-variance}
Let $\lambda\in(1,\infty)$ and $f\in L^2(\pi)$ such that $\var_\pi(f)>0$. Then,
\begin{enumerate}[(i)]
\item \label{item:approx-asvar} $\displaystyle
\var(\tilde{P}_\lambda, f) = \frac{1+\epsilon(\lambda)}{1-\epsilon(\lambda)} \var_\pi(f).
$
\item \label{item:approx-bounds} $\var_\pi(f) \le \var(\tilde{P}_\lambda, f) \le \frac{4 \hat{w} + \lambda - 1}{\lambda - 1} \var_\pi(f)$.
\item \label{item:approx-decreasing-and-convex} $\lambda \mapsto \var(\tilde{P}_\lambda, f)$ is strictly convex and decreasing on $(1,\infty)$.
\end{enumerate}
\end{proposition}
Proposition \ref{proposition:approximate-i-sir-asymptotic-variance} \eqref{item:approx-asvar} follows from Lemma \ref{lemma:single-proposal-chain-asymptotic-variance} in Appendix \ref{app:approximate-properties}, 
\eqref{item:approx-bounds} from Lemma \ref{lemma:lower-and-upper-bound-for-rejection}
in Appendix \ref{app:upper-lower-bounds} and \eqref{item:approx-decreasing-and-convex} 
from Lemma \ref{lemma:convexity-of-the-estimate} in Appendix \ref{app:approximate-properties}. Proposition \ref{proposition:approximate-i-sir-asymptotic-variance} indicates that $\var(\tilde{P}_\lambda, f)$ has similar properties to $\var(P_\lambda, f)$ (cf.~Theorem \ref{thm:isir-asvar-properties}). Our method for finding an appropriate `number of proposals' $\lambda$ is based on this expression: we aim for approximately minimising $c(\lambda) \var(\tilde{P}_\lambda, f)$, which we discuss next.

\section{Adaptation of the number of proposals}\label{section:adaptation}

Our method for finding an appropriate $\lambda$ is a stochastic gradient type algorithm. To simplify discussion in this section, we focus on a continuous case (e.g. $\X \subset\R^d$), and assume the following to hold:

\begin{assumption}
   \label{a:non-atomic}
The proposal distribution $q$ is non-atomic, that is, if $Y\sim q$, then $\P(Y=c)=0$ for all $c\in\mathbb{X}$.
\end{assumption}

The i-SIR is valid also when $q$ has atoms, like when $\X$ is discrete; we include discussion about this case in Sections \ref{section:inner-product-convex-sequence}--\ref{sec:approximate}.

Under \ref{a:non-atomic}, the i-SIR chain has $X_{k} = X_{k-1}$ exactly when $I_k=1$, and therefore, the average holding (or rejection) probability of the chain is:
\begin{equation*}
\epsilon(\lambda) = \int \pi( z_1) \bigg( \prod_{i = 2}^{\lfloor \lambda \rfloor+1} q( z_i)\bigg) w(z_1) \bigg( 
   \frac{\beta(\lambda)}{\sum_{j=1}^{\lfloor \lambda \rfloor} w(z_j)} 
   + \frac{1-\beta(\lambda)}{\sum_{j=1}^{\lfloor \lambda \rfloor+1} w(z_j)} \bigg) \ud z_1\cdots \ud z_{\lfloor \lambda \rfloor+1}.
\end{equation*}
Clearly, the function $\epsilon$ is differentiable at non-integer points $\lambda$:
\begin{align}\label{eq:intro-derivative}
    \epsilon'(\lambda)
    &= \int \pi( z_1) \bigg( \prod_{i= 2}^{\lfloor \lambda \rfloor+1} q( z_i)\bigg) w(z_1)
    \bigg(
      \frac{1}{\sum_{j=1}^{\lfloor \lambda \rfloor+1} w(z_j)}
      - \frac{1}{\sum_{j=1}^{\lfloor \lambda \rfloor} w(z_j)}
    \bigg) \ud z_1\cdots \ud z_{\lfloor \lambda \rfloor+1},
\end{align}
and we extend $\epsilon'$ to all $\lambda\ge 1$ by taking the c\`adl\`ag extension
of the right-hand side of \eqref{eq:intro-derivative}. 
If we iterated Algorithm \ref{alg:continuous-isir} (with fixed $\lambda$), we could estimate $\epsilon(\lambda)$ and $\epsilon'(\lambda)$ by averaging the following quantities: 
\begin{align}
  \hat{\epsilon}_k(\lambda) &= 
  w(Y_{k}^1) \bigg( 
   \frac{\beta(\lambda)}{\sum_{j=1}^{\lfloor \lambda \rfloor} w(Y_k^j)}
   + \frac{1-\beta(\lambda)}{\sum_{j=1}^{\lfloor \lambda \rfloor+1} w(Y_k^j)} \bigg),
   \label{eq:epsilon-est} \\
  \hat{\epsilon}'_k(\lambda) & = 
  w(Y_k^1)
    \bigg(
      \frac{1}{\sum_{j=1}^{\lfloor \lambda \rfloor+1} w(Y_k^j)}
      - \frac{1}{\sum_{j=1}^{\lfloor \lambda \rfloor} w(Y_k^j)}
    \bigg).
   \label{eq:depsilon-est}
\end{align}
Under stationarity, that is, when $X_{k-1}\sim \pi$, it is easy to see that the expected values of $\hat{\epsilon}_k(\lambda)$ and $\hat{\epsilon}'_k(\lambda)$ are indeed $\epsilon(\lambda)$ and $\epsilon'(\lambda)$, respectively. 
Ignoring the multiplicative constant $\mathrm{var}_\pi(f)$ in $\var(\tilde{P}_\lambda,f)$, which does not depend on $\lambda$, we consider the following loss function:
\begin{equation}
\tilde{L}(\lambda) = 
   \frac{1+\epsilon(\lambda)}{1-\epsilon(\lambda)} c(\lambda), 
   \label{eq:approximate-loss}
\end{equation}
which is continuous and differentiable (at non-integer $\lambda$):
\begin{align*}
    \tilde{L}'(\lambda)
   &= c'(\lambda) \frac{1 + \epsilon(\lambda)}{1 - \epsilon(\lambda)} + c(\lambda) \frac{\epsilon'(\lambda)\big(1-\epsilon(\lambda)\big) + \epsilon'(\lambda)\big(1+\epsilon(\lambda)\big)}{\big(1-\epsilon(\lambda)\big)^2} \\ 
   & = \frac{c'(\lambda) (1 - \epsilon^2(\lambda)) + 2 c(\lambda)\epsilon'(\lambda)}{\big(1-\epsilon(\lambda)\big)^2}.
\end{align*}
For a gradient-like minimisation algorithm, we may ignore the positive denominator, which scales the gradient. 
Furthermore, we employ the approximation $\epsilon^2(\lambda)\approx \epsilon_s(\lambda)$, where $\epsilon_s(\lambda)$ is the stationary expectation of the squares $\hat{\epsilon}^2_k(\lambda)$ of \eqref{eq:epsilon-est}; see the definition in \eqref{eq:ergodic-averages-jensen-2} in Appendix \ref{app:upper-lower-bounds}. This approximation slightly deflates $\tilde{L}'(\lambda)$, because by Jensen's inequality $\epsilon^2(\lambda) \le \epsilon_s(\lambda)$.

We end up with a rescaled approximate gradient of the form
$$
G(\lambda) = c'(\lambda) (1 - \epsilon_s(\lambda)) + 2 c(\lambda) \epsilon'(\lambda).
$$
Furthermore, when we adapt $\lambda$, we use a monotone reparametrisation $\lambda = 1 + e^{\xi}$, where $\xi\in\R$ is the parameter we adapt.

One iteration of the algorithm, that is, update $(\xi_{k-1},X_{k-1})$ to $(\xi_{k},X_{k})$ is summarised in Algorithm \ref{alg:adaptive-isir}.
\begin{algorithm}
   \caption{$\textsc{iSIRAdapt}(X_{k-1}, \xi_{k-1})$}
   \label{alg:adaptive-isir}
\begin{algorithmic}[1]
   \State Let $\lambda_{k-1} = e^{\xi_{k-1}} + 1$
   \State \label{line:isir-start} Let $Y_{k}^1 = X_{k-1}$ and simulate $Y_{k}^{2:(\floor{\lambda_{k-1}}+1)} \overset{\text{i.i.d.}}{\sim} q$
   \State With probability $\lambda_{k-1}-\floor{\lambda_{k-1}}$ set $N_k=\floor{\lambda_{k-1}}+1$, else $N_k=\floor{\lambda_{k-1}}$
   \State Draw $I_k \sim \mathrm{Categ}(w(Y_{k}^1),\ldots,w(Y_k^{N_k}))$
   \State \label{line:isir-end}Set $X_{k} = Y_{k}^{I_k}$
   \State Calculate $\hat{\epsilon}_{k}(\lambda_{k-1})$ and $\hat{\epsilon}'_{k}(\lambda_{k-1})$ as in \eqref{eq:epsilon-est} and \eqref{eq:depsilon-est}, respectively
   \State \label{line:adapt} Calculate $\xi_k = \mathrm{Proj}\big\{\xi_{k-1} - \gamma_k 
   \big[ c'(\lambda_{k-1})\big(1-\hat{\epsilon}_k(\lambda_{k-1})^2\big)
   + 2 c(\lambda_{k-1}) \hat{\epsilon}'_k(\lambda_{k-1}) \big]\big\}$
\end{algorithmic}
\end{algorithm}
Lines \ref{line:isir-start}--\ref{line:isir-end} correspond to Algorithm \ref{alg:continuous-isir} with $\lambda_{k-1}$.
The adaptation in Line \ref{line:adapt} uses random variables whose stationary expectations coincide with $G(\lambda_{k-1})$. The algorithm is an instance of stochastic approximation type adaptive MCMC \citep[cf.][]{andrieu-robert,andrieu-moulines,andrieu-thoms} where $\gamma_k$ is a non-random step size sequence, which we choose as $k^{-\beta}$ where $\beta = 0.75$; values $\beta\in(1/2,1)$ are advisable \citep[e.g.][]{andrieu-moulines}. The additional projection function $\mathrm{Proj}\{\xi\} = \min\big\{\max\{0,\xi\}, \log(N_{\max}-1)\big\}$ ensures that $2 \le \lambda_k \le N_{\max} $ when $N_{\max}$ is finite, and $2 \le \lambda_k<\infty$ when $N_{\max}=\infty$. A finite upper bound is not necessary for the validity, but it can simplify the practical implementation of the method. The initial value of $\xi_0$ is arbitrary, and can correspond for instance to $\lambda_0 = N_{\max}/2$.

\begin{theorem}
   \label{thm:adaptive-isir-slln}
Suppose that $\gamma_k = k^{-\beta}$ where $\beta>0$. Suppose the cost function $c$ is continuous and differentiable, and in the case $N_{\max}=\infty$, additionally $c(\lambda)=O(\lambda)$ and $c'$ is bounded.
Then, the variables $X_1,X_2,\ldots$ generated by iteration of Algorithm \ref{alg:adaptive-isir} satisfy the following strong law of large numbers for any bounded (measurable) test function $f$:
$$
\frac{1}{n}\sum_{k=1}^n f(X_k) \xrightarrow{n\to\infty} \pi(f).
$$
\end{theorem}
\begin{proof}
When $N_{\max}$ is finite, we set $A=[0,\log(N_{\max}-1)]$, and when $N_{\max}=\infty$, we set $A=[0,\infty)$. Let us denote
$$H_k=-\big[ c'(\lambda_{k-1})\big(1-\hat{\epsilon}_k(\lambda_{k-1})^2\big) + 2 c(\lambda_{k-1}) \hat{\epsilon}'_k(\lambda_{k-1}) \big].$$
By definition,
 \begin{equation*}
\mathrm{Proj}(\xi_{k-1}+\gamma_k H_k) 
=\begin{cases}
      \xi_{k-1} + \gamma_k H_k, & \xi_{k-1} + \gamma_k H_k \in A, \\
      0, & \xi_{k-1} + \gamma_k H_k<0 \\
      \log(N_{\max}-1), & \xi_{k-1} + \gamma_k H_k>\log(N_{\max}-1).
   \end{cases}
\end{equation*}
Let $\tilde{H}_k=\frac{\mathrm{Proj}(\xi_{k-1}+\gamma_k H_k)-\xi_{k-1}}{\gamma_k}$, then we obtain
$$\xi_k=\mathrm{Proj}(\xi_{k-1}+\gamma_k H_k)=\xi_{k-1}+\gamma_k \tilde{H}_k,$$
$|\xi_k-\xi_{k-1}|= \gamma_k |\tilde{H}_k|$ and clearly $\gamma_k|\tilde{H}_k| \le \gamma_k |H_k|$.
We have $|1-\hat{\epsilon}_k(\lambda_{k-1})^2|\le 1$. Now $\lambda_{k-1}=e^{\xi_{k-1}}+1$, let us denote $N(x)=\floor{e^{x}+1}$, then
\begin{align*}
\E\big[|\hat{\epsilon}_k'(\lambda_{k-1})| \;\big|\; X_{k-1},\xi_{k-1}\big]
&=\E\Bigg[\bigg|\frac{w(Y_{k}^1)}{\sum_{j=1}^{N(\xi_{k-1})+1} w(Y_k^j)}-\frac{w(Y_{k}^1)}{\sum_{j=1}^{N(\xi_{k-1})} w(Y_k^j)}\bigg| \;\Bigg|\; X_{k-1},\xi_{k-1}\Bigg] \\
&=\E\bigg[\frac{w(Y_{k}^1)w(Y_{k}^{N(\xi_{k-1})+1})}{\big(\sum_{i=1}^{N(\xi_{k-1})} w(Y_k^i)\big)\big(\sum_{j=1}^{N(\xi_{k-1})+1} w(Y_k^j)\big)} \;\bigg|\;X_{k-1},\xi_{k-1}\bigg] \\
&\le \E\bigg[\frac{w(Y_{k}^{N(\xi_{k-1})+1})}{\sum_{j=1}^{N(\xi_{k-1})+1} w(Y_k^j)} \;\bigg|\;X_{k-1},\xi_{k-1}\bigg] \\
&=\E\bigg[\frac{1}{N(\xi_{k-1})}\frac{\sum_{i=2}^{N(\xi_{k-1})+1} w(Y_k^i)}{\sum_{j=1}^{N(\xi_{k-1})+1} w(Y_k^j)} \;\bigg|\;X_{k-1},\xi_{k-1}\bigg] \\
&\le \frac{1}{N(\xi_{k-1})},
\end{align*}
where the last equality follows from the fact that $Y_2,\dots,Y_N$ are exchangeable for fixed $N$. Therefore,
\begin{align*}
\E\big[|c(\lambda_{k-1})\hat{\epsilon}_k'(\lambda_{k-1})|\;\big|\; X_{k-1},\xi_{k-1}\big]\le\frac{|c(e^{\xi_{k-1}}+1)|}{\floor{e^{\xi_{k-1}}+1}}\le \frac{|c(e^{\xi_{k-1}}+1)|}{e^{\xi_{k-1}}}.
\end{align*}
When $N_{\max}$ is finite, functions $x \mapsto \frac{|c(e^{x}+1)|}{e^{x}}$ and $x\mapsto c'(e^x+1)$ are bounded on $A$. When $N_{\max}=\infty$ and the additional assumptions hold, then $x\mapsto \frac{|c(e^{x}+1)|}{e^{x}}$ is continuous and $\frac{|c(e^{x}+1)|}{e^{x}}=O(1)$, hence bounded on $A$.
Thus, in both cases
\begin{align*}
\E\big[|H_k|\big]
\le \E\big[|c'(\lambda_{k-1})|(1-\hat{\epsilon}_k(\lambda_{k-1})^2)\big]
+2\E\big[\E\big[|c(\lambda_{k-1})\hat{\epsilon}'_k(\lambda_{k-1})| \;\big|\;  X_{k-1},\xi_{k-1}\big]\big]\le M,
\end{align*}
for some $M>0$, hence
\begin{align*}
\sum_{k=1}^\infty \frac{\gamma_k}{k} \E\big[|\tilde{H}_k|\big]\le \sum_{k=1}^\infty \frac{\gamma_k}{k} \E\big[|H_k|\big]=\sum_{k=1}^\infty k^{-(\beta+1)} \E\big[|H_k|\big]\le M\sum_{k=1}^\infty k^{-(\beta+1)}<\infty.
\end{align*}
Let $\xi,\xi'\in A$, then by Lemma \ref{lemma:auxiliary-1}
\begin{align*}
d_\tv\big(P_{e^{\xi'}+1}(x,\uarg), P_{e^{\xi}+1}(x,\uarg)\big)& \le \frac{6|e^{\xi'}-e^{\xi}|}{\max\braces{e^{\xi'}+1,e^{\xi}+1}} \le 6|1-e^{-|\xi'-\xi|}| \le 6|\xi'-\xi|,
\end{align*}
where the last inequality follows since $x\mapsto x-(1-e^{-x})$ is increasing for $x\ge 0$ and therefore non-negative.
Therefore, by combining all the previous and Lemma \ref{lemma:auxiliary-2} in Appendix \ref{app:ergodicity-auxiliaries} with \cite[Lemma 5.7 $(ii)$, Theorem 4.4 $(ii)$]{laitinen-vihola-adaptive}, we obtain strong law of large numbers.
\end{proof}

Theorem \ref{thm:adaptive-isir-slln} ensures that the samples simulated by Algorithm \ref{alg:adaptive-isir} lead to consistent inference with respect to the target distribution $\pi$. However, it does not guarantee the convergence of the adaptation, that is, whether there exists a limit $\lambda_*$ such that $\lambda_k\to\lambda_*$. Precise theoretical results about this convergence are left for future work. The empirical observations in Section \ref{sec:experiments} suggest that such convergence can happen, and that $\lambda_*$ appears to indeed (approximately) minimise the loss in \eqref{eq:approximate-loss} (and also approximately minimises the inverse relative efficiency). We conclude this section with a discussion of some partial results related to the convergence analysis.

Here, we consider a gradient-descent-like stochastic approximation for an unbounded interval $\lambda\in(1,\infty)$.
If it converges, the limit point is the zero of the `gradient-like' function $G(\lambda)$.

\begin{proposition}\label{proposition:between-positive-negative}
Let $c(\lambda)=a+b\lambda$, $a\ge 0$, $b>0$ and $h_1,h_2,h_3:[1,\infty)\to \R$,
\begin{align*}
h_1(\lambda)&=c'(\lambda)(1-\lambda^{-2}) + 2c(\lambda) \epsilon'(\lambda), \\
h_2(\lambda)&=c'(\lambda)(1-\epsilon(\lambda)^2) + 2c(\lambda) \epsilon'(\lambda), \\
h_3(\lambda)&=c'(\lambda)(1-\epsilon_s(\lambda)) + 2c(\lambda) \epsilon'(\lambda).
\end{align*}
Then, $h_3\le h_2 \le h_1$, $\liminf_{\lambda \to \infty} h_i(\lambda) > 0$ and $\limsup_{\lambda\to 1+} h_i(\lambda) < 0$ for all $i$.
\end{proposition}
Proposition \ref{proposition:between-positive-negative} follows from properties listed in Theorem \ref{theorem:i-sir-convexity-2} and Lemmas \ref{lemma:lower-and-upper-bound-for-rejection}, \ref{lemma:ergodic-averages-jensen} and \ref{lemma:continuous-isir-rejection-derivative} in Appendix \ref{app:upper-lower-bounds}.
If we consider an affine cost function $c(\lambda) = a + b\lambda$, then Proposition \ref{proposition:between-positive-negative} implies that if we used $G=h_3$ as the gradient like function, the gradient descent is expected to remain bounded, and therefore the stochastic gradient descent is expected to converge to a stationary point within $(1,\infty)$. Furthermore, the minimiser of $\tilde{L}(\lambda)$ corresponds to the zero of $h_2$, which has similar behaviour. We conclude this section by a (crude) upper bound for the minimiser of true asymptotic variance:
\begin{remark}
For a cost function $c(\lambda) = a + b\lambda$ where $a\ge 0$ and $b>0$ and a test function $f\in L^2(\pi)$, $\mathrm{var}_\pi(f)>0$, the minimiser $\lambda^*$ of $c(\lambda)\mathrm{var}(P_\lambda,f)$ satisfies $2\le \lambda^*\le 4\sqrt{\hat{w}(a/b+1)}+4\hat{w}+1$ by Theorem \ref{thm:isir-asvar-properties} and Propositions \ref{proposition:lower-and-upper-bound-cost-function} and \ref{propostion:cost-function-with-asymptotic-variance-1} in Appendix \ref{app:cost-function}.
\end{remark}

\section{Convexity and monotonicity of i-SIR covariances and acceptance rates} \label{section:inner-product-convex-sequence}

Next, we discuss some technical properties of the i-SIR (Algorithm \ref{alg:isir}) and its continuous relaxation (Algorithm \ref{alg:continuous-isir}). Namely, Sections \ref{subsection:inner-product-convex-sequence} and \ref{subsection:rejection-probabilities-convex-sequence} establish properties for asymptotic covariances and rejection probability for the i-SIR, and Section \ref{section-convexity-monotonicity-of-isir} extends these properties to the continuous relaxation. In particular, the results are necessary for establishing convexity of asymptotic variance in Section \ref{sec:asymptotic-variance}. 
Recall the notion of convexity of a sequence, or sequential convexity \cite[cf.][]{mitrinovic}:

\begin{definition}
A real sequence $(a_i)$ is called convex if $2a_{i+1}\le a_{i}+a_{i+2}$ for all $i$.
\end{definition}

Sequential convexity is convenient because it ensures convexity for the interpolations of the sequence; see Lemma \ref{lemma:generalisation-to-continuous-function} \eqref{enum:generalisation-1} in Appendix \ref{app:modification}.

We will now give the definitions we require.
We denote the natural numbers as $\N=\braces{0,1,2,\dots}$ and positive natural numbers as $\N_+=\braces{1,2,\dots}$. Throughout the upcoming sections, we have fixed $\pi$ and $q$ as mentioned in Section \ref{sec:introduction}  (more detailed in Appendix \ref{app:setup}). 
We use the shorthand notation $\pi(\ud x)=\pi(x)\mu(\ud x)$ and $q(\ud x)=q(x)\mu(\ud x)$, where $\mu$ is the dominating measure.
When we speak of atoms, we mean atoms with respect to $\pi$.
Recall that the support of $\pi$ is denoted by $\Sp$, then we denote the following subsets: set of atoms in the support by $\Sp_+=\braces{x\in \Sp:\pi(\braces{x})>0}$ and the set of non-atoms by $\Sp_0=\braces{x\in \Sp:\pi(\braces{x})=0}$.
We define the inner product 
$\angles{g|h}_\pi=\int g(x)h(x)\pi(\ud x)$ for $g,h\in L^2(\pi)$, the corresponding norm $\|f\|_\pi=\sqrt{\angles{f|f}_\pi}$, and the function spaces
\begin{align*}
L_0^2(\pi)=\braces{g\in L^2(\pi):\E_\pi[g]=0}, \qquad L_{0,+}^2(\pi)=\braces{g\in L_0^2(\pi):\|g\|_\pi>0}.
\end{align*}
The centred version of a $\pi$-integrable function is denoted as $\bar{f}=f-\E_\pi[f]$. 
The counting measure, that is, the number of elements in a set, is denoted by $\#$.

In our proofs, (permutation) symmetricity of functions plays a central role, because of exchangeability of the proposal variables in Algorithm \ref{alg:isir}. 

\begin{definition}\label{def:symmetricity}
Let $n\in \N_+$, $J\subset \braces{1,\dots,n}$, $U\subset \X^n$ and $f:U \to \R$. The function $f$ is called symmetric with respect to the arguments $J$ if for any permutation $\sigma:J\to J$, we have
$f(x)=f(y)$, where
$$y_{i}=
\begin{cases}
x_i, & i\notin J \\
x_{\sigma(i)}, &i\in J,
\end{cases}
\qquad \text{for }i\in \braces{1,\dots,n}.
$$
\end{definition}

Let $N\in \N_+$, then in particular the mapping $(z_1,\ldots,z_N)\mapsto \sum_{i=1}^Nw(z_i)$ is symmetric. We use the following notation for the set of states where this mapping is positive:
$$\W{N}
=\biggbraces{z\in \X^N\;:\;\sum_{i=1}^Nw(z_i)>0} = \bigcup_{i=1}^N \biggbraces{z\in \X^N\;:\;w(z_i)>0}.
$$
In what follows, we use the following convention:
$$\int_{\X^{N-1}}f(x) \bigg(\prod_{n=2}^N q(\ud z_n)\bigg)=f(x),$$
when $N=1$.
When we use a test set $A$ we always mean a measurable set.
The transition probability which we obtain from Algorithm \ref{alg:isir} when $N\in \N_+$ will be denoted by $P_N$ throughout the upcoming sections, and is defined in \eqref{def:isir-transition-probability},
\begin{equation*}
P_N(z_1,\ud y) =\int_{\X^{N-1}}  \sum_{i=1}^N \frac{w(z_i)}{\sum_{j=1}^N w(z_j)}\delta_{z_i}(\ud y)\bigg( \prod_{n=2}^N q(\ud z_n) \bigg)
\end{equation*}
with $P_1(z_1,\ud y)= \delta_{z_1}(\ud y)$.

\subsection{Stationary covariances}\label{subsection:inner-product-convex-sequence}

If $(X_k)_{k\ge 1}$ is the stationary i-SIR chain, then the expressions $\E[f(X_k)f(X_{k+1})]$ may be written in terms of inner products $\angles{f|P_Nf}_\pi$. In particular, when $\pi(f)=0$, this expression coincides with $\mathrm{Cov}\big(f(X_k), f(X_{k+1})\big)$, so we will call these inner products `covariances' hereafter.

For $f\in L^2(\pi)$ and $N\in \N_+$, we have $P_Nf\in L^2(\pi)$ and therefore the inner products are well-defined and finite. They can be written as
\begin{align}\label{equation:isir-inner-product}
\begin{split}
\angles{f|P_Nf}_\pi 
&=\int_{\Sp}\int_{\X^{N-1}} \sum_{i=1}^N \frac{w(z_i)f(z_1)f(z_i)}{\sum_{j=1}^N w(z_j)} \bigg( \prod_{n=2}^N q(\ud z_n) \bigg)\pi(\ud z_1) \\
&= \int_{\W{N}} \sum_{i=1}^N \frac{w(z_1)w(z_i)f(z_1)f(z_i)}{\sum_{j=1}^N w(z_j)} \bigg( \prod_{n=1}^N q(\ud z_n) \bigg),
\end{split}
\end{align}
because $w(z_1)q(\ud z_1) = \pi(\ud z_1)$. Note also that $\angles{f|P_1f}_\pi=\angles{f|f}_\pi$. The following result shows that the integral above can be written in a form involving a quadratic term, which turns out to be central for the properties we establish. In particular, it directly implies the positivity of the covariances in \eqref{equation:isir-inner-product}.

\begin{lemma}\label{lemma:inner-product-another-representation}
Let $N\in \N_+$ and $f\in L^2(\pi)$, then
$$
 \angles{f|P_Nf}_\pi=\frac{1}{N}\int_{\W{N}}\frac{\Big(\sum_{i=1}^Nw(z_i)f(z_i)\Big)^2}{\sum_{j=1}^N w(z_j)} \bigg( \prod_{n=1}^N q(\ud z_n) \bigg)\ge 0.
$$
Additionally, $\angles{f|P_Nf}_\pi=0$ if and only if $\|f\|_\pi=0$. 
\end{lemma}
\begin{proof}
The first claim follows from Lemma \ref{lemma:i-sir-exchangeable-property} in Appendix \ref{app:monotonicity-convexity}, take any $M\in \N_+$ and $h=1$. If $\|f\|_\pi=0$, then $f=0$ $\pi$-almost surely, hence \eqref{equation:isir-inner-product} implies $\angles{f|P_Nf}_\pi=0$.
If $\|f\|_\pi\neq 0$, then $f\neq 0$ in a $\pi$-positive set and thus at least $A=\braces{x\in \Sp:f(x)>0}$ or $B=\braces{x\in \Sp:f(x)<0}$ has to be $\pi$-positive, which is also $q$-positive. It holds for all $x\in \Sp$ that $w(x)>0$, hence if $A$ is $\pi$-positive, it follows that
$$
\angles{f|P_Nf}_\pi\ge \frac{1}{N}\int_{A^N}\frac{\Big(\sum_{i=1}^Nw(z_i)f(z_i)\Big)^2}{\sum_{j=1}^N w(z_j)} \bigg( \prod_{n=1}^N q(\ud z_n) \bigg)>0,
$$
and if $B$ is $\pi$-positive, we get similar inequality.
\end{proof}

Using the quadratic representation above, and symmetricity arguments, allows us to conclude that the covariance decreases with respect to $N$:
\begin{lemma}\label{lemma:i-sir-inner-product-differences}
Let $N\in \N_+$ and $f\in L^2(\pi)$, then
\begin{equation*}
\angles{f|P_Nf}_\pi \ge \angles{f|P_{N+1}f}_\pi,
\end{equation*}
where the inequality is strict if and only if $\bar{f}\in L_{0,+}^2(\pi)$.
\end{lemma}
\begin{proof}
If $\bar{f}\notin L_{0,+}^2(\pi)$, then there exists $c\in \R$ such that $f=c$ $\pi$-almost surely, hence $\angles{f|P_Nf}_\pi=c^2$ for all $N\in \N_+$.
Combining this with Lemma \ref{lemma:i-sir-inequality} in Appendix \ref{app:monotonicity-convexity} with $h=1$, we obtain the claim.
\end{proof}

The next result, which shows that the covariances form a convex sequence with respect to $N$, is based on a similar but slightly more involved arguments:

\begin{lemma}\label{lemma:i-sir-inner-product-differences-3}
Let $N\in \N_+$ and $f\in L^2(\pi)$, then
\begin{equation*}
\angles{f|P_{N}f}_\pi-2\angles{f|P_{N+1}f}_\pi+\angles{f|P_{N+2}f}_\pi\ge 0,
\end{equation*}
where the inequality is strict if and only if $\bar{f}\in L_{0,+}^2(\pi)$.
\end{lemma}
\begin{proof}
We first show that
\begin{align}
&\angles{f|P_{N}f}_\pi-2\angles{f|P_{N+1}f}_\pi+\angles{f|P_{N+2}f}_\pi \nonumber \\
&\quad=2\int_{\W{N}\times \X^{2}} 
\frac{w(z_{N+2})}{\sum_{j=1}^{N+2} w(z_j)}\frac{\sum_{i=1}^{N}w(z_1)f(z_1)w(z_i)f(z_i)}{\sum_{k=1}^{N} w(z_k)}\bigg( \prod_{n=1}^{N+2} q(\ud z_n) \bigg) \label{eq:cx-expression} \\
&\qquad-2\int_{\W{N+1}\times \X} 
\frac{w(z_{N+2})}{\sum_{j=1}^{N+2} w(z_j)}\frac{\sum_{i=1}^{N+1}w(z_1)f(z_1)w(z_i)f(z_i)}{\sum_{k=1}^{N+1} w(z_k)}\bigg( \prod_{n=1}^{N+2} q(\ud z_n) \bigg).\nonumber
\end{align}
Let us denote $A_n=\sum_{i=1}^nw(z_1)w(z_i)f(z_1)f(z_i)$ and $B_n=\sum_{i=1}^nw(z_i)$. We have
\begin{align*}
&\angles{f|P_{N}f}_\pi-2\angles{f|P_{N+1}f}_\pi+\angles{f|P_{N+2}f}_\pi \\
&\quad=\int_{\Sp\times \X^{N+1}}\Big(\frac{A_{N}}{B_{N}}-2\frac{A_{N+1}}{B_{N+1}}+\frac{A_{N+2}}{B_{N+2}}\Big)\bigg( \prod_{n=1}^{N+2}q(\ud z_n) \bigg) \\
&\quad=\int_{\Sp\times \X^{N+1}} \frac{B_{N+2}B_{N+1}A_{N}-2B_{N+2}B_{N}A_{N+1}+B_{N+1}B_{N}A_{N+2}}{B_{N+2}B_{N+1}B_{N}} \bigg( \prod_{n=1}^{N+2} q(\ud z_n) \bigg).
\end{align*}
By Lemma \ref{lemma:differences} in Appendix \ref{app:differences},
\begin{align*}
&B_{N+2}B_{N+1}A_{N}-2B_{N+2}B_{N}A_{N+1}+B_{N+1}B_{N}A_{N+2} \\
&\quad=2w(z_{N+2})B_{N+1}A_{N}-2w(z_{N+2})B_{N}A_{N+1}\\
&\qquad-B_{N+1}B_{N}w(z_1)f(z_1)[w(z_{N+1})f(z_{N+1})-w(z_{N+2})f(z_{N+2})] \\
&\qquad+B_{N+1}A_{N}[w(z_{N+1})-w(z_{N+2})].
\end{align*}
The claim follows, since
\begin{align*}
&\int_{\Sp\times\X^{N+1}} \frac{w(z_1)f(z_1)w(z_{N+1})f(z_{N+1})}{B_{N+2}} \bigg( \prod_{n=1}^{N+2} q(\ud z_n) \bigg) \\
&\quad=\int_{\Sp\times\X^{N+1}}\frac{w(z_1)f(z_1)w(z_{N+2})f(z_{N+2})}{B_{N+2}}\bigg( \prod_{n=1}^{N+2} q(\ud z_n) \bigg)
\end{align*}
and
$$
\int_{\Sp\times\X^{N+1}} \frac{A_{N}w(z_{N+1})}{B_{N+2}B_N}\bigg( \prod_{n=1}^{N+2}q(\ud z_n) \bigg)
=\int_{\Sp\times\X^{N+1}}\frac{A_{N}w(z_{N+2})}{B_{N+2}B_N}\bigg( \prod_{n=1}^{N+2} q(\ud z_n) \bigg),
$$
where both equations above follow from symmetricity of the integrand with respect to $z_{N+1}$ and $z_{N+2}$.

Combining \eqref{eq:cx-expression} with application of Lemma \ref{lemma:i-sir-exchangeable-property} in Appendix \ref{app:monotonicity-convexity} with $h:\W{N+1}\times \X\to [0,\infty)$, $h(z_1,\dots,z_{N+2})=w(z_{N+2})/(\sum_{j=1}^{N+2} w(z_j))$, it follows that $\angles{f|P_{N}f}_\pi-2\angles{f|P_{N+1}f}_\pi+\angles{f|P_{N+2}f}_\pi$ may be written as
\begin{align*}
&\frac{2}{N}\int_{\W{N}\times \X^2} 
\frac{w(z_{N+2})}{\sum_{j=1}^{N+2} w(z_j)}\frac{\Big(\sum_{i=1}^{N}w(z_i)f(z_i)\Big)^2}{\sum_{k=1}^{N} w(z_k)}\bigg( \prod_{n=1}^{N+2} q(\ud z_n) \bigg) \\
&\qquad-\frac{2}{N+1}\int_{\W{N+1}\times \X} 
\frac{w(z_{N+2})}{\sum_{j=1}^{N+2} w(z_j)}\frac{\Big(\sum_{i=1}^{N+1}w(z_i)f(z_i)\Big)^2}{\sum_{k=1}^{N+1} w(z_k)}\bigg( \prod_{n=1}^{N+2} q(\ud z_n) \bigg).
\end{align*}
Therefore, the claim follows from Lemma \ref{lemma:i-sir-inequality} in Appendix \ref{app:monotonicity-convexity} with the $h$ defined above, combined with Lemma \ref{lemma:i-sir-inner-product-differences} and the fact that $w(x)>0$ for all $x\in \Sp$.
\end{proof}

Covariances can be expressed with the help of i.i.d. random variables from $q$. We will use this property and the fact that if $X\sim q$, then $\E[w(X)f(X)]=\E_\pi[f]$ for $\pi$-integrable $f$ and $\E[w(X)]=1$, to show the following limit property:

\begin{proposition}
Let $f\in L^2(\pi)$, then 
$$\angles{f|P_{N}f}_\pi\xrightarrow{N\to \infty} \E_\pi[f]^2.$$
\end{proposition}
\begin{proof}
Let $Z_1,Z_2,\dots$ be i.i.d. from $q$. Let us define sequences of random variables $(A_N)$, $(B_N)$, $(C_N)$ and $(D_N)$:
\begin{align*}
A_N&=\frac{1}{N}\frac{\sum_{i=1}^Nw(Z_i)^2f(Z_i)^2}{\sum_{j=1}^Nw(Z_j)},  \qquad &B_N&=\frac{1}{N}\frac{\sum_{i=1}^N\sum_{j=1,j\neq i}^Nw(Z_i)w(Z_j)f(Z_i)f(Z_j)}{\sum_{k=1}^N w(Z_k)}, \\
C_N&=\frac{w(Z_1)^2f(Z_1)^2}{\sum_{j=1}^N w(Z_j)},\qquad &D_N&=\frac{1}{N^2}\frac{\big(\sum_{i=1}^N\sum_{j=1,j\neq i}^N w(Z_i)w(Z_j)f(Z_i)f(Z_j)\big)^2}{\sum_{k=1}^N \sum_{l=1,l\neq k}^N w(Z_k)w(Z_l)},
\end{align*}
whenever the denominator is positive, and otherwise $0$.
We have by Lemma \ref{lemma:inner-product-another-representation}, that
\begin{equation*}
\angles{f|P_{N}f}_\pi= \E[A_N]+\E[B_N].
\end{equation*}
By strong law of large numbers, $C_N\to 0$ and $A_N+B_N\to \E_\pi[f]^2$ almost surely. Since $|C_N(Z)|\le w(Z_1)f(Z_1)^2$ for all $N$, we have by dominated convergence $\E[C_N]\to 0$ and because $Z_i$ are exchangeable, we have $\E[A_N]=\E[C_N]\to 0$, hence $A_N\to 0$ in probability and we obtain $B_N\to \E_\pi[f]^2$ in probability.
It holds that $\E[|B_N|^2] \le \E[D_N]$, and 
for $N\ge 2$, we have
\begin{align*}
\E[D_N]\le \E\Big[\frac{1}{N^2}\sum_{i=1}^N\sum_{j=1,j\neq i}^Nw(Z_i)w(Z_j)f(Z_i)^2f(Z_j)^2\Big] 
=\frac{N(N-1)}{N^2}\E_\pi[f^2]^2, 
\end{align*}
where the inequality follows from Sedrakyan's inequality (Lemma \ref{lemma:sedrakyan-inequality} in Appendix \ref{app:differences}) and the equality from the fact that $Z_1,\dots,Z_N$ are i.i.d. Thus, $\sup_{N}\E[|B_N|^2]<\infty$, which implies that $(B_N)$ is uniformly integrable and therefore $\E[B_N] \to \E_\pi[f]^2$.
\end{proof}

\subsection{Rejection probabilities}\label{subsection:rejection-probabilities-convex-sequence}

The rejection probability for Algorithm \ref{alg:isir} is defined as
\begin{align*}
\epsilon(N,z_1)=\int_{\X^{N-1}} \frac{w(z_1)}{\sum_{i=1}^Nw(z_i)} \bigg(\prod_{n=2}^N q(\ud z_n) \bigg)
\end{align*}
and the average rejection probability
\begin{align*}
\epsilon(N)=\int_{\Sp} \epsilon(N,z_1) \pi(\ud z_1)=\int_{\W{N}} \frac{w(z_1)^2}{\sum_{i=1}^Nw(z_i)} \bigg(\prod_{n=1}^N q(\ud z_n) \bigg),
\end{align*}
when $N\in \N_+$. We use this notation throughout the upcoming sections. The following result establishes the monotonicity and convexity for the rejection probabilities, which lead to these properties for the approximate kernel $\tilde{P}_\lambda$.

\begin{lemma}\label{lemma:i-sir-rejection-inequalities}
Let $N\in \N_+$ and $x\in \Sp$, then
\begin{align*}
\epsilon(N)-\epsilon(N+1)&>0,\qquad&\epsilon(N)-2\epsilon(N+1)+\epsilon(N+2)&>0, \\
\epsilon(N,x)-\epsilon(N+1,x)&>0, \qquad &\epsilon(N,x)-2\epsilon(N+1,x)+\epsilon(N+2,x)&>0.
\end{align*}
\end{lemma}
\begin{proof}
Let us denote $z_1=x$ and $B_n=\sum_{i=1}^n w(z_i)$, then
\begin{align*}
\epsilon(N,z_1)-\epsilon(N+1,z_1)
&=w(z_1)\int_{\X^{N}} \frac{w(z_{N+1})}{B_NB_{N+1}} \bigg( \prod_{n=2}^{N+1} q(\ud z_n) \bigg) >0.
\end{align*}
We have also
\begin{align*}
&\epsilon(N,z_1)-2\epsilon(N+1,z_1)+\epsilon(N+2,z_1) \\
&\quad=w(z_1)\int_{\X^{N+1}}\frac{B_{N+1}B_{N+2}-2B_{N}B_{N+2}+B_{N}B_{N+1}}{B_{N}B_{N+1}B_{N+2}}\bigg( \prod_{n=2}^{N+2} q(\ud z_n) \bigg) \\
&\quad=w(z_1)\int_{\X^{N+1}}\frac{[w(z_{N+1})-w(z_{N+2})]B_{N+1}+2w(z_{N+1})w(z_{N+2})}{B_{N}B_{N+1}B_{N+2}}\bigg( \prod_{n=2}^{N+2} q(\ud z_n) \bigg) \\
&\quad=w(z_1)\int_{\X^{N+1}}\frac{2w(z_{N+1})w(z_{N+2})}{B_{N}B_{N+1}B_{N+2}}\bigg( \prod_{n=2}^{N+2} q(\ud z_n) \bigg) \\
&\quad>0,
\end{align*}
where the second equality follows from Lemma \ref{lemma:differences} in Appendix \ref{app:differences}, and the third from
$$
\int_{\X^{N+1}}\frac{w(z_{N+1})}{B_{N}B_{N+2}}\bigg( \prod_{n=2}^{N+2} q(\ud z_n) \bigg) =
\int_{\X^{N+1}}\frac{w(z_{N+2})}{B_{N}B_{N+2}}\bigg( \prod_{n=2}^{N+2} q(\ud z_n) \bigg),
$$
where the equality follows from symmetricity of the integrand with respect to $z_{N+1}$ and $z_{N+2}$.
Strict inequalities hold also for the average rejection probability, since $\epsilon(N)=\int_{\Sp}\epsilon(N,x)\pi(\ud x)$.
\end{proof}

Note that the holding probability of Algorithm \ref{alg:isir} for $z_1\in \Sp$ satisfies
\begin{align*}
P_N(z_1,\braces{z_1})&=\int_{\X^{N-1}}\sum_{i=1}^N \frac{w(z_i)}{\sum_{k=1}^Nw(z_k)} \charfun{z_i=z_1} \bigg( \prod_{n=2}^{N} q(\ud z_n) \bigg).
\end{align*}
In the case $z_1$ is not an atom, the holding probability $P_N(z_1,\braces{z_1})$ coincides with the rejection probability $\epsilon(N,z_1)$, for which the monotonicity and sequential convexity are shown above. When $z_1$ is an atom, we cannot use the same results directly, but we can show instead that there exists $g\in L^2(\pi)$ such that $P_N(z_1,\braces{z_1})=\angles{g|P_\lambda g}_\pi$ (see Lemma \ref{lemma:discrete-i-sir-transition} in Section \ref{section-convexity-monotonicity-of-isir}), so the same properties follow from results in Section \ref{subsection:inner-product-convex-sequence}. 

Rejection probability and holding probability can be expressed with the help of i.i.d. random variables from $q$. We will use this property to show the following lemma:
\begin{lemma}\label{lemma:holding-rejection-limit}
Let $x\in \Sp$. It holds that
\begin{align*}
\lim_{N\to \infty} \epsilon(N)=0, \qquad \lim_{N\to \infty} \epsilon(N,x)=0, \qquad \lim_{N\to \infty} P_N(x,\braces{x})=\pi(\braces{x}).
\end{align*}
\end{lemma}
\begin{proof}
Let $Z_2,Z_3,\dots$ be i.i.d. from $q$, then
$$P_N(x,\braces{x})=\E\bigg[\frac{w(x)+\sum_{i=2}^Nw(Z_i)\charfun{Z_i=x}}{w(x)+\sum_{j=2}^N w(Z_j)} \bigg]
,\qquad \epsilon(N,x)=\E\bigg[\frac{w(x)}{w(x)+\sum_{j=2}^N w(Z_j)} \bigg].
$$
In both equations the term inside the expectation is dominated by $1$ for all $N$ and by strong law of large numbers the ratio in the first equation converges to $\pi(\braces{x})$ and the ratio in the second equation converges to $0$ almost surely as $N\to \infty$.
Therefore, by dominated convergence we obtain the claim for the holding probability and the rejection probability. The claim for the average rejection probability follows by dominated convergence.
\end{proof}

\subsection{Properties for the continuous relaxation}\label{section-convexity-monotonicity-of-isir}

We can generalise the properties of the i-SIR (Algorithm \ref{alg:isir}) for its continuous relaxation (Algorithm \ref{alg:continuous-isir}).  This is possible, since as mentioned in Section \ref{sec:isir}, it holds for the Markov chain of Algorithm \ref{alg:continuous-isir} that if $\lambda\in [1,\infty)$ then
$$P_\lambda(x,A)=\beta(\lambda) P_{\floor{\lambda}}(x,A) + \big(1-\beta(\lambda)\big)P_{\floor{\lambda}+1}(x,A),$$
where $\beta(\lambda) = \floor{\lambda}+1-\lambda \in (0,1]$, which yields
\begin{align}
\begin{split}\label{equation:continuous-versions}
\angles{f|P_\lambda f}_\pi&=\beta(\lambda) \angles{f|P_{\floor{\lambda}}f}_\pi + \big(1-\beta(\lambda)\big) \angles{f|P_{\floor{\lambda}+1}f}_\pi, \\
\epsilon(\lambda,x)&=\beta(\lambda)\epsilon(\floor{\lambda},x)+\big(1-\beta(\lambda)\big)\epsilon(\floor{\lambda}+1,x), \\
\epsilon(\lambda)&=\beta(\lambda) \epsilon(\floor{\lambda})+\big(1-\beta(\lambda)\big) \epsilon(\floor{\lambda}+1),
\end{split}
\end{align}
where $f\in L^2(\pi)$ and $x\in \Sp$. We use this notation throughout the upcoming sections. Because of the interpolation form of \eqref{equation:continuous-versions}, the properties established for integer $\lambda$ above  can be generalised to the space $[1,\infty)$:

\begin{theorem}\label{theorem:i-sir-convexity}
Let $f\in L^2(\pi)$, then $\lambda\mapsto \angles{f|P_\lambda f}_\pi$, $\lambda\ge1$ is continuous non-negative decreasing convex function. Additionally, we have the following properties:
\begin{enumerate}[(i)]
\item For any $\lambda\ge 1$, it holds that  $\angles{f|P_\lambda f}_\pi=0$ if and only if $\|f\|_\pi=0$.
\item $\lambda\mapsto \angles{f|P_\lambda f}_\pi$ is strictly decreasing if and only if $\bar{f}\in L_{0,+}^2(\pi)$.
\end{enumerate}
\end{theorem}
\begin{proof}
Follows from Lemma \ref{lemma:generalisation-to-continuous-function} in Appendix \ref{app:modification} when combined with Lemma \ref{lemma:inner-product-another-representation}, \ref{lemma:i-sir-inner-product-differences} and \ref{lemma:i-sir-inner-product-differences-3}. 
\end{proof}

\begin{theorem}\label{theorem:i-sir-convexity-2}
It holds that $\lambda\mapsto \epsilon(\lambda)$, $\lambda\ge 1$ is positive continuous convex and strictly decreasing function, the same holds for $\lambda\mapsto \epsilon(\lambda,x)$, $\lambda\ge 1$ when $x\in \Sp$.
\end{theorem}
\begin{proof}
By definition, $\epsilon(N)$ and $\epsilon(N,x)$ are positive for all $N\in \N_+$. Therefore, the claim follows from Lemma \ref{lemma:generalisation-to-continuous-function} in Appendix \ref{app:modification} when combined with Lemma \ref{lemma:i-sir-rejection-inequalities}.
\end{proof}

\begin{lemma}\label{lemma:discrete-i-sir-transition}
Let $\lambda\ge 1$ and $x\in \Sp$. If $x$ is an atom, then
\begin{align*}
P_\lambda(x,\braces{x})=\angles{g_{x}|P_\lambda g_{x}}_\pi,
\end{align*}
where 
$g_{x}(y)=\frac{1}{\sqrt{\pi(\braces{x})}}\charfun{y=x}$.
\end{lemma}
\begin{proof}
Since $x$ is an atom, we have $P_1(x,\braces{x})=\angles{g_x|g_x}_\pi$ and if $N\in \N_+$, $N\ge 2$, then
\begin{align*}
P_N(x,\braces{x})&=\int_{\X^{N-1}}\frac{w(x)+\sum_{i=2}^Nw(z_i)\charfun{z_i=x}}{w(x)+\sum_{k=2}^Nw(z_k)}\bigg( \prod_{n=2}^Nq(\ud z_n)\bigg) \\
&=\int_{\Sp}\int_{\X^{N-1}}\sum_{i=1}^N \frac{w(z_i)\charfun{z_i=x}}{\sum_{k=1}^Nw(z_k)}\frac{\charfun{z_1=x}}{\pi(\braces{x})}\bigg( \prod_{n=2}^Nq(\ud z_n)\bigg)\pi(\ud z_1).
\end{align*}
Thus, by \eqref{equation:isir-inner-product} we get the claim for $N\in \N_+$, from which the general version follows.
\end{proof}

\begin{theorem}\label{theorem:i-sir-convexity-3}
If $\#(\Sp)>1$ and $x\in \Sp$, then $\lambda\mapsto P_\lambda(x,\braces{x})$, $\lambda\ge 1$ is positive continuous convex and strictly decreasing function.
\end{theorem}
\begin{proof}
When $x$ is not an atom, the claim follows from Theorem \ref{theorem:i-sir-convexity-2}. When $x$ is an atom, the claim follows from Lemma \ref{lemma:discrete-i-sir-transition} and Theorem \ref{theorem:i-sir-convexity} by noticing $\bar{g}_x\in L_{0,+}^2(\pi)$.
\end{proof}

Similar properties for functions related to approximate i-SIR are given in Appendix \ref{app:approximate-properties}.

\section{Asymptotic variance}
\label{sec:asymptotic-variance}

One of our main results is showing for Algorithm \ref{alg:continuous-isir}, that the asymptotic variance as a function $\lambda\mapsto \mathrm{var}(P_\lambda,f)$, $\lambda\in [1,\infty)$ is monotonic and convex for all $f\in L^2(\pi)$. The monotonicity can sometimes be shown with the Peskun order \citep{peskun,tierney-note}, but for i-SIR, a straightforward calculation shows a contradiction with
$$\pi(x)=
\begin{cases}
2x, & x\in [0,1] \\
0, &x\notin [0,1],
\end{cases}
\qquad
q(x)=
\begin{cases}
1, & x\in [0,1] \\
0, &x\notin [0,1],
\end{cases}
$$
when choosing $x=0.2$ and comparing $P_2(x,\uarg)$ and $P_3(x,\uarg)$ with both sets $(0,0.1)$ and $(0.9,1)$. 
However, there is a weaker condition for showing monotonicity of asymptotic variance, namely the covariance ordering \cite{tierney-note,mira-geyer}, which in our setting means the monotonicity of $\lambda\mapsto \angles{f|P_\lambda f}_\pi$ for all $f\in L^2(\pi)$. In this paper monotonicity and convexity of asymptotic variance are obtained in a similar fashion as in \cite[e.g.][Proposition 26]{andrieu-vihola-order} using Dirichlet forms, these results are combined in Theorem \ref{theorem:asymptotic-variance-results} in Appendix \ref{app:hilbert-spaces}. Without loss of generality we can study $L_0^2(\pi)$ functions.

\begin{lemma}\label{lemma:i-sir-properties}
Let $\lambda\ge1$, then $P_\lambda$ is $\pi$-reversible. If Assumption \ref{a:bounded-weights} holds, then $I-P_\lambda:L_0^2(\pi)\to L_0^2(\pi)$ is invertible for $\lambda>1$.
\end{lemma}
\begin{proof}
We have $P_N$ is $\pi$-reversible for $N\in \N_+$, $N\ge 2$ \cite[e.g.][]{tjelmeland,andrieu-lee-vihola}, which also holds for $N=1$, since $P_1(x,\ud y)=\delta_{x}(\ud y)$.
Therefore, we have that $P_\lambda$ is $\pi$-reversible for $\lambda\ge 1$, since
$$P_\lambda(x,\ud y)=(\floor{\lambda}+1-\lambda)P_{\floor{\lambda}}(x,\ud y) + (\lambda-\floor{\lambda})P_{\floor{\lambda}+1}(x,\ud y).$$
If $\lambda>1$, then by Lemma \ref{lemma:i-sir-probability-inequality} in Appendix \ref{app:ergodicity-auxiliaries}
$$P_\lambda(x,A)=c\pi(A)+(1-c)Q(x,A),$$
where $c=\frac{\lambda-1}{2\hat{w}+\lambda-1}$ and $Q$ is $\pi$-reversible Markov kernel. Thus, for $f\in L_0^2(\pi)$, it holds that
$$(I-P_\lambda)f=f-(1-c)Qf=cf+(1-c)(I-Q)f$$
and
\begin{align*}
\angles{(I-P_\lambda)f|(I-P_\lambda)f}_\pi
&=c^2\angles{f|f}_\pi+c(1-c)\angles{f|(I-Q)f}_\pi+c(1-c)\angles{(I-Q)f|f}_\pi \\
&\quad +(1-c)^2\angles{(I-Q)f|(I-Q)f}_\pi \\
& \ge c^2\angles{f|f}_\pi,
\end{align*}
since the Dirichlet form is non-negative for Markov kernel. Therefore, $I-P_\lambda$ is invertible by Lemma \ref{lemma:invertible-requirement} in Appendix \ref{app:hilbert-spaces}.
\end{proof}

\begin{theorem}\label{theorem:i-sir-asymptotic-variance}
Let Assumption \ref{a:bounded-weights} hold and $f\in L_{0,+}^2(\pi)$, then $\lambda\mapsto \mathrm{var}(P_\lambda,f)$, $\lambda>1$ is strictly convex and decreasing function.
\end{theorem}
\begin{proof}
The claim follows from Theorem \ref{theorem:asymptotic-variance-results} \eqref{enum-asymptotic-variance-result-5} in Appendix \ref{app:hilbert-spaces} when combined with Lemma \ref{lemma:i-sir-properties} and Theorem \ref{theorem:i-sir-convexity}.
\end{proof}

In some applications Dirichlet form can be more interesting than the inner product $\angles{f|P_\lambda f}_\pi$. Let $\lambda\ge 1$, then the Dirichlet form for Algorithm \ref{alg:continuous-isir} is defined
$$\mathcal{E}_\lambda(f)=\angles{f|(I-P_\lambda )f}_\pi, \qquad $$
where $f\in L^2(\pi)$. Thus, we can give another version of Theorem \ref{theorem:i-sir-convexity}:
\begin{proposition}
Let $f\in L_{0,+}^2(\pi)$, then for Algorithm \ref{alg:continuous-isir}, it holds that $\lambda\mapsto \mathcal{E}_\lambda(f)$, $\lambda\ge 1$ is non-negative, continuous, concave and strictly increasing function.
\end{proposition}\label{prop:dirichlet-form-i-sir}
\begin{proof}
Since $\mathcal{E}_\lambda(g)=\angles{f|f}_\pi-\angles{f|P_\lambda f}_\pi$ and $\angles{f|P_1 f}_\pi=\angles{f|f}_\pi$, the claim follows from Theorem \ref{theorem:i-sir-convexity}.
\end{proof}

\section{Approximate i-SIR in the case of atoms}\label{sec:approximate}

Here, we show how approximation of the i-SIR kernel (Definition \ref{def:approximate-isir} in Section \ref{sec:efficiency}) is derived in the general case, that is, beyond \ref{a:non-atomic}.
Let us denote by $\hat{P}_c$ the transition probability in Lemma \ref{lemma:single-proposal-chain-asymptotic-variance} in Appendix \ref{app:approximate-properties}, where $c$ is the holding probability (rejection probability when $\pi$ has no atoms), then it also holds that
\begin{equation}\label{equation:defining-approximations}
\int \hat{P}_c(x,\braces{x}) \pi(\ud x)=\int c+(1-c)\pi(\braces{x}) \pi(\ud x).
\end{equation}
Finding the approximate kernel is done by solving $c$ from \eqref{equation:defining-approximations}, where $\hat{P}_c$ is replaced with i-SIR transition probability $P_\lambda$. Then, the single proposal Markov chain $\hat{P}_c$ with the corresponding holding probability $c$ is used as the approximation. 
The solution $c$ to \eqref{equation:defining-approximations} with the replacement is
\begin{align}\label{function:problem-solution}
\begin{split}
\psi(\lambda)&=\frac{\int_{\Sp}[P_\lambda(x,\braces{x})-\pi(\braces{x})]\pi(\ud x)}{1-\int_{\Sp}\pi(\braces{x})\pi(\ud x)} \\
&=\frac{\int_{\Sp_0}\epsilon(\lambda,x)\pi(\ud x)+\int_{\Sp_+}[P_\lambda(x,\braces{x})-\pi(\braces{x})]\pi(\ud x)}{1-\int_{\Sp_+}\pi(\braces{x})\pi(\ud x)}.
\end{split}
\end{align}
Thus, we obtain the following:
\begin{definition}
 \label{def:approximate-isir-2}
The approximate i-SIR transition $\tilde{Q}_{\lambda}$ of $P_\lambda$ is defined as
$$
\tilde{Q}_{\lambda}(x,A) = \big(1-\psi(\lambda)\big) \pi(A) + \psi(\lambda) \delta_x(A),
$$
where $\psi$ is given in \eqref{function:problem-solution}.
\end{definition}

\begin{proposition}
\label{proposition:approximate-i-sir-asymptotic-variance-2}
Let $\lambda\in(1,\infty)$ and $f\in L^2(\pi)$ such that $\var_\pi(f)>0$. If \ref{a:bounded-weights} holds, then
\begin{enumerate}[(i)]
\item \label{item:approx-asvar-2} $\displaystyle
\var(\tilde{Q}_{\lambda}, f) = \frac{1+\psi(\lambda)}{1-\psi(\lambda)} \var_\pi(f).
$
\item \label{item:approx-bounds-2} $\var_\pi(f) \le \var(\tilde{Q}_{\lambda}, f) \le \frac{4 \hat{w} + \lambda - 1}{\lambda - 1} \var_\pi(f)$.
\item \label{item:approx-decreasing-and-convex-2} $\lambda \mapsto \var(\tilde{Q}_{\lambda}, f)$ is strictly convex and decreasing on $(1,\infty)$.
\end{enumerate}
\end{proposition}
\begin{proof}
Case \eqref{item:approx-asvar-2} follows from Lemma \ref{lemma:single-proposal-chain-asymptotic-variance} in Appendix \ref{app:approximate-properties}, case \eqref{item:approx-bounds-2} from Lemma \ref{lemma:i-sir-probability-inequality-2} in Appendix \ref{app:upper-lower-bounds} and case \eqref{item:approx-decreasing-and-convex-2} from Lemma \ref{lemma:convexity-of-the-estimate-2} in Appendix \ref{app:approximate-properties}.
\end{proof}

We note that Definition \ref{def:approximate-isir} and Proposition \ref{proposition:approximate-i-sir-asymptotic-variance} in Section \ref{sec:efficiency} correspond to Definition \ref{def:approximate-isir-2}  and Proposition \ref{proposition:approximate-i-sir-asymptotic-variance-2} when Assumption \ref{a:non-atomic} holds.

We use the following shorthand notation
\begin{equation}
V_f(\lambda)=\mathrm{var}(P_\lambda,f), \qquad
G_f(\lambda)=\frac{1+\epsilon(\lambda)}{1-\epsilon(\lambda)}\mathrm{var}_\pi(f), \qquad
H_f(\lambda)=\frac{1+\psi(\lambda)}{1-\psi(\lambda)}\mathrm{var}_\pi(f).
\label{eq:v-g-h}
\end{equation}
Both functions $G_f$ and $H_f$ are well defined (Lemma \ref{lemma:convexity-of-the-estimate} and \ref{lemma:convexity-of-the-estimate-2} in Appendix \ref{app:approximate-properties}) and have similar properties compared to the asymptotic variance $V_f$ when \ref{a:bounded-weights} holds, namely they are strictly convex decreasing functions and they all behave between the same bounds (Theorem \ref {thm:isir-asvar-properties}, Proposition \ref{proposition:approximate-i-sir-asymptotic-variance} and \ref{proposition:approximate-i-sir-asymptotic-variance-2}).
These imply that to some extent we can use $G_{f}$ and $H_{f}$ as approximations to $V_f$. Nevertheless, the smaller $\hat{w}$ gets, the closer $G_f$ and $H_f$ are to the asymptotic variance, and when $\hat{w}=1$, they are the same:

\begin{proposition}\label{proposition:asymptotic-variance-connection}
If $\pi=q$ $\mu$-almost everywhere, then the following holds for all $x\in \Sp$ and $\lambda\ge 1$:
\begin{align*}
P_\lambda(x,A)&=b_\lambda\charfun{x\in A}+(1-b_\lambda)\pi(A), \qquad
\epsilon(\lambda,x)=b_\lambda, \qquad
\epsilon(\lambda)=b_\lambda,
\end{align*}
where $b_\lambda=\frac{1}{\floor{\lambda}}-\frac{\lambda-\floor{\lambda}}{(\floor{\lambda}+1)\floor{\lambda}}$ and for $f\in L_{0,+}^2(\pi)$, it holds that
$$\mathrm{var}(P_\lambda,f)=G_{f}(\lambda)=H_{f}(\lambda)=\frac{1 + b_\lambda}{1-b_\lambda}\mathrm{var}_\pi(f).$$
\end{proposition}
\begin{proof}
Since $\pi=q$ $\mu$-almost everywhere, we have $w=1$ $\mu$-almost everywhere in $\Sp$. Let $N\in \N_+$, then
\begin{align*}
P_N(x,A)
=\frac{1}{N}\charfun{x\in A}+\frac{N-1}{N}\pi(A), \qquad
\epsilon(N,x)&=\frac{1}{N}, \qquad \epsilon(N)=\frac{1}{N}.
\end{align*}
It holds for $\lambda\ge 1$ and $\beta(\lambda)=\floor{\lambda}+1-\lambda$ that
\begin{align*}
\beta(\lambda)\frac{1}{\floor{\lambda}}+\big(1-\beta(\lambda)\big)\frac{1}{\floor{\lambda}+1}=b_\lambda, \qquad
\beta(\lambda)\frac{\floor{\lambda}-1}{\floor{\lambda}}+\big(1-\beta(\lambda)\big)\frac{\floor{\lambda}}{\floor{\lambda}+1}=1-b_\lambda.
\end{align*}
Combining these with the definitions and Lemma \ref{lemma:single-proposal-chain-asymptotic-variance} in Appendix \ref{app:approximate-properties}, we obtain the claim.
\end{proof}

Namely if $\pi=q$ $\mu$-almost everywhere, then Algorithm \ref{alg:continuous-isir} gives the Markov chain in Lemma \ref{lemma:single-proposal-chain-asymptotic-variance} in Appendix \ref{app:approximate-properties}. 

\section{Experiments}
\label{sec:experiments}

We investigate three questions in our experiments: the quality of the approximation of the asymptotic variance, the minimisation based on it, and the behaviour of the adaptive i-SIR algorithm in practice. Section \ref{sec:discrete-experiments} focuses on the two first questions, investigating the behaviour of i-SIR in finite state spaces $\Sp\subset\R$, where we can either calculate numerically or estimate efficiently the asymptotic variance of the i-SIR from its transition probability. 
Sections \ref{sec:mixture-example} and \ref{sec:logistic} investigate the use of the adaptive i-SIR in $\Sp=\R^d$: Section \ref{sec:mixture-example} presents a simple mixture example and Section \ref{sec:logistic} an example on Bayesian logistic regression.

\subsection{Quality of asymptotic variance approximations in a finite state space}
\label{sec:discrete-experiments}

We first compare the behaviour of the asymptotic variance $V_f$ of the i-SIR algorithm with respect to the approximate asymptotic variances $G_f$ and $H_f$ given in \eqref{eq:v-g-h}, the corresponding loss functions and their minimisers. We consider different cost functions of type $c(\lambda)=a+\lambda$ for $a> 0$, different target and proposal distributions and test functions. 

In a finite state space \ref{a:bounded-weights} always holds, and the asymptotic variance of a Markov chain can be calculated in principle using the spectral decomposition of its transition probability (Lemma \ref{lemma:asymptotic-variance-of-finite-i-SIR} in Appendix \ref{app:asymptotic-variance-results}). In the case of i-SIR, the transition and rejection probabilities are not directly available, and direct calculation is too costly for large state space (and/or large number of proposals). Therefore, we consider two cases:
\begin{enumerate}[(i)]
\item Direct calculations ($\Sp$ small;  Experiments \ref{ex:1}--\ref{ex:4}).
\item Monte Carlo approximations ($\Sp$ large; Experiment \ref{ex:5}; see Proposition \ref{proposition:discrete-isir-matrix-and-rejection} in Appendix \ref{app:discrete-case-isir}).
\end{enumerate}

We let $\lambda_f = \arg\min_{\lambda>1} c(\lambda) V_f(\lambda)$, and similarly $\lambda_{G_f}$ and $\lambda_{H_f}$ stand for the minimisers  of the loss and the approximate loss functions using $G_f$ and $H_f$ in place of $V_f$, respectively. We define the following suboptimality factors:
\begin{equation*}
SO_{G}^{f}=\frac{V_f(\lambda_{G_f})}{V_f(\lambda_{f})} \qquad\text{and}\qquad 
SO_{H}^{f}=\frac{V_f(\lambda_{H_f})}{V_f(\lambda_{f})},
\end{equation*}
which measure the effective loss of using the minimisers of the approximate losses instead of the true loss minimiser $\lambda_f$.
For any (non-constant) $f\in L^2(\pi)$, we define its standardised version $\hat{f}(x)={\mathrm{var}^{-1/2}_\pi(f)}\big(f(x)-\pi(f)\big)$. Note that the minimisers $\lambda_f$, $\lambda_G$ and $\lambda_H$ are unaffected by the standardisation, and that $H_{\hat{f}}$ and $G_{\hat{f}}$ do not depend on the test function at all. We denote 
$\hat{G}=G_{\hat{f}}$, $\hat{H}=H_{\hat{f}}$, $\lambda_{\hat{G}}=\lambda_{G_{\hat{f}}}$ and $\lambda_{\hat{H}}=\lambda_{H_{\hat{f}}}$.

If the number of states $ n = \#(\Sp)=2$, then $H_f = V_f$ by Proposition \ref{proposition:two-state-isir} in Appendix \ref{app:discrete-case-isir}, but for $n \ge 3$, $V_f$  generally differs from $G_f$ and $H_f$. If the weight upper bound $\hat{w}$ is close to one, then Proposition \ref{proposition:asymptotic-variance-connection} and the upper and lower bounds in Appendix \ref{app:upper-lower-bounds} imply that all these functions are close to each other. But for large $\hat{w}$, these functions can differ more.

It is relatively easy to see from Algorithm \ref{alg:continuous-isir} that the i-SIR chain can get stuck to states $s\in \Sp$ for which $w(s)$ is large. We denote by $U_M=\braces{s\in \Sp:w(s)\ge M}$ and $L_m=\braces{s\in \Sp:w(s)\le m}$ the superlevel and sublevel sets of $w$, and consider the standardised versions of the following test functions: the identity function $f(x)=x$, the inverse weight function $g(x) = w^{-1}(x)$, the function 
$$h(x)=
\begin{cases}
x, &x\notin U_M \\
\sum_{z\notin U_M}z\frac{\pi(z)}{1-\pi(U_M)}, & x\in U_M,
\end{cases}
$$
the superlevel indicator $k(x)= \charfun{x \in U_M}$ and the sublevel indicator $l(x) = \charfun{x\in L_m}$.
The function $\hat{f}$ admits its smallest and largest values at the endpoints of $\Sp$, and for $\hat{g}$, these occur at the largest and smallest weight, respectively.
Function $\hat{h}$ is similar to $\hat{f}$, but $\hat{h}(s)=0$ whenever $w(s)\ge M$. 
In all the experiments, we calculate the asymptotic variance, its upper and lower bounds (Appendix \ref{app:upper-lower-bounds}), and the approximations $\hat{G}$ and $\hat{H}$ for $\lambda\in [2,150]$ with a spacing of $0.01$.
In Experiments \ref{ex:1}--\ref{ex:5}, we have omitted from the tables results related to $\hat{H}$, since in these examples $\lambda_{\hat{H}}$ was equal to $\lambda_{\hat{G}}$ with the tested $a$ values.

\experiment
\label{ex:1}

We consider the state space $\Sp=\braces{1,2,3,4,5}$ and target $\pi$ and proposal $q$ shown in Figure \ref{figure:example-1-prob-mass-function}. We choose $M=\hat{w}$ and $m=\min_{s\in \Sp}w(s)$, hence $U_M=\braces{4}$ and $L_m=\braces{3}$. 
The results of the experiment are illustrated in Figure \ref{figure-experiment-1}; note the logarithmic scale in \ref{figure:example-1-asymptotic-variance} and \ref{figure:example-1-asymptotic-variance-with-c}.
The weight upper bound is $\hat{w}\approx 1.76$, so the lower and upper bounds are not far off from the functions in Figure \ref {figure:example-1-asymptotic-variance} and \ref {figure:example-1-asymptotic-variance-with-c}, although there are some differences for test functions $\hat{k}$ and $\hat{l}$, which are concentrated on the extreme weights. 
The suboptimality factors in Table \ref{table:example-1-minimum} show that their values stay close to $1$ for all test functions, even when $a$ increases. 

\begin{figure}
\begin{minipage}[b]{0.26\textwidth}
\begin{subfigure}[b]{\textwidth}
\includegraphics{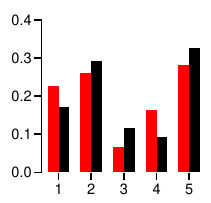}
\caption{}\label{figure:example-1-prob-mass-function}
\end{subfigure}\\
\begin{subfigure}[b]{\textwidth}
\includegraphics{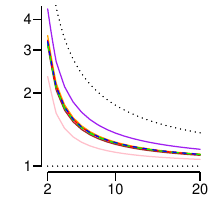}
\caption{}\label{figure:example-1-asymptotic-variance}
\end{subfigure}
\end{minipage}
\begin{subfigure}[b]{0.64\textwidth}
\includegraphics{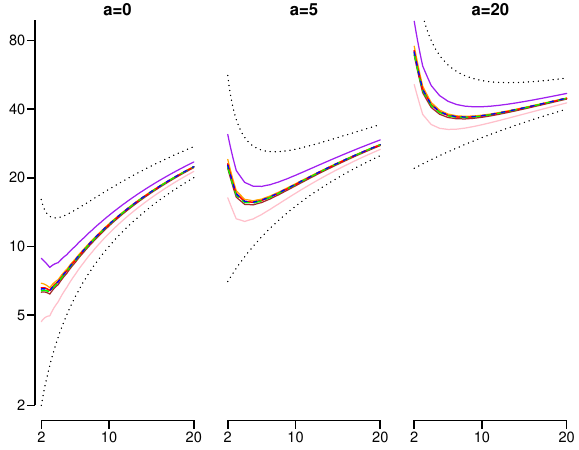}
\caption{}\label{figure:example-1-asymptotic-variance-with-c}
\end{subfigure}
\caption{Experiment \ref{ex:1}: \subref{figure:example-1-prob-mass-function} probability mass functions $\pi$ (red) and $q$ (black). \subref{figure:example-1-asymptotic-variance} functions $V_{\hat{f}}$ (red), $V_{\hat{g}}$ (orange), $V_{\hat{h}}$ (brown), $V_{\hat{k}}$ (purple), $V_{\hat{l}}$ (pink), $\hat{G}$ (green dashed), $\hat{H}$ (blue dashed) with upper/lower bounds (black dotted). \subref{figure:example-1-asymptotic-variance-with-c} loss i.e.~the functions in \subref{figure:example-1-asymptotic-variance} multiplied by the cost $c(\lambda) = a + \lambda$.
}\label{figure-experiment-1}
\end{figure}

\begin{table}
\Small{
\tiny
\begin{tabular}[t]{|c|c|c|c|c|c|c|c|c|c|c|c|}
\hline
\rowcolor{gray!30}
a&$\lambda_{\hat{G}}$&$\lambda_{\hat{f}}$&$SO_G^{\hat{f}}$&$\lambda_{\hat{g}}$
&$SO_G^{\hat{g}}$&$\lambda_{\hat{h}}$&$SO_G^{\hat{h}}$&$\lambda_{\hat{k}}$&$SO_G^{\hat{k}}$&$\lambda_{\hat{l}}$&$SO_G^{\hat{l}}$\\
\hline
0&3&3&1&3&1&3&1&3&1&2&1.06\\
\hline
0.1&3&3&1&3&1&3&1&3&1&2&1.04\\
\hline
1&3&3&1&3&1&3&1&4&1.02&3&1\\
\hline
2&4&4&1&4&1&4&1&4&1&3&1.04\\
\hline
5&5&5&1&5&1&5&1&6&1.00&4&1.03\\
\hline
10&6&6&1&6&1&6&1&7&1.01&5&1.01\\
\hline
20&8&8&1&8&1&8&1&9&1.01&6&1.01\\
\hline
\end{tabular}
}
\vspace{0.1 cm}
\caption{Experiment \ref{ex:1}: Minimisers and their corresponding suboptimality factors.}\label{table:example-1-minimum}
\end{table}

\experiment
\label{ex:2}

Next, we consider state space $\Sp=\braces{1,2,\ldots,6}$ and the probability mass functions $\pi,q:\Sp\to [0,1]$ shown in Figure \ref{figure:example-1.2-prob-mass-function}, where states $1$ to $4$ have the same mass as in Experiment \ref{ex:1}. 
As in Experiment \ref{ex:1} we choose $M=\hat{w}$ and $m=\min_{s\in \Sp}w(s)$, hence $U_M=\braces{6}$ and $L_m=\braces{5}$. 
The weights can take much larger values $w(6)=\hat{w}\approx 430.99$ than in Experiment \ref{ex:1}, so the lower and upper bounds are much more apart, and the shape of the asymptotic variance in Figure \ref{figure:example-1.2-asymptotic-variance} differs substantially from the approximations for most of the test functions. Similarly, in Figure \ref{figure:example-1.2-asymptotic-variance-with-c}, the loss with approximations $\hat{H}$ and $\hat{G}$ have a noticeable discrepancy for $\hat{f}$, $\hat{g}$, $\hat{k}$ and $\hat{l}$. This reflects also in the discrepancy of the minimisers, which is sometimes ten-fold, and the suboptimality factors shown in Table \ref{table:example-1.2-minimum}.

\begin{figure}
\begin{minipage}[b]{0.26\textwidth}
\begin{subfigure}[b]{\textwidth}
\includegraphics{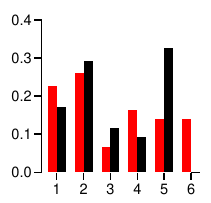}
\caption{}\label{figure:example-1.2-prob-mass-function}
\end{subfigure}\\
\begin{subfigure}[b]{\textwidth}
\includegraphics{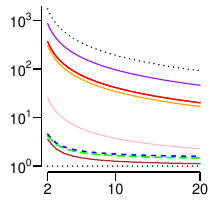}
\caption{}\label{figure:example-1.2-asymptotic-variance}
\end{subfigure}
\end{minipage}
\begin{subfigure}[b]{0.64\textwidth}
\includegraphics{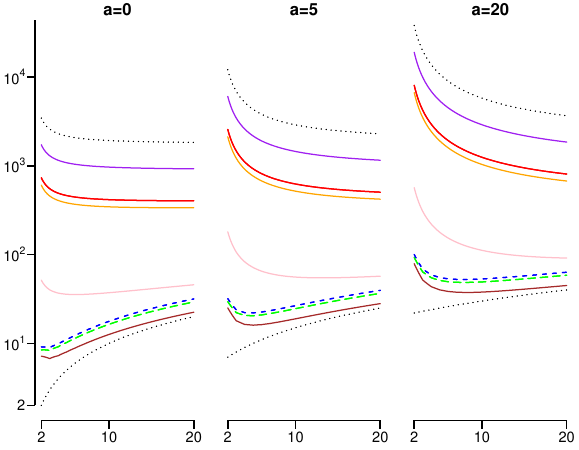}
\caption{}\label{figure:example-1.2-asymptotic-variance-with-c}
\end{subfigure}
\caption{Experiment \ref{ex:2}: \subref{figure:example-1.2-prob-mass-function} probability mass functions $\pi$ (red) and $q$ (black). \subref{figure:example-1.2-asymptotic-variance} functions $V_{\hat{f}}$ (red), $V_{\hat{g}}$ (orange), $V_{\hat{h}}$ (brown), $V_{\hat{k}}$ (purple), $V_{\hat{l}}$ (pink), $\hat{G}$ (green dashed), $\hat{H}$ (blue dashed) with upper/lower bounds (black dotted). \subref{figure:example-1.2-asymptotic-variance-with-c} loss i.e.~the functions in \subref{figure:example-1.2-asymptotic-variance} multiplied by the cost $c(\lambda) = a + \lambda$.}
\end{figure}

\begin{table}
\Small{
\tiny
\begin{tabular}[t]{|c|c|c|c|c|c|c|c|c|c|c|c|}
\hline
\rowcolor{gray!30}
a&$\lambda_{\hat{G}}$&$\lambda_{\hat{f}}$&$SO_G^{\hat{f}}$&$\lambda_{\hat{g}}$
&$SO_G^{\hat{g}}$&$\lambda_{\hat{h}}$&$SO_G^{\hat{h}}$&$\lambda_{\hat{k}}$&$SO_G^{\hat{k}}$&$\lambda_{\hat{l}}$&$SO_G^{\hat{l}}$\\
\hline
0&3&21&1.37&19&1.35&3&1&33&1.41&6&1.13\\
\hline
0.1&3&22&1.40&20&1.39&3&1&34&1.45&6.15&1.15\\
\hline
1&3&29&1.75&26&1.73&3&1&45&1.84&8&1.32\\
\hline
2&4&35&1.70&32&1.68&4&1&55&1.80&9.7&1.24\\
\hline
5&5&49&2.00&44.85&1.96&5&1&78&2.16&13.1&1.31\\
\hline
10&6&66&2.37&60&2.31&6&1&104&2.63&17.48&1.39\\
\hline
20&8&91&2.68&82&2.59&8&1&143&3.07&23.82&1.39\\
\hline
\end{tabular}
}
\vspace{0.1 cm}
\caption{Experiment \ref{ex:2}: Minimisers and their corresponding suboptimality factors.}\label{table:example-1.2-minimum}
\end{table}

\experiment
\label{ex:3}

We consider again the same state space $\Sp=\braces{1,2,\ldots, 6}$. The probability mass functions $\pi,q:\Sp\to [0,1]$ shown in Figure \ref{figure:example-1.3-prob-mass-function} are the same as in Experiments \ref{ex:1} and \ref{ex:2} for states $1$ to $4$, and $\hat{w}$ is same as in Experiment \ref{ex:2}, but $\pi(6)$ and $q(6)$ are $1000$ times smaller than in Experiment \ref{ex:2}. 
This example mimics the scenario where $\hat{w}$ is large only in a `small' part of the state space, which may be a common feature in practical settings.
We choose $M$ and $m$ as before, hence $U_M=\braces{6}$ and $L_m=\braces{3}$.
The results in Figure \ref{figure-experiment-1.3} and Table \ref{table:example-1.3-minimum} are favourable (similar to Experiment \ref{ex:1}) for most test functions. This suggests that $\tilde{P}_\lambda$ can remain a useful approximation to $P_\lambda$ even if $\hat{w}$ is large. However, the results for the function $\hat{k}$, which has all of its variability concentrated on $U_{\hat{w}}$, show discrepancy similar to what was observed in Experiment \ref{ex:2}: as $a$ increases, the suboptimality factor increases.

\begin{figure}
\begin{minipage}[b]{0.26\textwidth}
\begin{subfigure}[b]{\textwidth}
\includegraphics{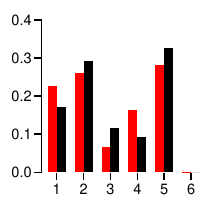}
\caption{}\label{figure:example-1.3-prob-mass-function}
\end{subfigure}\\
\begin{subfigure}[b]{\textwidth}
\includegraphics{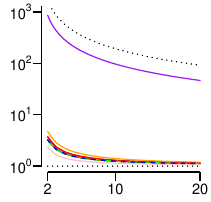}
\caption{}\label{figure:example-1.3-asymptotic-variance}
\end{subfigure}
\end{minipage}
\begin{subfigure}[b]{0.64\textwidth}
\includegraphics{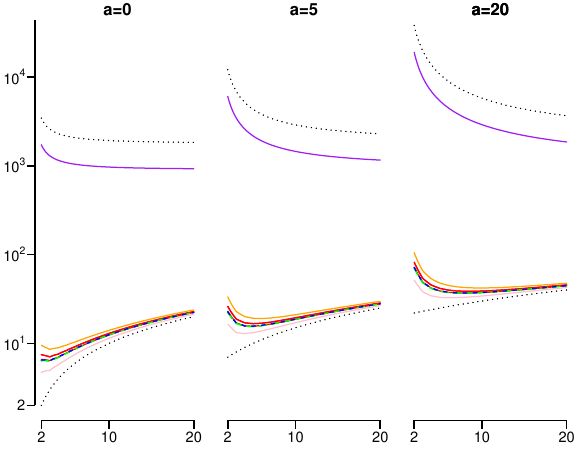}
\caption{}\label{figure:example-1.3-asymptotic-variance-with-c}
\end{subfigure}
\caption{Experiment \ref{ex:3}: \subref{figure:example-1.3-prob-mass-function} probability mass functions $\pi$ (red) and $q$ (black). \subref{figure:example-1.3-asymptotic-variance} functions $V_{\hat{f}}$ (red), $V_{\hat{g}}$ (orange), $V_{\hat{h}}$ (brown), $V_{\hat{k}}$ (purple), $V_{\hat{l}}$ (pink), $\hat{G}$ (green dashed), $\hat{H}$ (blue dashed) with upper/lower bounds (black dotted). \subref{figure:example-1.3-asymptotic-variance-with-c} loss i.e.~the functions in \subref{figure:example-1.3-asymptotic-variance} multiplied by the cost $c(\lambda) = a + \lambda$.}\label{figure-experiment-1.3}
\end{figure}

\begin{table}
\Small{
\tiny
\begin{tabular}[t]{|c|c|c|c|c|c|c|c|c|c|c|c|}
\hline
\rowcolor{gray!30}
a&$\lambda_{\hat{G}}$&$\lambda_{\hat{f}}$&$SO_G^{\hat{f}}$&$\lambda_{\hat{g}}$
&$SO_G^{\hat{g}}$&$\lambda_{\hat{h}}$&$SO_G^{\hat{h}}$&$\lambda_{\hat{k}}$&$SO_G^{\hat{k}}$&$\lambda_{\hat{l}}$&$SO_G^{\hat{l}}$\\
\hline
0&3&3&1&3&1&3&1&30&1.41&2&1.06\\
\hline
0.1&3&3&1&3&1&3&1&32&1.45&2&1.04\\
\hline
1&3&3&1&3&1&3&1&43&1.82&3&1\\
\hline
2&4&4&1&4&1&4&1&52&1.79&3&1.04\\
\hline
5&5&5&1&5&1&5&1&73&2.14&4&1.03\\
\hline
10&6&6&1&6&1&6&1&98&2.60&5&1.01\\
\hline
20&8&9&1.00&9&1.00&8&1&136&3.02&6&1.01\\
\hline
\end{tabular}
}
\vspace{0.1 cm}
\caption{Experiment \ref{ex:3}: Minimisers and their corresponding suboptimality factors.}\label{table:example-1.3-minimum}
\end{table}

\experiment
\label{ex:4}

Here, we let $\Sp=\braces{1,2,3,4,5}$ and consider the probability mass functions $\pi,q:\Sp\to [0,1]$ shown in Figure \ref{figure:example-1.4-prob-mass-function}. The proposal $q$ is the same as in Experiment \ref{ex:1} but $\pi$ is modified: we divide the mass from state $5$ evenly to the other states in a way that $w(5)=1/d$, where $d$ is $\hat{w}$ of Experiment \ref{ex:2} and \ref{ex:3}. The weights are small $w(4)=\hat{w}\approx 2.50$. We choose $M$ and $m$ as before, hence $U_M=\braces{4}$ and $L_m=\braces{5}$. With test functions $\hat{g}$ and $\hat{l}$, which take their extremal values at state 5, the results differ from earlier experiments:  with large $a$, the cost function minimisers are much smaller than with approximations. These functions admit their extreme values at $5$, which the i-SIR chain does not visit often. For the rest of the test functions, the losses looks similar to the approximate losses as the results in Table \ref{table:example-1.4-minimum} confirm.

\begin{figure}
\begin{minipage}[b]{0.26\textwidth}
\begin{subfigure}[b]{\textwidth}
\includegraphics{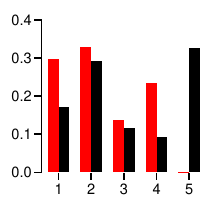}
\caption{}\label{figure:example-1.4-prob-mass-function}
\end{subfigure}\\
\begin{subfigure}[b]{\textwidth}
\includegraphics{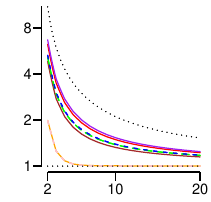}
\caption{}\label{figure:example-1.4-asymptotic-variance}
\end{subfigure}
\end{minipage}
\begin{subfigure}[b]{0.64\textwidth}
\includegraphics{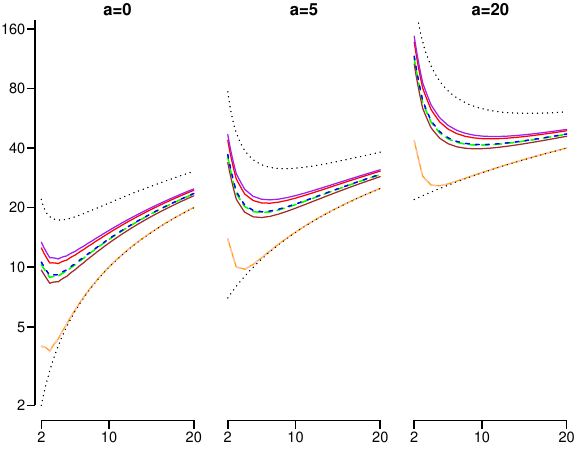}
\caption{}\label{figure:example-1.4-asymptotic-variance-with-c}
\end{subfigure}
\caption{Experiment \ref{ex:4}: \subref{figure:example-1.4-prob-mass-function} probability mass functions $\pi$ (red) and $q$ (black). \subref{figure:example-1.4-asymptotic-variance} functions $V_{\hat{f}}$ (red), $V_{\hat{g}}$ (orange), $V_{\hat{h}}$ (brown), $V_{\hat{k}}$ (purple), $V_{\hat{l}}$ (pink dashed), $\hat{G}$ (green dashed), $\hat{H}$ (blue dashed) with upper/lower bounds (black dotted). \subref{figure:example-1.4-asymptotic-variance-with-c} loss i.e.~the functions in \subref{figure:example-1.4-asymptotic-variance} multiplied by the cost $c(\lambda) = a + \lambda$.}
\end{figure}

\begin{table}
\Small{
\tiny
\begin{tabular}[t]{|c|c|c|c|c|c|c|c|c|c|c|c|}
\hline
\rowcolor{gray!30}
a&$\lambda_{\hat{G}}$&$\lambda_{\hat{f}}$&$SO_G^{\hat{f}}$&$\lambda_{\hat{g}}$
&$SO_G^{\hat{g}}$&$\lambda_{\hat{h}}$&$SO_G^{\hat{h}}$&$\lambda_{\hat{k}}$&$SO_G^{\hat{k}}$&$\lambda_{\hat{l}}$&$SO_G^{\hat{l}}$\\
\hline
0&3&4&1.01&3&1&3&1&4&1.01&3&1\\
\hline
0.1&3&4&1.02&3&1&3&1&4&1.02&3&1\\
\hline
1&4&4&1&3&1.08&4&1&5&1.01&3&1.08\\
\hline
2&5&5&1&3&1.15&5&1&5&1&3&1.15\\
\hline
5&6&7&1.00&4&1.14&6&1&7&1.01&4&1.15\\
\hline
10&8&8&1&4&1.19&7&1.00&9&1.00&4&1.19\\
\hline
20&10&11&1.00&5&1.17&9&1.00&11&1.00&5&1.17\\
\hline
\end{tabular}
}
\vspace{0.1 cm}
\caption{Experiment \ref{ex:4}: Minimisers and their corresponding suboptimality factors.}\label{table:example-1.4-minimum}
\end{table}

\experiment
\label{ex:5}

Contrary to the previous, purely artificial examples, this example is a discretisation of a continuous state model.
We consider $\Sp$ which is a discretisation of $[-3,3]$ to 61 equally spaced states, and $\pi$ and $q$ correspond to discretisation of two normal distributions $N(0,1/4)$ and $N(0,1)$, respectively, illustrated in Figure \ref{figure:example-2-prob-mass-function}.

For integers $N\ge 2$ we have used Monte Carlo estimation to estimate the transition and rejection probabilities (i.e.~\eqref{eq:transition-probability-as-expectation}  and \eqref{eq:rejection-probability-as-expectation} of Proposition \ref{proposition:discrete-isir-matrix-and-rejection} in Appendix \ref{app:discrete-case-isir}) by simulating $Z_1^N,\dots,Z_{100,000}^N\sim \mathrm{Multinom}(N-1,Q)$. We calculate both $P_N$ and $\epsilon(N)$ using the same $(Z_i^N)$, and also create new $Z_i^{N+1}$ conditional on  $Z_i^N$. 
The estimates $P_\lambda$ and $\epsilon(\lambda)$ for non-integer $\lambda\ge 1$ are derived from these. 

The weight upper bound $\hat{w}\approx 2.00$ is fairly small. We choose $U_M$ and $L_m$, where $M=1.9$ and $m=0.2$, for which $\pi(U_M)\approx 0.23$ and $\pi(L_m)\approx 0.012$. Figures \ref{figure:example-2-asymptotic-variance} and \ref{figure:example-2-asymptotic-variance-with-c} show that the asymptotic variances have some variation, but Table \ref{table:example-2-minimum} indicates that all the suboptimality factors are relatively small. With $\hat{g}$, the factors are the highest, but as $a$ grows, they get smaller.

\begin{figure}
\begin{minipage}[b]{0.26\textwidth}
\begin{subfigure}[b]{\textwidth}
\includegraphics{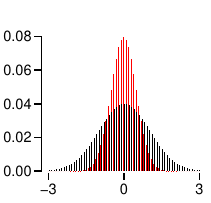}
\caption{}\label{figure:example-2-prob-mass-function}
\end{subfigure}\\
\begin{subfigure}[b]{\textwidth}
\includegraphics{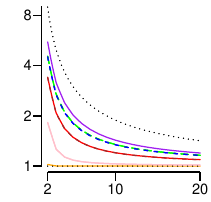}
\caption{}\label{figure:example-2-asymptotic-variance}
\end{subfigure}
\end{minipage}
\begin{subfigure}[b]{0.64\textwidth}
\includegraphics{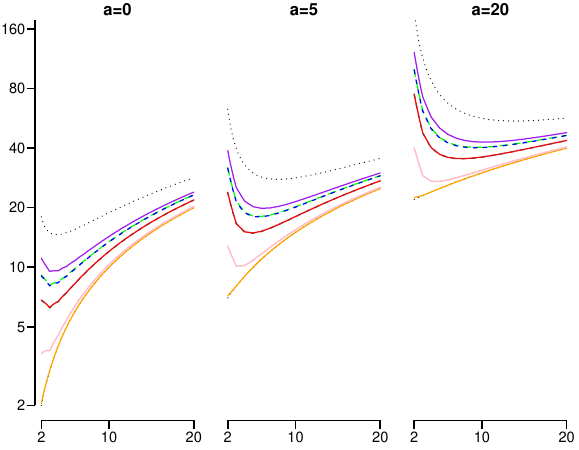}
\caption{}\label{figure:example-2-asymptotic-variance-with-c}
\end{subfigure}
\caption{Experiment 5: \subref{figure:example-2-prob-mass-function} probability mass functions $\pi$ (red) and $q$ (black). \subref{figure:example-2-asymptotic-variance} functions $V_{\hat{f}}$ (red), $V_{\hat{g}}$ (orange), $V_{\hat{h}}$ (brown dashed), $V_{\hat{k}}$ (purple), $V_{\hat{l}}$ (pink), $\hat{G}$ (green dashed), $\hat{H}$ (blue dashed) with upper/lower bounds (black dotted). \subref{figure:example-2-asymptotic-variance-with-c} loss i.e.~the functions in \subref{figure:example-2-asymptotic-variance} multiplied by the cost $c(\lambda) = a + \lambda$.}
\end{figure}

\begin{table}
\Small{
\tiny
\begin{tabular}[t]{|c|c|c|c|c|c|c|c|c|c|c|c|}
\hline
\rowcolor{gray!30}
a&$\lambda_{\hat{G}}$&$\lambda_{\hat{f}}$&$SO_G^{\hat{f}}$&$\lambda_{\hat{g}}$
&$SO_G^{\hat{g}}$&$\lambda_{\hat{h}}$&$SO_G^{\hat{h}}$&$\lambda_{\hat{k}}$&$SO_G^{\hat{k}}$&$\lambda_{\hat{l}}$&$SO_G^{\hat{l}}$\\
\hline
0&3&3&1&2&1.47&3&1&3&1&2&1.04\\
\hline
0.1&3&3&1&2&1.45&3&1&3&1&2&1.02\\
\hline
1&4&3&1.01&2&1.63&3&1.01&4&1&3&1.12\\
\hline
2&4&4&1&2&1.47&4&1&5&1.02&3&1.08\\
\hline
5&6&5&1.02&2&1.54&5&1.02&6&1&3&1.16\\
\hline
10&7&6&1.01&2&1.39&6&1.01&8&1.01&4&1.13\\
\hline
20&9&8&1.01&2&1.29&8&1.01&10&1.01&5&1.11\\
\hline
\end{tabular}
}
\vspace{0.1 cm}
\caption{Experiment 5: Minimisers and their corresponding suboptimality factors.}\label{table:example-2-minimum}
\end{table}

\subsection{Mixture distribution}
\label{sec:mixture-example}

We run our adaptive i-SIR (Algorithm \ref{alg:adaptive-isir}) on a Gaussian mixture target example from \citep{cardoso-samsonov-thin-moulines-olsson} in dimensions 7 and 10. The modes of the target are $m_1=(1,\ldots,1)^T$ and $m_2=(-2,0,\ldots,0)^T$ and the components have identity covariance. The proposal has zero mean and identity scale Student's $t$ distribution with $\nu=3$ degrees of freedom. We consider the same two test functions as in \citep{cardoso-samsonov-thin-moulines-olsson}: $f_1(x) = x_1$ and $f_2(x) = \charfun{x_{1:7}\in A} - \charfun{x_{1:7} \in B}$, with sets $A = [-2,6]\times[-1,1]^6$ and $B = [0.75,1.25]\times[1,2]\times[-0.1,0.1]^5$.
This example satisfies \ref{a:bounded-weights}, and we determined numerically\footnote{\label{fn:numerical_what}More specifically, we found the maximum of the unnormalised weight function numerically and divided it by the mean weight value from all i-SIR outputs with fixed $N$.} the value of the upper bound $\hat{w}\approx 20$ in dimension 7 and $\hat{w}\approx 64$ in dimension 10.

The methods are implemented in the Julia programming language \citep{julia}, using its native parallelisation to the i-SIR proposal/density evaluation.\footnote{Code available at \url{https://github.com/mvihola/AdaptiveIteratedSIR.jl}} The number of threads is set to four, which is the number of performance cores in the Macbook Pro 2020 with Apple M1 processor, where the tests are run. We first do pilot runs of $n=20,000$ iterations with i-SIR having fixed $N_i = 2^i+1$ for $i\in\{2,\ldots,p\}$ with $p=13$ in dimension 7 and $p=14$ in dimension 10. The runs are timed, and a linear regression model is fit to the times $T_i = a + b N_i + \epsilon_i$, where $\epsilon_i$ are the residuals (Figure \ref{fig:mixture-timing-adaptive} left). The cost function is formed by scaling the fitted model: $c(\lambda) = (a/b) + \lambda$. After this, we run the adaptive i-SIR for the model for 10,000 iterations in dimension 7 and for 100,000 iterations in dimension 10. We repeated the adaptive algorithm 100 times, with $N_{\mathrm{max}}=2^p+1$. Figure \ref{fig:mixture-timing-adaptive} (middle) shows the quantiles of the adapted $\lambda_k$ over the 100 replications (minimum/maximum, 90\% and 50\% bands shown). Figure \ref{fig:mixture-timing-adaptive} (right) shows the 
trace plot of the first coordinate of the samples, after a 10\% burn-in.

\begin{figure}
   \includegraphics{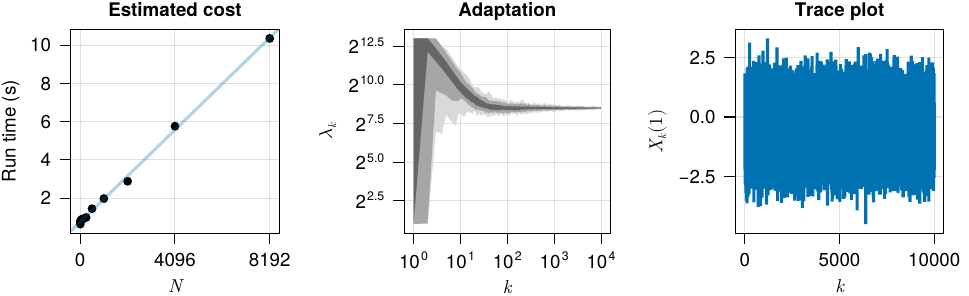}
   \includegraphics{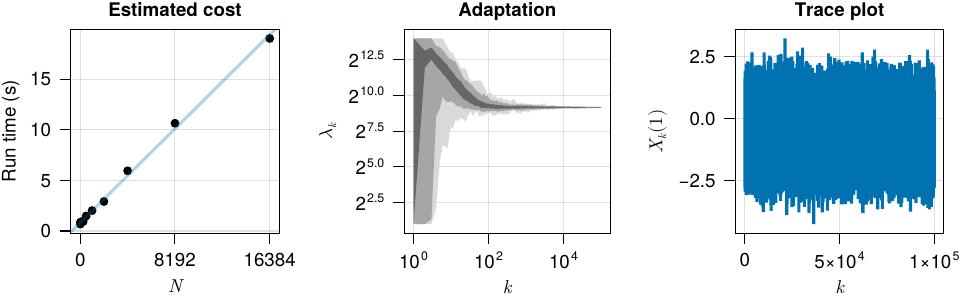}
   \caption{Average running time $T_i$ for different $N_i$ and the fitted regression line (left), the evolution of the quantiles of $\lambda_k$ and the trace plot of the first coordinate of adaptive i-SIR output (right) in the mixture example of Section \ref{sec:mixture-example} with dimension 7 (top) and 10 (bottom).}
   \label{fig:mixture-timing-adaptive}
\end{figure}

Then, we run i-SIRs with $N$ fixed to each $N_i$ and the final value of the adaptation. 
We use the initial sequence estimator to estimate the asymptotic variance \citep{geyer-practical}. Figure \ref{fig:mixture-example} (left) shows the integrated autocorrelation time (IACT) estimates for the test functions and based on our approximation. The IACT of $f_2$ behaves close to the approximation, whereas the IACT of $f_1$ becomes close to the approximation only for higher number of proposals $N$. Figure \ref{fig:mixture-example} (right) shows the estimated inverse relative efficiencies \citep{glynn-whitt}, that is, average running time per iteration times asymptotic variance. Note that both the $x$ and $y$ axes are logarithmic. The vertical line indicates the value of the number of proposals corresponding to the final value of adaptation. The adaptation minimises the approximate IRE, and approximately minimises the IRE of the second test function. For the first test function, a greater number of proposals would be optimal. However, the loss of efficiency is not large.

\begin{figure}
   \begin{tabular}{cc}
\includegraphics{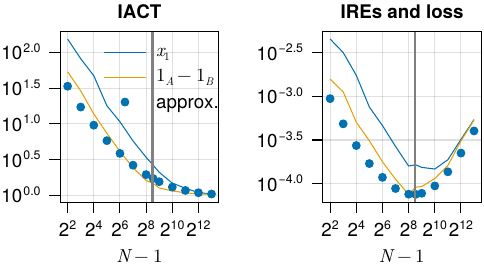} & \includegraphics{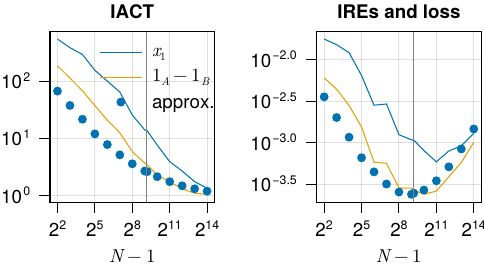} \\
(a) & (b)
   \end{tabular}
\caption{The estimated integrated autocorrelation times (left) and the inverse relative efficiencies and the loss (right) in the mixture example of Section \ref{sec:mixture-example} in dimension 7 (a) and 10 (b). The vertical grey band indicates the range of terminal values of adaptation.}
\label{fig:mixture-example}
\end{figure}

\subsection{Logistic regression}
\label{sec:logistic}

Our final example is a Bayesian logistic regression model on the Wisconsin Diagnostic Breast Cancer dataset (doi:10.24432/C5DW2B, \citep{street-wolberg-mangasarian}), consisting of 569 patients' data, with a binary outcome variable (benign/malignant), and with 30 covariates plus intercept. We use a Gaussian prior $\mathrm{pr}(x)=N(x; 0, 20 I_{31})$ for the regression coefficients. We first find the Laplace approximation $\hat{\pi}(x) = N(x; \mu,\Sigma)$ of the posterior, and use a (`defensive') mixture $q(x) = 0.1 \mathrm{pr}(x) + 0.9 \hat{\pi}(x)$ as the proposal; this guarantees bounded weights $w(x) = \pi(x)/q(x)$ as long as the likelihood is bounded. We found numerically\textsuperscript{\ref{fn:numerical_what}} the weight upper bound $\hat{w} \approx 9 \times 10^{37}$. The test functions are posterior density values $f_1(x) = \pi(x)$ and $f_2(x) = \| x - \mu \|$.

We use the same implementation as in Section \ref{sec:mixture-example}, with 10,000 pilot samples for timing
with $N_i = 2^i+1$ for $i\in\{2,\ldots,5\}$, and 100,000 iterations for the adaptive algorithm with $N_{\mathrm{max}}=2^5+1$. Figure \ref{fig:logistic} illustrates the results similar to Figure \ref{fig:mixture-timing-adaptive} (middle) and \ref{fig:mixture-example}. Also here, the adaptive algorithm appears to minimise the approximate IRE, and the IACT and IRE of the first test function follow the same shape as the approximation. The IRE of $f_2$ seems to be roughly minimised, but because there remains fluctuations in the estimated IACT of $f_2$, no clear conclusions can be drawn.

\begin{figure}
\includegraphics{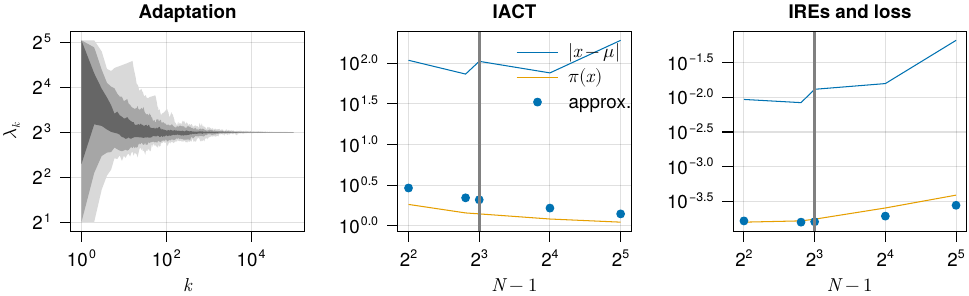}
\caption{The evolution of $\lambda_k$ (left), the estimated integrated autocorrelation times (middle) and the inverse relative efficiencies and the loss (right) in the logistic regression example of Section \ref{sec:logistic}. The vertical grey band indicates the terminal value of adaptation.}
\label{fig:logistic}
\end{figure}

\section{Discussion}
\label{sec:discussion}

We presented an i-SIR algorithm with adaptive number of proposals, based on an acceptance rate based asymptotic variance approximation. The approximate asymptotic variance has similar properties (convexity and monotonicity) as the true i-SIR asymptotic variance, and our experiments suggest that it often leads to an appropriate choice of the number of proposals. While our numerical results illustrate that in some scenarios, where $\hat{w} = \sup_x w(x)$ is large, the approximation can be poor (Experiment \ref{ex:2} in Section \ref{sec:discrete-experiments}), if the states with large $w(x)$ are unlikely, the approximation can remain useful (Experiment \ref{ex:3} in Section \ref{sec:discrete-experiments}). Our experiments on the behaviour of the adaptation method (Sections \ref{sec:mixture-example} and \ref{sec:logistic}) confirm that the method can remain useful for large $\hat{w}$, and suggest that the phenomenon observed in Experiment \ref{ex:3} can be indicative of practical scenarios.

We focused on the adaptation with an affine cost, which is assumed to be known. Our method can accommodate other forms of (smooth) cost functions, which might be appropriate for instance when run times have a large variation --- because the i-SIR selection step requires to wait for the slowest sampling process to finish. We estimated the cost based on a pilot run with a range of number of proposals, by simply fitting a linear model to the observed running times. It might be possible to develop an online strategy to learn and/or refine the cost function based on observed run times.

We used an approximate asymptotic variance based on modelling i-SIR as `lazy perfect sampling.' The approximation was derived for the non-atomic case \ref{a:non-atomic}, but our experiments suggest that it is close to an approximation specifically tailored for discrete state spaces (cf. Section \ref{sec:approximate}). While the approximation appeared to work well in many settings, it can be a bit crude, because sometimes the i-SIR chain can get stuck for longer time than anticipated by the i-SIR approximation. 
The extent of this phenomenon can be investigated empirically by comparing the discrepnacy of the realised holding times of the chain with respect to a geometric distribution. One approach to generalise our approximation and mitigate such discrepancy might be to use a similar `slowed down' perfect sampling, but in a semi-Markov manner: replacing the geometric holding times with some other, heavier tailed distribution. This could potentially account better for the excess `stickiness' due to varying weights. 
In particular, if the i-SIR algorithm is used in scenarios with unbounded weights (i.e.~beyond \ref{a:bounded-weights}), the local rejection rate $\epsilon(\lambda,x)$ is not bounded away from zero. This means that extensions such as discussed above might be necessary.
Another possible way to robustify our method could be to use a combination of our asymptotic variance proxy and a a (worst case) proxy based on the estimated weight upper bound $\hat{w}$ \citep{cardoso-samsonov-thin-moulines-olsson}.

We presented our algorithm for the simplest possible variant of the i-SIR, and could be elaborated and extended in a number of ways. For instance, it is possible to develop lower variance estimators of $\epsilon$ and $\epsilon'$ by exchangeability of $Y_k^{2:(\lfloor \lambda \rfloor+1)}$: the estimates \eqref{eq:epsilon-est} and \eqref{eq:depsilon-est} remain valid if the last item $Y_k^{\lfloor \lambda \rfloor+1}$ is changed to each of the proposals. It is possible to slightly increase the acceptance rate of the i-SIR by the forced-move variant \citep{liu,chopinsingh2015}. We believe that our method could be used also with the forced-move extension, but we are unaware of which theoretical properties remain valid.

Instead of using only the current value of the i-SIR Markov chain, it is possible to use the entire weighted sample in estimation; see for instance \citep{cardoso-samsonov-thin-moulines-olsson}. It might be possible to modify our adaptation to this setting, but using all proposals leads to different efficiency, and this must be accounted for by developing a new asymptotic variance proxy. The i-SIR can be extended to non-independent proposals \citep{mendes-scharth-kohn}; see also \citep{karppinen-vihola}, who also suggest covariance adaptation in this context. This setting also requires a different approximate asymptotic variance. Finally, adapting the proposal distribution, for instance with ideas from importance sampling \citep{bugallo-elvira-martino-luengo-miguez-djuric}, together with the number of proposals might be possible.

The approximate asymptotic variance which we use as a basis for adaptation leads to a loss function whose minimum does not depend on the test function $f$. We believe that this is often favourable (when we care about multiple test functions), but if we focus on one specific test function $f$, it may be possible to devise another function-specific strategy for instance based on asymptotic correlations (the mean square jumping distance).

\section*{Acknowledgements}

The authors were supported by the Research Council of Finland (grants 346311 \& 364216), and wish to acknowledge CSC -- IT Center for Science, Finland, for computational resources.

\bibliographystyle{abbrvnat}
\bibliography{refs}

\appendix

\section{Details about the setup}\label{app:setup}

In our setup we have a sigma-finite measure $\mu$ and probability measures $\Pi$ and $\mathbb{Q}$ on a measurable space $(\X,\mathcal{X})$ such that $\Pi \ll \mathbb{Q} \ll \mu$. The Radon--Nikodym derivatives of $\Pi$ and $\mathbb{Q}$ with respect to $\mu$ shall be denoted by $\pi$ and $q$. These derivatives are assumed to be chosen to satisfy $\pi(x),q(x)\in [0,\infty)$ for all $x\in \X$ and that $\pi(x)>0$ implies $q(x)>0$, which is always possible. The support of $\pi$ is denoted by $\Sp$.
Let us define weight function $w:\X\to [0,\infty)$,
$$
w(x)=
\begin{cases}
\frac{\pi(x)}{q(x)}, &x\in \Sp \\
0	&x\notin \Sp.
\end{cases}
$$
Then, the Radon--Nikodym derivative of $\Pi$ with respect to $\mathbb{Q}$ is $w$, that is, $\pi=wq$ $\mu$-almost everywhere.
We say that $x\in \X$ is an atom for measure $\mu$ when $\mu(\braces{x})>0$. For $\mathbb{Q}$ this condition is equivalent to $q(x)>0$. Note that $\mathbb{Q}\ll \mu$ implies that if $x$ is not an atom for $\mu$, then it is neither for $\mathbb{Q}$. However  an atom for $\mu$ does not require to be an atom for $\mathbb{Q}$. Similar conditions also hold for $\Pi$ with respect to $\mathbb{Q}$.
Note that in general $q(x)=\mathbb{Q}(\braces{x})$ does not hold for an atom of $\mathbb{Q}$, but for instance when $\mu$ restricted to the atoms of $\mathbb{Q}$ is counting measure, then $q(x)=\mathbb{Q}(\braces{x})$.
In particular, when $\mu$ is the Lebesgue measure on $\R^d$ (or counting measure on a countable $\X$), we can have the usual probability density functions (or probability mass functions) as the Radon--Nikodym derivatives $\pi$ and $q$.

\section{Asymptotic variance results}
\label{app:asymptotic-variance-results} 

Recall that the asymptotic variance of a Markov chain $P$ and function $f\in L_0^2(\pi)$ is 
\begin{equation}
  \mathrm{var}(P,f) = 2 \sum_{k=0}^\infty \angles{f|P^k f}_\pi-\|f\|_\pi^2,
\end{equation}
whenever the sum on the right is well-defined \citep[cf.][]{tierney-note}.

When the state space is finite, we have the following formulation:

\begin{lemma}\label{lemma:discrete-asymptotic-variance-form}
Let $P=[P_{i,j}]_{i,j=1,\dots,n}$ be a transition probability matrix, such that $P$ is irreducible, aperiodic and $\pi$-reversible with respect to some probability vector $\pi\in \R^n$. Let $(X_k)_{k=0}^\infty$ be a Markov chain in state space $\mathbb{S}=(s_1,\dots,s_n)$ with transition probability $P$ and initial distribution $\pi$, then it holds for $f:\mathbb{S}\to \R$ that
$$\mathrm{var}(P,\bar{f})=\sum_{i\in \braces{1,\dots,n}\setminus \braces{m}}\angles{u_i|\bar{f}}_{\pi}^2\frac{1+\lambda_i}{1-\lambda_i},$$
where $\lambda_1,\dots,\lambda_n$ are the eigenvalues and $u_1,\dots,u_n$ the corresponding eigenvectors of matrix $P$ with $\lambda_{m}=1$ for some $m\in \braces{1,\dots,n}$, $|\lambda_i|<1$ for all $i\neq m$, $u_m=(1,\dots,1)$, where the vectors $u_1,\dots,u_n$ form an orthonormal basis for $(\R^n,\angles{\uarg|\uarg}_{\pi})$.
\end{lemma}
\begin{proof}
We obtain the given eigenvalues and eigenvectors by \cite[Lemma 12.1, Lemma 12.2]{levin-peres-wilmer}. 
By these properties, $P$ is diagonalizable and
$$P^k=CD^kC^{-1},$$
where $D=\mathrm{diag}(\lambda_1,\dots,\lambda_n)$ and $C=[u_1,\dots,u_n]$.
Let us denote
$$\Pi=\mathrm{diag}(\pi_1,\dots,\pi_n),$$
then
$$\Pi^{1/2}P^k\Pi^{-1/2}=\Pi^{1/2}CD^kC^{-1}\Pi^{-1/2}=QD^kQ^{-1},$$
where $Q=\Pi^{1/2}C=[\Pi_j^{1/2}u_{ij}]_{i,j=1,\dots,n}$. We can easily verify $Q^TQ=I$, that is, $Q^T=Q^{-1}$, since 
$$\angles{u_i|u_j}_\pi=\begin{cases}
1 & i=j \\
0 & i\neq j.
\end{cases}
$$
We have
\begin{align*}
\angles{\bar{f}|P^k\bar{f}}_\pi=\bar{f}^T\Pi P^k \bar{f}=\bar{f}^T\Pi^{1/2}QD^kQ^{-1} \Pi^{1/2}\bar{f}=\bar{f}^T\Pi^{1/2}QD^kQ^{T} \Pi^{1/2}\bar{f}=e^TD^ke,
\end{align*}
where $e=Q^T\Pi^{1/2}\bar{f}$ and thus
$$e=Q^T\Pi^{1/2}\bar{f}=C^T\Pi^{1/2}\Pi^{1/2}\bar{f}=C^T\Pi\bar{f}=[\angles{u_i|\bar{f}}_\pi]_{i=1,\dots,n},$$
that is, $e_i=\angles{u_i|\bar{f}}_\pi$, where $e_m=0$, since $\bar{f}\in u_m^{\perp}$. Thus,
$$\|\bar{f}\|_\pi^2=\angles{\bar{f}|\bar{f}}_\pi=e^Te=\sum_{i=1}^n e_i^2=\sum_{i=1}^n\angles{u_i|\bar{f}}_\pi^2=\sum_{i\in \braces{1,\dots n}\setminus \braces{m}}\angles{u_i|\bar{f}}_\pi^2.$$ 
Let us define diagonal matrix $R=\mathrm{diag}(\lambda_1,\dots,\lambda_{m-1},0,\lambda_{m+1},\dots,\lambda_n)$, then
$$\sum_{k=0}^\infty R^k=\mathrm{diag}\Big(\frac{1}{1-\lambda_1},\dots,\frac{1}{1-\lambda_{m-1}},0,\frac{1}{1-\lambda_{m+1}},\dots, \frac{1}{1-\lambda_{n}}\Big),$$
and
\begin{align*}
\sum_{k=0}^\infty \angles{\bar{f}|P^k\bar{f}}_\pi=\sum_{k=0}^\infty e^TD^ke=\sum_{k=0}^\infty  e^TR^ke=\sum_{i\in\braces{1,\dots,n}\setminus  \braces{m}} \angles{u_i|\bar{f}}_\pi^2\frac{1}{1-\lambda_i}.
\end{align*}
Therefore,
\begin{align*}
\mathrm{var}(P,\bar{f}) 
&=2\sum_{k=0}^\infty \angles{\bar{f}| P^k \bar{f}}_\pi-\|\bar{f}\|_\pi^2 \\
&=\sum_{i\in\braces{1,\dots,n}\setminus  \braces{m}} \angles{u_i|\bar{f}}_\pi^2\Big[\frac{2}{1-\lambda_i}-1\Big] \\
&=\sum_{i\in\braces{1,\dots,n}\setminus  \braces{m}} \angles{u_i|\bar{f}}_\pi^2\frac{1+\lambda_i}{1-\lambda_i}.	\qedhere
\end{align*}
\end{proof}

\begin{lemma}\label{lemma:asymptotic-variance-of-finite-i-SIR}
Let $\mathbb{S}=\braces{s_1,\dots,s_n}$, $f:\mathbb{S}\to \R$ and $\lambda>1$. It holds for Algorithm \ref{alg:continuous-isir} with $\mathrm{supp}(\pi)=\mathbb{S}$ that
$$\mathrm{var}(P_\lambda,\bar{f})=\sum_{i\in \braces{1,\dots,n}\setminus \braces{m}}\angles{u_i|\bar{f}}_{\pi}^2\frac{1+\lambda_i}{1-\lambda_i},$$
where $\lambda_1,\dots,\lambda_n$ are eigenvalues and $u_1,\dots,u_n$ eigenvectors of matrix $P_\lambda$ with $\lambda_{m}=1$ for some $m\in \braces{1,\dots,n}$, $|\lambda_i|<1$ for all $i\neq m$, $u_m=(1,\dots,1)$ and the vectors $u_1,\dots,u_n$ form the orthonormal eigenbasis of $P_\lambda$ in the space $(\R^n,\angles{\uarg|\uarg}_{\pi})$.
\end{lemma}
\begin{proof}
In the finite discrete case Assumption \ref{a:bounded-weights} holds. Algorithm \ref{alg:continuous-isir} forms $\pi$-reversible Markov chain by Lemma \ref{lemma:i-sir-properties}, and $P_\lambda(x,\braces{y})>0$ for all $x,y\in \Sp$ by Lemma \ref{lemma:i-sir-probability-inequality} in Appendix \ref{app:ergodicity-auxiliaries}, hence $P_\lambda$ is irreducible and aperiodic. Therefore, the claim follows from Lemma \ref{lemma:discrete-asymptotic-variance-form}.
\end{proof}

\section{Results about real sequences and their sums and differences}\label{app:differences}
The following lemma is well known and is easily verified by Cauchy--Schwarz:
\begin{lemma}[Sedrakyan's inequality]\label{lemma:sedrakyan-inequality}
Let $n\in \N_+$, $a_1,\dots,a_n\in \R$ and $b_1,\dots,b_n>0$, then
$$\sum_{k=1}^{n}\frac{a_k^2}{b_k}\ge \frac{(\sum_{k=1}^{n}a_k)^2}{\sum_{l=1}^{n}b_l},$$
where the equality holds if and only if $(a_1,\dots,a_n)$ and $(b_1,\dots,b_n)$ are linearly dependent.
\end{lemma}

The following two lemmas use the following basic property: If $n\in \N_+$ and $c_1,\dots,c_{n+1}\in \R$, then
\begin{equation}\label{equation:sequence-sum-equality}
\sum_{k=1}^{n+1}\sum_{j=1,j\neq k}^{n+1}c_j=\sum_{k=1}^{n+1}\sum_{j=1}^{n+1}c_j \charfun{j\neq k}=\sum_{j=1}^{n+1}c_j \sum_{k=1}^{n+1} \charfun{j\neq k}=n\sum_{j=1}^{n+1}c_j. 
\end{equation}

\begin{lemma}\label{lemma:corollary-of-sedrakyan}
Let $n\in \N_+$, $a_1,\dots,a_{n+1}\in \R$, $b_1,\dots,b_{n+1}\ge 0$, $\sum_{j=1}^{n+1} b_j>0$, then
$$
\frac{1}{n}\sum_{k\in J}\frac{\big(\sum_{i=1,i\neq k}^{n+1}b_ia_i\big)^2}{\sum_{j=1,j\neq k}^{n+1}b_j}\ge 
\frac{\big(\sum_{i=1}^{n+1}b_ia_i\big)^2}{\sum_{j=1}^{n+1} b_j},
$$
where $J=\bigbraces{k\in \braces{1,2,\dots,{n+1}}:\sum_{i=1,i\neq k}^{n+1}b_i>0}$. Equality holds if and only if $A=(A_1,\dots,A_{n+1})$ and $B=(B_1,\dots,B_{n+1})$ are linearly dependent,
where $A_k=\sum_{j=1,j\neq k}^{n+1}b_ja_j$ and $B_k=\sum_{j=1,j\neq k}^{n+1} b_j$.
\end{lemma}
\begin{proof}
By Sedrakyan's inequality,
\begin{align*}
\frac{1}{n}\sum_{k\in J}\frac{A_k^2}{B_k} 
&\ge  \frac{1}{n}\frac{\big(\sum_{k\in J}A_k\big)^2}{\sum_{l\in J}B_l}
=  \frac{1}{n}\frac{\big(\sum_{k=1}^{n+1}A_k\big)^2}{\sum_{l=1}^{n+1}B_l},
\end{align*}
for which equality holds if and only if $(A_k)_{k\in J}$ and $(B_k)_{k\in J}$ are linearly dependent which is equivalent to $A$ and $B$ being linearly dependent, since $A_k=0=B_k$ for $k\notin J$. Therefore, the claim follows from \eqref{equation:sequence-sum-equality}.
\end{proof}

\begin{lemma}\label{lemma:linearly-independent-condition}
Let $n\in \N_+$, $a_1,\dots,a_{n+1}\in \R$, $b_1,\dots,b_{n+1}\ge 0$, $A_k=\sum_{j=1,j\neq k}^{n+1}b_ja_j$ and $B_k=\sum_{j=1,j\neq k}^{n+1} b_j$. Then, $A=(A_1,\dots,A_{n+1})$ and $B=(B_1,\dots,B_{n+1})$ are linearly independent if $a_i\neq a_j$ and $b_i,b_j>0$ for some $i,j\in \braces{1,2,\dots,n+1}$.
\end{lemma}
\begin{proof}
Assume $A$ and $B$ are linearly dependent, that is, $cA+B=0$ for some $c\neq 0$, then by \eqref{equation:sequence-sum-equality}
$$0=\sum_{k=1}^{n+1}(cA_k+B_k)=n\sum_{l=1}^{n+1}b_l(ca_l+1)$$
and thus it holds for all $k\in \braces{1,2,\dots,n+1}$ that
$$b_k(ca_k+1)=\sum_{l=1}^{n+1}b_l(ca_l+1)-\sum_{m=1,m\neq k}^{n+1}b_m(ca_m+1)=-(cA_k+B_k)
=0,$$
hence for each $k$ has to hold either $b_k=0$ or $ca_k+1=0$.
Since $a_i\neq a_j$ and $b_i,b_j>0$, we get a contradiction.
\end{proof}

The final result is about differences of partial sums of sequences.

\begin{lemma}\label{lemma:differences}
Let $N \in \N_+$, $a_1,a_2,\dots \in \R$, $b_1,b_2,\dots\in \R$, $A_n=\sum_{i=1}^n a_i$ and $B_n=\sum_{i=1}^n b_i$, then
\begin{align*}
B_{N+1}B_{N+2}-2B_{N}B_{N+2}+B_{N}B_{N+1}=(b_{N+1}-b_{N+2})B_{N+1}+2b_{N+1}b_{N+2}
\end{align*}
and
\begin{align*}
&B_{N+1}B_{N+2}A_{N}-2B_{N}B_{N+2}A_{N+1}+B_{N}B_{N+1}A_{N+2} \\
&=2b_{N+2}(B_{N+1}A_{N}-B_{N}A_{N+1})-B_{N+1}B_{N}(a_{N+1}-a_{N+2})+B_{N+1}A_{N}(b_{N+1}-b_{N+2}).
\end{align*}
\end{lemma}
\begin{proof}
These are straightforward, since $B_{n+1}=B_n+b_{n+1}$ and $A_{n+1}=A_n+a_{n+1}$:
\begin{align*}
B_{N+1}B_{N+2}-2B_{N}B_{N+2}+B_{N}B_{N+1} 
&=B_{N+2}b_{N+1}-B_{N}b_{N+2}\\
&=B_{N+1}(b_{N+1}-b_{N+2})+2b_{N+1}b_{N+2}
\end{align*}
and thus
\begin{align*}
&B_{N+1}B_{N+2}A_{N}-2B_{N}B_{N+2}A_{N+1}+B_{N}B_{N+1}A_{N+2} \\
&=B_{N+1}A_{N}(b_{N+1}-b_{N+2})+2b_{N+1}b_{N+2}A_{N} \\
&\quad-2B_{N}B_{N+2}a_{N+1} +B_{N}B_{N+1}(a_{N+1} +a_{N+2})  \\
&=B_{N+1}A_{N}(b_{N+1}-b_{N+2})+2(B_{N+1}-B_{N})b_{N+2}A_{N}   \\
&\quad-2B_{N}(B_{N+1}+b_{N+2})a_{N+1}+B_{N}B_{N+1}a_{N+1} +B_{N}B_{N+1}a_{N+2}  \\
&=B_{N+1}A_{N}(b_{N+1}-b_{N+2})+2B_{N+1}b_{N+2}A_{N} -2B_{N}b_{N+2}A_{N+1}  \\
&\quad-B_{N}B_{N+1}(a_{N+1} -a_{N+2}).		\qedhere
\end{align*}
\end{proof}

\section{Technical results for the monotonicity and convexity}\label{app:monotonicity-convexity}
Here we present results related to sequential monotonicity and convexity of i-SIR covariances.

\begin{lemma}\label{lemma:i-sir-exchangeable-property}
Let $N,M\in \N_+$, $f\in L^2(\pi)$ and $h:\W{N}\times \X^M\to [0,\infty)$ be measurable bounded function and symmetric with respect to the first $N$ arguments, then
\begin{align*}
&\int_{\W{N}\times \X^M} h(z)\sum_{i=1}^N \frac{w(z_1)w(z_i)f(z_1)f(z_i)}{\sum_{j=1}^N w(z_j)} \bigg( \prod_{n=1}^{N+M} q(\ud z_n) \bigg) \\
&\quad=\int_{\W{N}\times \X^M} h(z)\frac{1}{N}\frac{\Big(\sum_{i=1}^Nw(z_i)f(z_i)\Big)^2}{\sum_{j=1}^N w(z_j)} \bigg( \prod_{n=1}^{N+M} q(\ud z_n) \bigg).
\end{align*}
\end{lemma}
\begin{proof}
Function $h$ is symmetric with respect to the first $N$ arguments, hence the expression on the left can be written as
\begin{align*}
\int_{\W{N}\times \X^M} h(z)\frac{1}{N}\sum_{i=1}^N\sum_{k=1}^N \frac{w(z_i)w(z_k)f(z_i)f(z_k)}{\sum_{j=1}^N w(z_j)} \bigg( \prod_{n=1}^{N+M} q(\ud z_n) \bigg)
\end{align*}
by symmetricity of the integrand with respect to $z_1,\dots,z_{N}$, which equals the expression on the right.
\end{proof}

The following lemma connects expressions similar to Lemma \ref{lemma:inner-product-another-representation} for different $N$:

\begin{lemma}\label{lemma:i-sir-inequality}
Let $N\in \N_+$, $f\in L^2(\pi)$ and let $h:\W{N+1}\times \X\to [0,\infty)$ be measurable bounded function and symmetric with respect to the first $N+1$ arguments, then
\begin{align}\label{eq:i-sir-inequality-1}
\begin{split}
&\frac{1}{N}\int_{\W{N}\times \X^2} h(z)\frac{\Big(\sum_{i=1}^{N}w(z_i)f(z_i)\Big)^2}{\sum_{j=1}^{N} w(z_j)}\bigg( \prod_{n=1}^{N+2} q(\ud z_n) \bigg)\\
&\quad \ge \frac{1}{N+1}\int_{\W{N+1}\times \X}h(z)\frac{\Big(\sum_{i=1}^{N+1}w(z_i)f(z_i)\Big)^2}{\sum_{j=1}^{N+1} w(z_j)} \bigg( \prod_{n=1}^{N+2} q(\ud z_n) \bigg).
\end{split}
\end{align}
If $\bar{f}\in L_{0,+}^2(\pi)$ and there exists a $\pi$-positive set $A$ such that $h(z)>0$ for every $z\in \W{N+1}\times A$, then the inequality is strict.
\end{lemma}
\begin{proof}
Let us define $g_{k}:\W{N+1}\times \X\to [0,\infty)$,
$$
g_{k}(z_1,\dots,z_{N+2})=
\begin{cases}
\frac{\big(\sum_{i=1,i\neq k}^{N+1}w(z_i)f(z_i)\big)^2}{\sum_{j=1,j\neq k}^{N+1} w(z_j)}, &\sum_{l=1,l\neq k}^{N+1}w(z_l)>0 \\
0, &\sum_{l=1,l\neq k}^{N+1}w(z_l)=0,
\end{cases}
$$
for $k\in \braces{1,2,\dots,N+2}$. We have
\begin{align}\label{eq:i-sir-inequality-2}
\int_{B}\frac{h(z)}{(N+1)N}\sum_{k=1}^{N+1}g_{k}(z) \bigg( \prod_{n=1}^{N+2} q(\ud z_n) \bigg) 
\ge \int_{B}\frac{h(z)}{N+1}g_{N+2}(z) \bigg( \prod_{n=1}^{N+2} q(\ud z_n) \bigg)
\end{align}
for any $B\subset \W{N+1}\times \X$ that is measurable with respect to the product sigma-algebra by Lemma \ref{lemma:corollary-of-sedrakyan} in Appendix \ref{app:differences}. When $B=\W{N+1}\times \X$, the right-hand side of \eqref{eq:i-sir-inequality-2} equals the right-hand side of \eqref{eq:i-sir-inequality-1}, and because $h$ is symmetric with respect to the first $N+1$ arguments, we have
\begin{align*}
&\int_{\W{N+1}\times \X}\frac{h(z)}{(N+1)N}\sum_{k=1}^{N+1}g_{k}(z) \bigg( \prod_{n=1}^{N+2} q(\ud z_n) \bigg)  \\
&\quad= \int_{\W{N+1}\times \X}\frac{h(z)}{N}g_{N+1}(z) \bigg( \prod_{n=1}^{N+2} q(\ud z_n) \bigg) \\
&\quad = \int_{\W{N}\times \X^2} \frac{h(z)}{N}\frac{\Big(\sum_{i=1}^{N}w(z_i)f(z_i)\Big)^2}{\sum_{j=1}^{N} w(z_j)}\bigg( \prod_{n=1}^{N+2} q(\ud z_n) \bigg),
\end{align*}
where the first equality follows from symmetricity of the integrand with respect to $z_1,\dots,z_{N+1}$, and the second from $\W{N}\times \X\subset \W{N+1}$.

If $\bar{f}\in L_{0,+}^2(\pi)$, then there exist two disjoint $\pi$-positive sets $U,V\subset \Sp$ such that $f$ differs on them, these sets and set $A$ in the assumption are $q$-positive. Let $W=U\times V\times \Sp^{N-1}\times A$ (for $N=1$, $W=U\times V\times A$). If $B=W$, we get strict inequality for \eqref{eq:i-sir-inequality-2} by Lemma \ref{lemma:corollary-of-sedrakyan} and \ref{lemma:linearly-independent-condition} in Appendix \ref{app:differences} and thus strict inequality also holds when $B=\W{N+1}\times \X=W\cup ((\W{N+1}\times \X)\setminus W)$, from which the claim follows.
\end{proof}

\section{Details about interpolation}\label{app:modification}
In the following lemma the assumption in case \eqref{enum:generalisation-1} is generalisation of sequential convexity and matches this if $a_{i}=i$ for all $i$.

\begin{lemma}\label{lemma:generalisation-to-continuous-function}
Let $(a_i)_{i\in \mathbb{Z}}\in \R$ be strictly increasing and let us denote $d_i=a_{i+1}-a_i$. Let $(b_i)_{i\in \Z}\in \R$ and $f:\R\to \R$,
\begin{align*}
f(x)=\Big(1-\frac{x-a_{g(x)}}{a_{g(x)+1}-a_{g(x)}}\Big)b_{g(x)}+\frac{x-a_{g(x)}}{a_{g(x)+1}-a_{g(x)}}b_{g(x)+1},
\end{align*}
where
$$
g(x)=\sup\braces{i\in \Z:x\ge a_i}.
$$
The function $f$ is continuous and for all $i\in \Z$ holds $f(a_i)=b_i$.
Additionally, if for all $n\in \Z$ holds
\begin{enumerate}[(i)]
\item $\frac{d_{n+1}}{d_n+d_{n+1}}b_n-b_{n+1}+\frac{d_n}{d_n+d_{n+1}}b_{n+2}\ge 0$ then $f$ is convex (if the reverse inequality holds, then $f$ is concave). \label{enum:generalisation-1}
\item $b_n-b_{n+1}\ge  0$ $(b_n-b_{n+1}>0)$, then $f$ is decreasing (strictly). \label{enum:generalisation-2}
\item $b_n-b_{n+1}\le  0$ $(b_n-b_{n+1}<0)$, then $f$ is increasing (strictly). \label{enum:generalisation-3}
\item $b_n\ge 0$ $(b_n>0)$, then $f$ is non-negative (positive). \label{enum:generalisation-4}
\item $b_n\le 0$ $(b_n<0)$, then $f$ is non-positive (negative). \label{enum:generalisation-5}
\end{enumerate}
\end{lemma}
\begin{proof}
Clearly, $f$ is continuous on the intervals $[a_i,a_{i+1})$, $f(a_i)=b_i$ and $\lim_{x\uparrow a_{i+1}}f(x)=b_{i+1}$ for all $i\in \Z$, hence $f:\R\to \R$ is continuous.

Let us assume \eqref{enum:generalisation-1} (convex version) holds, then
\begin{align*}
b_i&\ge \frac{d_{i-2}+d_{i-1}}{d_{i-2}}b_{i-1}-\frac{d_{i-1}}{d_{i-2}}b_{i-2} \\
&\ge \frac{d_{i-3}+d_{i-2}+d_{i-1}}{d_{i-3}}b_{i-2}-\frac{d_{i-2}+d_{i-1}}{d_{i-3}}b_{i-3}.
\end{align*}
Thus, by induction for any $m\in \N$, $m\ge 2$, it holds that
\begin{align}\label{eq:help-convex-1}
\begin{split}
b_i&\ge \frac{d_{i-m}+\dots+d_{i-1}}{d_{i-m}}b_{i+1-m}-\frac{d_{i+1-m}+\dots+d_{i-1}}{d_{i-m}}b_{i-m} \\
&=\frac{(a_{i}-a_{i-m})b_{i+1-m}-(a_{i}-a_{i+1-m})b_{i-m}}{a_{i+1-m}-a_{i-m}}
\end{split}
\end{align}
and in similar fashion, we have
\begin{align}\label{eq:help-convex-2}
\begin{split}
b_i&\ge \frac{(a_{i+m}-a_{i})b_{i+m-1}-(a_{i+m-1}-a_{i})b_{i+m}}{a_{i+m}-a_{i+m-1}}.
\end{split}
\end{align}
Let us denote $\angles{x|y}=x_1y_1+x_2y_2$ for $x,y\in \R^2$, let
$$J(x,y)=\braces{(1-t)x+ty\in \R^2:t\in [0,1]}$$
be the line segment between points $x$ and $y$, and let $w^i=(a_i,b_i) \in \R^2$. For every $x\in [a_i,a_{i+1}]$, we have clearly that $(x,f(x))\in J(w^i,w^{i+1})$. 
Let $y,z\in \R$, $y<z$ and $t\in (0,1)$, then 
$$c_t:=(1-t)y+tz=(1-s)a_k+sa_{k+1}$$
for some $s\in [0,1)$ and $k\in \Z$. The vector $v=(b_{k+1}-b_{k},a_{k}-a_{k+1})$ is perpendicular to the line segment $J(w^k,w^{k+1}) \ni (c_t,f(c_t))$. Let us define convex half space that contains this line segment on its boundary
$$A=\braces{u\in \R^2:\angles{u|v}\le \angles{w^k|v}}.$$
Clearly, $\angles{w^k|v}=a_kb_{k+1}-a_{k+1}b_k=(w^{k+1}|v)$, and by \eqref{eq:help-convex-1}, when we choose $m=i-k$, we have for all $i>k+1$ that
$$\angles{w^i|v}\le a_i(b_{k+1}-b_{k})+\frac{(a_{i}-a_{k})b_{k+1}-(a_{i}-a_{k+1})b_{k}}{a_{k+1}-a_{k}}(a_{k}-a_{k+1})=(w^k|v).$$
Similarly, by \eqref{eq:help-convex-2}, when we choose $m=k+1-i$, we have for all $i<k$ that $\angles{w^i|v}\le \angles{w^k|v}$. Together these imply $w^i\in A$ for every $i\in \Z$. 
If $x<f(c_t)$, then $(c_t,x)\notin A$, since
\begin{align*}
\angles{(c_t,x)|v}&=(b_{k+1}-b_{k})c_t+(a_{k}-a_{k+1})x 
>(b_{k+1}-b_{k})c_t+(a_{k}-a_{k+1})f(c_t),
\end{align*}
where the last expression equals
$$\bigangles{\big(c_t,f(c_t)\big)\big|v}=\angles{w^k|v}.$$
Therefore, if $x\in \R$ and $(c_t,x)\in A$, then it holds that $x\ge f(c_t)$.
Since $A$ is convex, we have $(y,f(y)),(z,f(z))\in A$ and $((1-t)y+tz,(1-t)f(y)+tf(z))\in A$, hence we must have $(1-t)f(y)+tf(z)\ge f((1-t)y+tz)$. Thus, we obtain convexity. Concavity follows similarly.

Assume $b_n\ge b_{n+1}$ and $x<y$. If $a_k\le x<y<a_{k+1}$ for some $k$, then
$$f(x)-f(y)=\frac{y-x}{a_{k+1}-a_{k}}b_{k}+\frac{x-y}{a_{k+1}-a_{k}}b_{k+1}\ge 0$$
and $f(y)\ge f(a_{k+1})$ is obvious, thus $f$ is decreasing. Similarly, we obtain strict inequalities if $b_i>b_{i+1}$, which proves \eqref{enum:generalisation-2}. We obtain \eqref{enum:generalisation-3} similarly.
Both \eqref{enum:generalisation-4} and \eqref{enum:generalisation-5} are clear by the definition of $f$.
\end{proof}

\begin{lemma}\label{lemma:upper-bound-for-functions-of-certain-type}
Let $1\le a<\infty$. If $g:\N_+\to \R$ is a function such that for all $N\in \N_+$
$$0\le g(N)\le \frac{2a-1}{2a+N-2},$$
then for $h:[1,\infty)\to \R$,
$$h(\lambda)=(\floor{\lambda}+1-\lambda)g(\floor{\lambda})+(\lambda-\floor{\lambda})g(\floor{\lambda}+1),$$
it holds for all $\lambda\ge 1$ that
$$0\le h(\lambda)\le \frac{2a}{2a+\lambda-1}.$$
\end{lemma}
\begin{proof}
If $\lambda\ge 1$, then $\lambda\in [N,N+1)$ for some $N\in \N_+$, hence
\begin{align*}
h(\lambda)&=(N+1-\lambda) g(N)+(\lambda-N)g(N+1) \\
&\le (N+1-\lambda)\frac{2a-1}{2a+N-2}+(\lambda-N)\frac{2a-1}{2a+N-1},
\end{align*}
where the last expression equals
\begin{equation*}
\frac{2a-1}{2a+N-2}\frac{2a+2N-\lambda-1}{2a+N-1} 
\le \frac{2a}{2a+N-1}\frac{2a+N-1}{2a+\lambda-1} =\frac{2a}{2a+\lambda-1}. \qedhere
\end{equation*}
\end{proof}

\begin{lemma}\label{lemma:extended-function}
Let $b> -1$ and $f:\N_+\to \R$,
$f(N)=(b+N)^{-1} $,
then for $k:[1,\infty)\to \R,$
$$k(\lambda)=(\floor{\lambda}+1-\lambda)f(\floor{\lambda})+(\lambda-\floor{\lambda})f(\floor{\lambda}+1),$$
it holds for all $\lambda\ge 1$ that
\begin{align*}
k(\lambda)&=\frac{b+2\floor{\lambda}-\lambda+1}{(b+\floor{\lambda}+1)(b+\floor{\lambda})}.
\end{align*}
\end{lemma}
\begin{proof}
We have
\begin{align*}
k(\lambda)&=(\floor{\lambda}+1-\lambda)\frac{1}{b+\floor{\lambda}}+(\lambda-\floor{\lambda})\frac{1}{b+\floor{\lambda}+1} \\
&=(\floor{\lambda}+1-\lambda)\frac{b+\floor{\lambda}+1}{(b+\floor{\lambda}+1)(b+\floor{\lambda})}+(\lambda-\floor{\lambda})\frac{b+\floor{\lambda}}{(b+\floor{\lambda}+1)(b+\floor{\lambda})} \\
&=\frac{b+\floor{\lambda}+1}{(b+\floor{\lambda}+1)(b+\floor{\lambda})}+\frac{\floor{\lambda}-\lambda}{(b+\floor{\lambda}+1)(b+\floor{\lambda})} \\
&=\frac{b+2\floor{\lambda}-\lambda+1}{(b+\floor{\lambda}+1)(b+\floor{\lambda})}.	\qedhere
\end{align*}
\end{proof}

\section{Results for the ergodicity of i-SIR}
\label{app:ergodicity-auxiliaries}
Here we list some properties related to the transition probability of Algorithm \ref{alg:continuous-isir}. The total variation distance between two probabilities $\mu$ and $\nu$ is denoted by
$$d_\tv(\mu,\nu) = \frac{1}{2}\sup_{\|f\|_\infty \le 1} |\mu(f)-\nu(f)| = \sup_{A} | \mu(A) - \nu(A)|.$$

\begin{lemma}\label{lemma:transition-probability-inequality}
Let $x\in \Sp$ and $M,N\in \N_+$, $M<N$, then
\begin{align*}
|P_M(x,A)-P_{N}(x,A)|\le \frac{N-M}{N-1}.
\end{align*}
\end{lemma}
\begin{proof}
Let us denote $z_1=x$, $B_n=\sum_{i=1}^nw(z_i)\charfun{z_i\in A}$ and $C_n=\sum_{i=1}^nw(z_i)$, then
\begin{align*}
&|P_M(x,A)-P_{N}(x,A)| \\
&\quad=\bigg|\int_{\X^{N-1}}\frac{B_MC_N-B_NC_M}{C_MC_N}\bigg( \prod_{n=2}^{N}q(\ud z_n)\bigg)\bigg| \\
&\quad=\bigg|\int_{\X^{N-1}}\frac{\sum_{i=1}^M\sum_{j=M+1}^Nw(z_i)w(z_j)[\charfun{z_i\in A}-\charfun{z_j\in A}]}{C_MC_N}\bigg( \prod_{n=2}^{N}q(\ud z_n)\bigg)\bigg| \\
&\quad\le \int_{\X^{N-1}}\frac{\sum_{i=1}^M\sum_{j=M+1}^Nw(z_i)w(z_j)}{C_MC_N}\bigg( \prod_{n=2}^{N}q(\ud z_n)\bigg),
\end{align*}
where the last expression equals
\begin{align*}
 \int_{\X^{N-1}}\frac{\sum_{j=M+1}^Nw(z_j)}{C_N}\bigg( \prod_{n=2}^{N}q(\ud z_n)\bigg) 
&=\frac{N-M}{N-1}\int_{\X^{N-1}}\frac{\sum_{i=2}^{N}w(z_i)}{C_N}\bigg( \prod_{n=2}^{N}q(\ud z_n)\bigg)
 \le \frac{N-M}{N-1},
\end{align*}
where the equality follows from symmetricity of the integrand with respect to $z_2,\dots,z_{N}$.
\end{proof}

\begin{lemma}\label{lemma:auxiliary-1}
It holds for all $\lambda,\lambda'\in [1,\infty)$ and $x\in \Sp$ that
$$
     d_\tv\big(P_\lambda(x,\uarg), P_{\lambda'}(x,\uarg)\big) \le \frac{6| \lambda - \lambda'|}{\max\braces{\lambda,\lambda'}}.
$$
\end{lemma}
\begin{proof}
Let $N,N'\in \N_+$. 
If $N\le \lambda < \lambda' \le N+1$, then $1/N\le 2/\lambda'$ and
\begin{align*}
d_\tv\big( P_{\lambda'}(x,\uarg) ,P_{\lambda}(x,\uarg) \big)&=\sup_A|P_{\lambda'}(x,A)-P_{\lambda}(x,A)| \\
& =\sup_A|(\lambda-\lambda')P_{N}(x,A)+(\lambda'-\lambda)P_{N+1}(x,A)| \\
&=|\lambda'-\lambda|\sup_A|P_{N+1}(x,A)-P_{N}(x,A)|  \\
&\le \frac{1}{N}|\lambda'-\lambda|,
\end{align*}
where the last inequality follows from Lemma \ref{lemma:transition-probability-inequality}.
If $N<\lambda \le N +1\le N'<\lambda' \le N'+1$, then by the previous inequality, Lemma \ref{lemma:transition-probability-inequality} and triangle inequality
\begin{align*}
&d_\tv\big( P_{N'}(x,\uarg), P_{N}(x,\uarg)\big) \\
&\quad \le d_\tv\big(P_{\lambda'}(x,\uarg),P_{N'}(x,\uarg) \big)+d_\tv\big(P_{N'}(x,\uarg),P_{N+1}(x,\uarg)\big)+d_\tv\big( P_{N+1}(x,\uarg), P_{\lambda}(x,\uarg)\big) \\
&\quad \le \big(\lambda'-N'\big)\frac{1}{N'}+\big(N'-(N+1)\big)\frac{1}{N'-1}+\big(N+1-\lambda\big)\frac{1}{N} \\
&\quad\le \frac{\lambda'-(N+1)}{N'-1}+\frac{N+1-\lambda}{N},
\end{align*}
where the last expression equals
$$
\frac{N(\lambda'-\lambda)+(N'-1-N)(N+1-\lambda)}{N(N'-1)} 
\le \frac{N(\lambda'-\lambda)+N'-(N+1)}{N(N'-1)} 
\le \frac{N(\lambda'-\lambda)+\lambda'-\lambda}{N(N'-1)},
$$
where
$$
\frac{N(\lambda'-\lambda)+\lambda'-\lambda}{N(N'-1)}=\frac{(N+1)(\lambda'-\lambda)}{N(N'-1)}\le \frac{2(\lambda'-\lambda)}{N'-1}\le\frac{6(\lambda'-\lambda)}{N'+1}\le \frac{6(\lambda'-\lambda)}{\lambda'},
$$
where the second last inequality is obvious since $N'\ge 2$.
Therefore, we obtain the claim.
\end{proof}

\begin{lemma}\label{lemma:i-sir-probability-inequality}
If Assumption \ref{a:bounded-weights} holds, then it holds for all $x\in \Sp$ and $\lambda\ge 1$ that
\begin{align*}
\frac{\lambda-1}{2\hat{w}+\lambda-1}\pi(A)&\le P_\lambda(x,A)\le \pi(A)+\pi(A\C)\frac{2\hat{w}}{2\hat{w}+\lambda-1}.
\end{align*}
\end{lemma}
\begin{proof}
Let us denote $z_1=x$. Let $N\in \N_+$, $N\ge 2$, then by dropping the first term $i=1$ from \eqref{def:isir-transition-probability} an application of Jensen's inequality yields \cite[e.g.][]{andrieu-lee-vihola}
\begin{align*}
P_N(z_1,A)
&\ge \frac{N-1}{2\hat{w}+N-2}\pi(A),
\end{align*}
which also holds for $N=1$, since $P_1(z_1,A)\ge 0$. Let $N\in \N_+$. If $N\le \lambda<N+1$, then by the definition of $P_\lambda$ and the previous inequality, we have
\begin{align*}
P_\lambda(z_1,A)&\ge (N+1-\lambda)\pi(A)\frac{N-1}{2\hat{w}+N-2}+(\lambda-N)\pi(A)\frac{N}{2\hat{w}+N-1} \\
&\ge \pi(A)(N+1-\lambda)\frac{N-1}{2\hat{w}+N-1}+\pi(A)(\lambda-N)\frac{N}{2\hat{w}+N-1} \\
&= \pi(A)\frac{\lambda-1}{2\hat{w}+N-1} \\
&\ge \pi(A)\frac{\lambda-1}{2\hat{w}+\lambda-1}.
\end{align*}
Let $N\in \N_+$, $N\ge 2$, then
\begin{align*}
P_N(z_1,A) &=\int_{\X^{N-1}}\frac{\sum_{i=1}^Nw(z_i)\charfun{z_i\in A}}{\sum_{k=1}^Nw(z_k)} \prod_{n=2}^Nq(\ud z_n)  \\
&=1-\int_{\X^{N-1}}\frac{\sum_{i=1}^Nw(z_i)\charfun{z_i\notin A}}{\sum_{k=1}^Nw(z_k)} \prod_{n=2}^Nq(\ud z_n) \\
&\le 1-\int_{\X^{N-1}}\frac{\sum_{i=2}^Nw(z_i)\charfun{z_i\notin A}}{\sum_{k=1}^Nw(z_k)} \prod_{n=2}^Nq(\ud z_n),
\end{align*}
where the last expression, by symmetricity of the integrand with respect to $z_2,\dots,z_N$, equals
\begin{align*}
1-(N-1)\int_{\X^{N-1}}\frac{w(z_2)\charfun{z_2\notin A}}{\sum_{k=1}^Nw(z_k)} \prod_{n=2}^Nq(\ud z_n)
&\le 1-(N-1)\int_{\X}\frac{w(z_2)\charfun{z_2\notin A}}{w(z_1)+w(z_2)+N-2} q(\ud z_2) \\
&\le 1-(N-1)\int_{\X}\frac{w(z_2)\charfun{z_2\notin A}}{2\hat{w}+N-2} q(\ud z_2) \\
&= 1-(N-1)\frac{\pi(A\C)}{2\hat{w}+N-2} \\
&= \pi(A)+\pi(A\C)\frac{2\hat{w}-1}{2\hat{w}+N-2},
\end{align*}
where the first inequality follows from Jensen's inequality. The case $N=1$ is clear. It follows by Lemma \ref{lemma:upper-bound-for-functions-of-certain-type} in Appendix \ref{app:modification}, that for all $\lambda\ge 1$
\begin{equation*}
P_\lambda(z_1,A)\le \pi(A)+\pi(A\C)\frac{2\hat{w}}{2\hat{w}+\lambda-1}.	\qedhere
\end{equation*}
\end{proof}

\begin{lemma}\label{lemma:auxiliary-2}
Let $s>1$. If Assumption \ref{a:bounded-weights} holds, then there exist constants $C<\infty$ and $\rho\in[0,1)$ such that for all $x\in \Sp$, $\lambda \in [s,\infty)$ and $k\ge 0$ holds
$$
d_\tv(P_\lambda^k(x, \uarg), \pi)
\le C \rho^k.
$$
\end{lemma}
\begin{proof}
By the lower bound in Lemma \ref{lemma:i-sir-probability-inequality}, we have for all $x\in \Sp$ and $\lambda\ge s$ a constant $c=\frac{s-1}{2\hat{w}+s-1}\in (0,1)$ such that
\begin{align*}
P_\lambda(x,A)=c\pi(A)+(1-c)Q_\lambda(x,A),
\end{align*}
where $Q_\lambda$ is probability kernel.
For probability kernel $R$ and signed measure $\nu$ on a measurable space $(\X,\mathcal{X})$, we have by Hahn decomposition theorem $\nu$-positive set $B^+$ and $\nu$-negative set $B^-$ for which $\X=B^+\cup B^-$, and if $G\in \mathcal{X}$, then
\begin{align*}
|\nu R(G)|&=|\int\nu(\ud x) R(x,G)| \\
&=|\int_{B^+}\nu(\ud x) R(x,G)+\int_{B^-}\nu(\ud x) R(x,G)| \\
&\le \max\Bigbraces{\Big|\int_{B^+}\nu(\ud x) R(x,G)\Big|,\Big|\int_{B^-}\nu(\ud x) R(x,G)\Big|} \\
&\le \max\braces{|\nu(B^+)|,|\nu(B^-)|} \\
&=\sup_{A}|\nu(A)|
\end{align*}
and thus
$$\sup_{A}|\nu R(A)|\le \sup_{A}|\nu(A)|,$$
which ensures for $x\in \Sp$ and $k\ge 1$ with the help of $\pi$-reversibility of $P_\lambda$ (Lemma \ref{lemma:i-sir-properties}) that
\begin{align*}
\sup_A|P_\lambda^k(x,A)-\pi(A)|
&=\sup_A|\delta_x P_\lambda^k(A)-\pi(A)| \\
&= \sup_A|\delta_x P_\lambda^{k-1} P_\lambda(A)-\pi P_\lambda^{k-1}P_\lambda(A)| \\
&=\sup_A|c\pi(A)+(1-c)\delta_x P_\lambda^{k-1}Q_\lambda(A)-c\pi(A)-(1-c)\pi P_\lambda^{k-1}Q_\lambda(A)| \\
&=(1-c)\sup_A|\delta_x P_\lambda^{k-1}Q_\lambda(A)-\pi P_\lambda^{k-1}Q_\lambda(A)| \\
&\le (1-c)\sup_A|\delta_x P_\lambda^{k-1}(A)-\pi P_\lambda^{k-1}(A)|,
\end{align*}
and iterating this bound yields
\begin{align*}
\sup_A|P_\lambda^k(x,A)-\pi(A)|
&\le (1-c)^k\sup_A|\delta_x Q_\lambda(A)-\pi Q_\lambda(A)| \\
&\le (1-c)^k\sup_A|\delta_x(A)-\pi(A)| \\
&\le (1-c)^k.	\qedhere
\end{align*}
\end{proof}

\section{Properties of the approximate i-SIR}
\label{app:approximate-properties}

\begin{lemma}\label{lemma:single-proposal-chain-asymptotic-variance}
   Let $\pi$ be a probability distribution on a general state space $E$, and let $\epsilon\in(0,1)$. Consider the Markov transition probability:
   $$
      P_\epsilon(x,A) = \epsilon \charfun{x\in A} + (1-\epsilon)\pi(A).
   $$
   The asymptotic variance of a function $f\in L_0^2(\pi)$ is:
   $$
      \mathrm{var}(P_\epsilon, f) = \frac{1 + \epsilon}{1-\epsilon} \|f\|_\pi^2
   $$
\end{lemma}
\begin{proof}
The claim is obvious, if $\|f\|_\pi^2=0$, hence we can assume $\|f\|_\pi^2>0$.
Let us have i.i.d. random variables $U_1,\dots,U_n\sim \mathrm{Unif}(0,1)$ and $X_0,Z_1,\dots,Z_n\sim \pi$ that are independent of each other. Let
$$X_n=X_{n-1}I(U_n\le \epsilon)+Z_{n}I(U_n> \epsilon)$$
for $n\ge 1$, hence $X_n\sim \pi$. Let us standardise the function $f$ and denote $g(x)=\frac{f(x)}{\sqrt{\|f\|_\pi^2}}$. We have for all $n\ge 1$ that
\begin{align*}
g(X_n)=g(X_{n-1})\charfun{U_n\le \epsilon}+g(Z_n)\charfun{U_n>\epsilon},
\end{align*}
hence $g(X_n)$ are identically distributed, and it holds for all $m\ge 0$ that $\E[g(X_m)]=0$ and $\E[g(X_m)^2]=1$.
By induction, we get for $0\le m<n$ that
\begin{align*}
g(X_n)=g(X_{m})\charfun{R_{m+1:n}}+\charfun{n-m\ge 2}\sum_{i=m+1}^{n-1}g(Z_i)\charfun{A_{i} \cap R_{i+1:n}}+g(Z_n)\charfun{A_n},
\end{align*}
where
\begin{align*}
R_{j:k}=\braces{U_j\le \epsilon,\dots,U_k\le \epsilon}, \qquad  A_{j}=\braces{U_j> \epsilon},
\end{align*}
hence we have
\begin{align*}
&\mathbb{E}[g(X_m)g(X_n)] \\
&\quad=\mathbb{E}[g(X_m)g(X_m)\charfun{R_{m+1:n}}]+\charfun{n-m\ge 2}\sum_{i=m+1}^{n-1}\mathbb{E}[g(X_m)g(Z_i)\charfun{A_i\cap R_{i+1:n}}] \\
&\qquad+\mathbb{E}[g(X_m)g(Z_n)\charfun{A_n}].
\end{align*}
Therefore,
\begin{align*}
\mathbb{E}[g(X_m)g(X_n)]
=\mathbb{E}[g(X_m)^2]\mathbb {E}[\charfun{R_{m+1:n}}] 
=\mathbb{E}[g(X_m)^2]\epsilon^{n-m} 
=\epsilon^{n-m},
\end{align*}
and it follows that
\begin{align*}
n\cdot \mathrm{var}\Big(\frac{\sum_{i=1}^n g(X_i)}{n}\Big)
&=\frac{1}{n}\mathbb{E}\Big[\Big(\sum_{i=1}^n g(X_i)\Big)^2\Big] \\
&=\frac{1}{n}\sum_{i=1}^n \sum_{j=1}^n\mathbb{E}[g(X_i)g(X_j)] \\
&=\frac{1}{n}\Big(\sum_{i=1}^n \mathbb{E}[g(X_i)^2]+2\sum_{i=1}^{n-1} \sum_{j=i+1}^n \mathbb{E}[g(X_i)g(X_j)]\Big) \\
&=1+\frac{2}{n}\sum_{i=1}^{n-1}\sum_{j=i+1}^n\epsilon^{j-i} \\
&=1+\frac{2}{n}\sum_{i=1}^{n}(n-i)\epsilon^i.
\end{align*}
We have $\sum_{i=1}^\infty \epsilon^i=\frac{\epsilon}{1-\epsilon}$, hence by Kronecker
\begin{align*}
\sum_{i=1}^{n}\frac{i\epsilon^i}{n}\xrightarrow{n\to \infty} 0.
\end{align*}
Thus,
\begin{align*}
\lim_{n\to \infty}n\cdot \mathrm{var}\Big(\frac{\sum_{i=1}^n g(X_i)}{n}\Big) &=1+2\sum_{i=1}^\infty \epsilon^i =\frac{1+\epsilon}{1-\epsilon}
\end{align*}
and it follows that
\begin{align*}
\lim_{n\to \infty}n\cdot \mathrm{var}\Big(\frac{\sum_{i=1}^n f(X_i)}{n}\Big) &=\frac{1+\epsilon}{1-\epsilon}\|f\|_\pi^2.
\end{align*}
Since $(X_n)_{n=0}^\infty$ is a Markov chain with Markov transition probability $P_\epsilon$ and initial distribution $\pi$, we have
\begin{equation*}
\mathrm{var}(P_\epsilon,f) =\frac{1+\epsilon}{1-\epsilon}\|f\|_\pi^2.	\qedhere
\end{equation*}
\end{proof}

We will give the following convexity result:
\begin{lemma}\label{lemma:convexity-of-ratio}
Let $f:(1,\infty)\to [0,1)$ be positive convex and strictly decreasing.
Then,
$$\lambda\mapsto \frac{1+f(\lambda)}{1-f(\lambda)}, \qquad \lambda>1$$
is positive strictly convex and decreasing function.
\end{lemma}
\begin{proof}
By noticing that $g:[0,1)\to \R$, $g(x)=\frac{1+x}{1-x}$ is positive strictly increasing strictly convex function, we get that $\lambda \mapsto g(f(\lambda))$ is positive strictly decreasing function, and it holds for $1<\lambda_1<\lambda_2$ and $t\in (0,1)$ that
\begin{equation*}
g\big(f(t\lambda_1+(1-t)\lambda_2)\big)\le g\big(t f(\lambda_1)+(1-t)f(\lambda_2)\big)<tg\big(f(\lambda_1)\big)+(1-t)g\big(f(\lambda_2)\big).	\qedhere
\end{equation*}
\end{proof}

\begin{lemma}\label{lemma:convexity-of-the-estimate}
The function
$$\lambda\mapsto \frac{1+\epsilon(\lambda)}{1-\epsilon(\lambda)}, \qquad \lambda>1$$
is positive, strictly convex and decreasing.
\end{lemma}
\begin{proof}
Since $\epsilon(1)=1$, the claim follows from Theorem \ref{theorem:i-sir-convexity-2} and Lemma \ref{lemma:convexity-of-ratio}. 
\end{proof}

When integrating holding probability, we can divide the integral into two parts
\begin{align*}
\int_{\Sp}P_\lambda(x,\braces{x})\pi(\ud x)&=\int_{\Sp_0}\epsilon(\lambda,x)\pi(\ud x)+\int_{\Sp_+}P_\lambda(x,\braces{x})\pi(\ud x),\\
\int_{\Sp}\pi(\braces{x})\pi(\ud x)&= \int_{\Sp_+}\pi(\braces{x})\pi(\ud x).
\end{align*}
The equations above are helpful in the following lemma:
\begin{lemma}\label{lemma:convexity-of-the-estimate-2}
If $\#(\Sp)>1$, then $\psi:[1,\infty)\to \R$,
$$\psi(\lambda)=\frac{\int_{\Sp}[P_\lambda(x,\braces{x})-\pi(\braces{x})]\pi(\ud x)}{1-\int_{\Sp}\pi(\braces{x})\pi(\ud x)}$$ 
is positive convex and strictly decreasing function, and
$$\lambda\mapsto \frac{1+\psi(\lambda)}{1-\psi(\lambda)}, \qquad \lambda>1$$
is positive strictly convex and decreasing function.
\end{lemma}
\begin{proof}
Function $\lambda\mapsto \int_{\Sp_0}\epsilon(\lambda,x)\pi(\ud x)+\int_{\Sp_+}P_\lambda(x,\braces{x})\pi(\ud x) \in (\int_{\Sp_+} \pi(\braces{x})\pi(\ud x),1]$ is positive convex and strictly decreasing by Theorem \ref{theorem:i-sir-convexity-2}, \ref{theorem:i-sir-convexity-3} and Lemma \ref{lemma:holding-rejection-limit}. Therefore, we obtain the claim for $\psi$ by noticing $g:(\int_{\Sp_+}\pi(\braces{x})\pi(\ud x),1]\to (0,1]$, $g(y)=\frac{y-\int_{\Sp_+}\pi(\braces{x})\pi(\ud x)}{1-\int_{\Sp_+}\pi(\braces{x})\pi(\ud x)}$ is linear function. Since $\psi(1)=1$, we have $\psi(y)\in [0,1)$ for all $y\in (1,\infty)$, hence the second claim follows from Lemma \ref{lemma:convexity-of-ratio}.
\end{proof}

\section{Bounds for i-SIR rejection rate and asymptotic variance}\label{app:upper-lower-bounds}
We present upper and lower bounds for certain functions obtained from Algorithm \ref{alg:continuous-isir}. These bounds may also provide asymptotic speed of the function.

\subsection{Rejection probabilities}
When Assumption \ref{a:bounded-weights} holds, certain upper and lower bounds can be obtained for rejection probabilities with the help of Jensen's inequality.

\begin{lemma}\label{lemma:lower-and-upper-bound-for-rejection}
Let $\lambda\ge 1$ and $z_1\in \Sp$, then
\begin{align*}
0&<w(z_1)\frac{w(z_1)+2\floor{\lambda}-\lambda}{(w(z_1)+\floor{\lambda})(w(z_1)+\floor{\lambda}-1)} 
\le \epsilon(\lambda,z_1)\le \frac{2\hat{w}}{2\hat{w}+\lambda-1}, \\
\frac{1}{\lambda}&\le \frac{2\floor{\lambda}+1-\lambda}{(\floor{\lambda}+1)\floor{\lambda}}
\le \epsilon(\lambda)\le \frac{2\hat{w}}{2\hat{w}+\lambda-1},
\end{align*}
and for $\lambda>1$, it holds that
$$1<\frac{\floor{\lambda}^2+3\floor{\lambda}-\lambda+1}{\floor{\lambda}^2-\floor{\lambda}+\lambda-1}\le \frac{1+\epsilon(\lambda)}{1-\epsilon(\lambda)}\le \frac{4\hat{w}+\lambda-1}{\lambda-1},$$
where all the above upper bounds require Assumption \ref{a:bounded-weights} to hold.
\end{lemma}
\begin{proof}
Let $N\in \N_+$ and $z_1\in \Sp$, then $\epsilon(1)=\epsilon(1,z_1)=1$, and for $N\ge 2$, it holds by Jensen's inequality that
\begin{align*}
\epsilon(N,z_1)&=\int_{\X^{N-1}} \frac{w(z_1)}{\sum_{i=1}^Nw(z_i)} \bigg(\prod_{n=2}^N q(\ud z_n) \bigg)\ge \frac{w(z_1)}{w(z_1)+N-1}.
\end{align*}
Thus, for $\lambda\ge 1$, it holds by Lemma \ref{lemma:extended-function} in Appendix \ref{app:modification} that
\begin{align*}
\epsilon(\lambda,z_1)
\ge w(z_1)\frac{w(z_1)+2\floor{\lambda}-\lambda}{(w(z_1)+\floor{\lambda})(w(z_1)+\floor{\lambda}-1)}
\ge w(z_1)\frac{2\floor{\lambda}-\lambda}{\floor{\lambda}(w(z_1)+\floor{\lambda}-1)}
\end{align*}
and thus
\begin{align*}
\epsilon(\lambda)
&\ge \int_{\Sp} w(z_1)^2\frac{2\floor{\lambda}-\lambda}{\floor{\lambda}(w(z_1)+\floor{\lambda}-1)} q(\ud z_1) \\
&= \int_{\X} w(z_1)^2\frac{2\floor{\lambda}-\lambda}{\floor{\lambda}(w(z_1)+\floor{\lambda}-1)} q(\ud z_1) \\
&\ge \frac{2\floor{\lambda}+1-\lambda}{(\floor{\lambda}+1)\floor{\lambda}} \\
&\ge \frac{\floor{\lambda}+1}{(\floor{\lambda}+1)\lambda} \\
&= \frac{1}{\lambda},
\end{align*}
where the second inequality follows from Jensen's inequality, since for $a>0$, it holds that $g:\R\to \R$,
$$
g(x)=\begin{cases}
\frac{x^2}{x+a}, & x>0\\
0, & x\le 0
\end{cases}
$$
is convex.

Let us now assume that \ref{a:bounded-weights} holds. Let $N\in \braces{2,3,\dots}$, then
\begin{align*}
\epsilon(N,z_1)
=\int_{\X^{N-1}} \frac{w(z_1)}{\sum_{i=1}^Nw(z_i)} \bigg(\prod_{n=2}^N q(\ud z_n) \bigg) 
=1-\int_{\X^{N-1}} \frac{\sum_{j=2}^Nw(z_j)}{\sum_{i=1}^Nw(z_i)} \bigg(\prod_{n=2}^N q(\ud z_n) \bigg),
\end{align*}
where the last expression, by symmetricity of the integrand with respect to $z_2,\dots,z_N$, equals
\begin{align*}
&1-(N-1)\int_{\X^{N-1}} \frac{w(z_2)}{\sum_{i=1}^Nw(z_i)} \bigg(\prod_{n=2}^N q(\ud z_n) \bigg) \\
&\quad \le 1-(N-1)\int_{\X} \frac{w(z_2)}{w(z_1)+w(z_2)+N-2} q(\ud z_2) \\
&\quad \le 1-(N-1)\int_{\X} \frac{w(z_2)}{2\hat{w}+N-2} q(\ud z_2) \\
&\quad = 1-\frac{N-1}{2\hat{w}+N-2} \\
&\quad= \frac{2\hat{w}-1}{2\hat{w}+N-2},
\end{align*}
where the first inequality follows from Jensen's inequality. Thus,
\begin{align*}
\epsilon(N)=\int_\Sp \epsilon(N,z_1)\pi(\ud z_1)\le \frac{2\hat{w}-1}{2\hat{w}+N-2}.
\end{align*}
For $N=1$, the inequality is clear for both rejection probabilities. Therefore, if $\lambda\ge 1$, then by Lemma \ref{lemma:upper-bound-for-functions-of-certain-type} in Appendix \ref{app:modification},
\begin{align*}
\epsilon(\lambda,z_1)\le \frac{2\hat{w}}{2\hat{w}+\lambda-1}, \qquad
\epsilon(\lambda)\le \frac{2\hat{w}}{2\hat{w}+\lambda-1}.
\end{align*}
Therefore, for $\lambda>1$, it holds that
\begin{equation*}
1<\frac{\floor{\lambda}^2+3\floor{\lambda}-\lambda+1}{\floor{\lambda}^2-\floor{\lambda}+\lambda-1}
=\frac{1+\frac{2\floor{\lambda}+1-\lambda}{(\floor{\lambda}+1)\floor{\lambda}} }{1-\frac{2\floor{\lambda}+1-\lambda}{(\floor{\lambda}+1)\floor{\lambda}}}
\le \frac{1+\epsilon(\lambda)}{1-\epsilon(\lambda)}\le 
\frac{1+\frac{2\hat{w}}{2\hat{w}+\lambda-1}}{1-\frac{2\hat{w}}{2\hat{w}+\lambda-1}}
=\frac{4\hat{w}+\lambda-1}{\lambda-1}.	\qedhere
\end{equation*}
\end{proof}

We have by Jensen's inequality
\begin{align}\label{eq:ergodic-averages-jensen}
\begin{split}
\epsilon(\lambda)^2
&\le
\int_{\Sp\times \X^{\floor{\lambda}}} \bigg(\frac{\beta(\lambda)w(z_1)}{\sum_{i=1}^{\floor{\lambda}} w(z_i)}+ \frac{(1-\beta(\lambda))w(z_1)}{\sum_{j=1}^{\floor{\lambda}+1} w(z_j)}\bigg)^2 \bigg( \prod_{n=2}^{\floor{\lambda}+1} q(\ud z_n) \bigg)\pi(\ud z_1),
\end{split}
\end{align}
where $\beta(\lambda) = \floor{\lambda}+1-\lambda \in (0,1]$. We shall define the function on the right-hand side of \eqref{eq:ergodic-averages-jensen}, which was already mentioned in Section \ref{section:adaptation}:
For Algorithm \ref{alg:continuous-isir} the upper bound function of squared average rejection probability obtained by Jensen's inequality is denoted by $\epsilon_s:[1,\infty)\to \R$,
\begin{align}\label{eq:ergodic-averages-jensen-2}
\begin{split}
\epsilon_s(\lambda)&
=\int_{\Sp\times \X^{\floor{\lambda}}} \bigg(\frac{\beta(\lambda)w(z_1)}{\sum_{i=1}^{\floor{\lambda}} w(z_i)}+ \frac{(1-\beta(\lambda))w(z_1)}{\sum_{j=1}^{\floor{\lambda}+1} w(z_j)}\bigg)^2 \bigg( \prod_{n=2}^{\floor{\lambda}+1} q(\ud z_n) \bigg)\pi(\ud z_1) \\
&=\int_{\Sp\times \X^{\floor{\lambda}}} w(z_1)^3\bigg(\frac{\beta(\lambda)w(z_{\floor{\lambda}+1})+\sum_{k=1}^{\floor{\lambda}} w(z_k)}{(\sum_{i=1}^{\floor{\lambda}} w(z_i))(\sum_{j=1}^{\floor{\lambda}+1} w(z_j))}\bigg)^2 \bigg( \prod_{n=2}^{\floor{\lambda}+1} q(\ud z_n) \bigg)\pi(\ud z_1).
\end{split}
\end{align}

\begin{lemma}\label{lemma:ergodic-averages-jensen}
It holds that $\epsilon_s:[1,\infty)\to \R$ is strictly decreasing positive function. Let $\lambda\ge 1$, then $\epsilon(\lambda)^2\le \epsilon_s(\lambda)$, and $\epsilon_s(1)=1$. If Assumption \ref{a:bounded-weights} holds, then
$\epsilon_s(\lambda)\le \floor{\lambda}^{-1}\hat{w}$.
\end{lemma}
\begin{proof}
First inequality follows from \eqref{eq:ergodic-averages-jensen} and $\epsilon_s(1)=1$ by the definition.
When $N\in \N_+$ and $N\le \lambda_1<\lambda_2<N+1$, then by the definition of $\epsilon_s$, we have $\epsilon_s(\lambda_1)>\epsilon_s(\lambda_2)$. If $\lambda\in [N,N+1)$, then
\begin{align*}
&\epsilon_s(\lambda)-\epsilon_s(N+1) \\
&=\int_{\Sp\times \X^{N}} w(z_1)^2\bigg(\frac{\beta(\lambda)w(z_{N+1})+\sum_{k=1}^{N} w(z_k)}{(\sum_{j=1}^{N} w(z_j))(\sum_{l=1}^{N+1} w(z_l))}\bigg)^2 -\frac{w(z_1)^2}{(\sum_{l=1}^{N+1} w(z_l))^2} \bigg( \prod_{n=2}^{N+1} q(\ud z_n) \bigg)\pi(\ud z_1) \\
&=\int_{\Sp\times \X^{N}}w(z_1)^2\frac{(\beta(\lambda)w(z_{N+1})+\sum_{k=1}^{N} w(z_k))^2-(\sum_{i=1}^{N} w(z_i))^2}{(\sum_{j=1}^{N} w(z_j))^2(\sum_{l=1}^{N+1} w(z_l))^2} \bigg( \prod_{n=2}^{N+1} q(\ud z_n) \bigg)\pi(\ud z_1) \\
&>0.
\end{align*}
Therefore, $\epsilon_s$ is strictly decreasing. When Assumption \ref{a:bounded-weights} holds, we have
\begin{align*}
\epsilon_s(N)&=\int_{\W{N}} w(z_1)^3\bigg(\frac{1}{\sum_{i=1}^{N} w(z_i)}\bigg)^2 \bigg( \prod_{n=1}^{N} q(\ud z_n) \bigg) \\
&\le \hat{w}\int_{\W{N}}\frac{w(z_1)^2}{(\sum_{i=1}^{N} w(z_i))^2} \bigg( \prod_{n=1}^{N} q(\ud z_n) \bigg)  \\
&\le \hat{w}\int_{\W{N}}\frac{w(z_1)^2}{\sum_{i=1}^{N} w(z_i)^2} \bigg( \prod_{n=1}^{N} q(\ud z_n) \bigg) \\
&=\frac{\hat{w}}{N}\int_{\W{N}}\frac{\sum_{j=1}^{N} w(z_j)^2}{\sum_{i=1}^{N} w(z_i)^2} \bigg( \prod_{n=1}^{N} q(\ud z_n) \bigg) \\
&\le\frac{\hat{w}}{N},
\end{align*}
where the last equality follows from symmetricity of the integrand with respect to $z_1,\dots z_N$. Since $\epsilon_s$ is decreasing, we have $\epsilon_s(\lambda)\le \floor{\lambda}^{-1}\hat{w}$. 
\end{proof}

\begin{lemma}\label{lemma:weight-bounds}
It holds that $\int_\Sp w(x) \pi(\ud x)\ge 1$. Additionally ,under Assumption \ref{a:bounded-weights}, it holds that $\int_\Sp w(x) \pi(\ud x)\le \hat{w}$.
\end{lemma}
\begin{proof}
Let us define $\nu=\mu|_{\Sp}$, then by Hölder's inequality
\begin{align*}
\Big(\int_\Sp \Big|\frac{\pi(x)}{\sqrt{q(x)}}\sqrt{q(x)}\Big| \nu(\ud x) \Big)^2\le \int_\Sp \Big|\frac{\pi(x)}{\sqrt{q(x)}}\Big|^2 \nu(\ud x)\int_\Sp\Big|\sqrt{q(y)}\Big|^2 \nu(\ud y),
\end{align*}
hence
\begin{align*}
\int_\Sp w(x) \pi(\ud x)=\int_\Sp \frac{\pi(x)^2}{q(x)} \nu(\ud x)\ge \frac{\Big(\int_\Sp \pi(x) \nu(\ud x) \Big)^2}{\int_\Sp q(x) \nu(\ud x)}=\frac{1}{\int_\Sp q(x) \mu(\ud x)} \ge 1.
\end{align*}
The upper bound under Assumption \ref{a:bounded-weights} is obvious.
\end{proof}

Let us define for Algorithm \ref{alg:continuous-isir} the derivative of average rejection probability. The derivative function $\epsilon'$ is undefined for $\lambda\in \N_+$, which can be seen from \eqref{equation:continuous-versions} and Lemma \ref{lemma:i-sir-rejection-inequalities}. We will use the same notation for the extended derivative function:
$$\epsilon'(\lambda) = \epsilon(\floor{\lambda}+1)-\epsilon(\floor{\lambda}), \qquad \lambda\in [1,\infty),$$
which is a c\`adl\`ag function. The following bounds hold for this function:

\begin{lemma}\label{lemma:continuous-isir-rejection-derivative}
The extended derivative function $\epsilon'(\lambda):[1,\infty) \to \R$ is increasing negative function. Let $N\in \N_+$, then for $N\le x\le y<N+1$, we have $\epsilon'(x)=\epsilon'(y)$, and for $N\le x<N+1\le y$, we have $\epsilon'(x)<\epsilon'(y)$.
If Assumption \ref{a:bounded-weights} holds, then
\begin{align*}
\frac{C}{(2\hat{w}+\floor{\lambda}-1)^2} \le |\epsilon'(\lambda)| \le \frac{\hat{w}}{\floor{\lambda} (\floor{\lambda}+1)},
\end{align*}
where $C=\int_\Sp w(x) \pi(\ud x)\in [1,\hat{w}]$.
\end{lemma}
\begin{proof}
Let us first assume Assumption \ref{a:bounded-weights} holds. Let $N\in \N_+$ and let us denote $B_n=\sum_{i=1}^n w(z_i)$, then 
\begin{align*}
\epsilon(N)-\epsilon(N+1)
&= \int_{\Sp\times \X^{N}} \frac{w(z_1)w(z_{N+1})}{B_NB_{N+1}} \bigg( \prod_{n=2}^{N+1} q(\ud z_n) \bigg) \pi(\ud z_1) \\
&= \int_{\W{N}\times \X} \frac{w(z_1)^2w(z_{N+1})}{B_NB_{N+1}} \bigg( \prod_{n=1}^{N+1} q(\ud z_n) \bigg) \\
&\le \hat{w} \int_{\W{N}\times \X} \frac{w(z_1)w(z_{N+1})}{B_NB_{N+1}} \bigg( \prod_{n=1}^{N+1} q(\ud z_n) \bigg),
\end{align*}
where the last expression, by symmetricity of the integrand with respect to $z_1,\dots,z_N$, equals
\begin{align*}
\frac{\hat{w}}{N}\int_{\W{N}\times \X} \frac{B_Nw(z_{N+1})}{B_NB_{N+1}} \bigg( \prod_{n=1}^{N+1} q(\ud z_n) \bigg) 
&=\frac{\hat{w}}{N}\int_{\W{N+1}} \frac{w(z_{N+1})}{B_{N+1}} \bigg( \prod_{n=1}^{N+1} q(\ud z_n) \bigg) \\
&=\frac{\hat{w}}{N(N+1)}\int_{\W{N+1}} \frac{B_{N+1}}{B_{N+1}} \bigg( \prod_{n=1}^{N+1} q(\ud z_n) \bigg) \\
& \le \frac{\hat{w}}{N(N+1)},
\end{align*}
where the second equality follows from symmetricity of the integrand with respect to $z_1,\dots,z_{N+1}$.
We have also
\begin{align*}
\epsilon(N,z_1)-\epsilon(N+1,z_1)
&=\int_{\X^{N}} \frac{w(z_1)w(z_{N+1})}{B_NB_{N+1}} \bigg( \prod_{n=2}^{N+1} q(\ud z_n) \bigg) \\
&\ge \int_{\X^{N}} \frac{w(z_1)w(z_{N+1})}{(B_{N+1})^2} \bigg( \prod_{n=2}^{N+1} q(\ud z_n) \bigg) \\
&\ge \int_{\X^{N}} \frac{w(z_1)w(z_{N+1})}{(2\hat{w}+\sum_{i=2}^{N}w(z_i))^2} \bigg( \prod_{n=2}^{N+1} q(\ud z_n) \bigg) \\
&\ge \int_{\X} \frac{w(z_1)w(z_{N+1})}{(2\hat{w}+N-1)^2}q(\ud z_{N+1}) \\
&=\frac{w(z_1)}{(2\hat{w}+N-1)^2},
\end{align*}
where the last inequality follows from Jensen's inequality.
Therefore, the claims follow from Lemma \ref{lemma:i-sir-rejection-inequalities}, Lemma \ref{lemma:weight-bounds} and the definition of $\epsilon'$.
\end{proof}

By summing up Lemma \ref{lemma:continuous-isir-rejection-derivative}, we have that the extended derivative function $\epsilon':[1,\infty)\to \R$ is negative step function with positive steps on each integer and $|\epsilon'(\lambda)|=O(\lambda^{-2})$. 
We can also give bounds for the holding probability:

\begin{lemma}\label{lemma:i-sir-probability-inequality-2}
If $\#(\Sp)>1$, then it holds for all $x\in \Sp$ and $\lambda\ge 1$ that
\begin{align*}
\pi(\braces{x})&<P_\lambda(x,\braces{x})\le \pi(\braces{x})+(1-\pi(\braces{x}))\frac{2\hat{w}}{2\hat{w}+\lambda-1}, \\
\int_{\Sp}\pi(\braces{x})\pi(\ud x)&<\int_{\Sp}P_\lambda(x,\braces{x})\pi(\ud x)\le \int_{\Sp}\pi(\braces{x}) \pi(\ud x)+\Big(1-\int_{\Sp}\pi(\braces{x})\pi(\ud x)\Big)\frac{2\hat{w}}{2\hat{w}+\lambda-1}
\end{align*}
and for $\psi:[1,\infty)\to \R$,
$$\psi(\lambda)=\frac{\int_{\Sp}[P_\lambda(x,\braces{x})-\pi(\braces{x})]\pi(\ud x)}{1-\int_{\Sp}\pi(\braces{x})\pi(\ud x)},$$
it holds that
\begin{align*}
0<\psi(\lambda)\le \frac{2\hat{w}}{2\hat{w}+\lambda-1}, \qquad 1< \frac{1+\psi(\lambda)}{1-\psi(\lambda)}\le \frac{4\hat{w}+\lambda-1}{\lambda-1},
\end{align*}
where all the above upper bounds require Assumption \ref{a:bounded-weights} to hold.
\end{lemma}
\begin{proof}
We have $P_\lambda(x,\braces{x})>\pi(\braces{x})$ by Lemma \ref{lemma:holding-rejection-limit} and Theorem \ref{theorem:i-sir-convexity-3}.
Rest of the inequalities follow from each other and Lemma \ref{lemma:i-sir-probability-inequality} in Appendix \ref{app:ergodicity-auxiliaries} when $A=\braces{x}$.
\end{proof}

\subsection{Asymptotic variance}

When Assumption \ref{a:bounded-weights} holds, we can give upper and lower bounds for the asymptotic variance:

\begin{lemma}\label{lemma:lower-and-upper-bound-for-asymptotic-variance}
If Assumption \ref{a:bounded-weights} holds, $\lambda>1$ and $f\in L_{0,+}^2(\pi)$, then
\begin{align*}
\mathrm{var}_\pi(f)<\mathrm{var}(P_\lambda,f)\le \frac{4\hat{w}+\lambda-1}{\lambda-1}\mathrm{var}_\pi(f).
\end{align*}
\end{lemma}
\begin{proof}
The transition probability $P_\lambda$ is $\pi$-reversible by Lemma \ref{lemma:i-sir-properties}, hence $\angles{f|P_\lambda g}_\pi=\angles{P_\lambda f|g}_\pi$ for all $g\in L^2(\pi)$, and it also holds that $\angles{g|P_\lambda g}_\pi\ge 0$ and $\angles{f|P_\lambda f}_\pi>0$ by Theorem \ref{theorem:i-sir-convexity}. Therefore,
$$\angles{f|(I-P_\lambda)^{-1}f}_\pi=\sum_{k=0}^\infty \angles{f|P_\lambda^k f}_\pi=\sum_{j=0}^\infty \angles{P_\lambda^jf|P_\lambda^jf}_\pi+\sum_{k=0}^\infty \angles{P_\lambda^kf|P_\lambda P_\lambda^kf}_\pi> \angles{f|f}_\pi$$
and thus
$$\mathrm{var}(P_\lambda,f)=2\angles{f|(I-P_\lambda)^{-1}f}_\pi-\angles{f|f}_\pi > \angles{f|f}_\pi=\mathrm{var}_\pi(f).$$
The upper bound can be found by combining \cite[Proposition 32 in the supplement]{andrieu-lee-vihola} with Lemma \ref{lemma:i-sir-probability-inequality} in Appendix \ref{app:ergodicity-auxiliaries}.
\end{proof}

\section{Cost function}\label{app:cost-function}
The following holds for affine cost function:
\begin{proposition}\label{proposition:lower-and-upper-bound-cost-function}
Let $c,u:(1,\infty)\to \R$, $c(\lambda)=a+b\lambda$ and $u(\lambda)=c(\lambda)\frac{d+\lambda-1}{\lambda-1}$, where $a\ge 0$, $b>0$ and $d>0$, then
\begin{enumerate}[(i)]
\item $u$ is strictly convex positive function, \label{enum:upper-bound-function-1}
\item $\lim_{\lambda\to \infty}u(\lambda)=\lim_{\lambda\to \infty}c(\lambda)=\infty$ and $\lim_{\lambda\to 1_+}u(\lambda)=\infty$, \label{enum:upper-bound-function-2}
\item $\lim_{\lambda\to \infty}[u(\lambda)-c(\lambda)]=bd$, \label{enum:upper-bound-function-3}
\item the minimum value of $u$ is $2\sqrt{bd(a+b)}+a+b+bd$ attained at $\sqrt{d(a/b+1)}+1$ and the function $c$ attains this value at
$2\sqrt{d(a/b+1)}+d+1$. \label{enum:upper-bound-function-4}
\end{enumerate}
\end{proposition}
\begin{proof}
Let $g_1,g_2:(1,\infty)\to \R$, 
$g_1(\lambda)=\frac{d+\lambda-1}{\lambda-1}$
and
$g_2(\lambda)=\lambda\frac{d+\lambda-1}{\lambda-1}$,
then
\begin{align*}
\frac{\partial}{\partial \lambda}g_1(\lambda)&=-\frac{d}{(\lambda-1)^2}, 
\qquad \frac{\partial^2}{\partial \lambda^2}g_1(\lambda)=\frac{2d}{(\lambda-1)^3}, \\
\frac{\partial}{\partial \lambda}g_2(\lambda)&=1-\frac{d}{(\lambda-1)^2}, 
\qquad \frac{\partial^2}{\partial \lambda^2}g_2(\lambda)=\frac{2d}{(\lambda-1)^3}.
\end{align*}
Therefore,
\begin{align*}
u(\lambda)=ag_1(\lambda)+bg_2(\lambda)>0, 
\qquad
\frac{\partial}{\partial \lambda}u(\lambda)=b-\frac{d(b+a)}{(\lambda-1)^2}, 
\qquad
\frac{\partial^2}{\partial \lambda^2}u(\lambda)=\frac{2d(b+a)}{(\lambda-1)^3}>0,
\end{align*}
from which we obtain \eqref{enum:upper-bound-function-1}. The minimum value of $u$ can be solved by finding the root of the derivative, which is $\lambda_m=1+\sqrt{\frac{d(b+a)}{b}}$, thus the minimum value is
$$m
=u(\lambda_m)
=\Big(a+b+\sqrt{bd(b+a)}\Big)\Big(\sqrt{\frac{bd}{b+a}}+1\Big)
=2\sqrt{bd(a+b)}+a+b+bd,$$
and if $c(\lambda)=m$, then $\lambda=\frac{m-a}{b}$ and we have \eqref{enum:upper-bound-function-4}. 
We obtain \eqref{enum:upper-bound-function-3} from
\begin{align*}
u(\lambda)-c(\lambda)=\frac{d(a+b\lambda)}{\lambda-1}\xrightarrow{\lambda\to \infty}bd,
\end{align*}
and \eqref{enum:upper-bound-function-2} is clear by the definitions.
\end{proof}

\begin{lemma}\label{lemma:cost-function-with-asymptotic-variance}
Let Assumption \ref{a:bounded-weights} hold, $f\in L_{0,+}^2(\pi)$ and $\lambda_1,\lambda_2\in (1,2]$, $\lambda_1<\lambda_2$, then
$$\lambda_1\cdot  \mathrm{var}(P_{\lambda_1},f)-\lambda_2\cdot \mathrm{var}(P_{\lambda_2},f)>0.$$
\end{lemma}
\begin{proof}
Let $g\in L_0^2(\pi)$, then
\begin{align*}
\mathcal{E}_{\lambda_1}(g)&=(\lambda_1-\lambda_2)\mathcal{E}_{2}(g)+\mathcal{E}_{\lambda_2}(g).
\end{align*}
By Lemma \ref{lemma:i-sir-properties}, we can set $h_{\lambda_2}=(I-P_{\lambda_2})^{-1}f$, for which 
$$\mathrm{var}(P_{\lambda_2},f)=2\angles{f|h_{\lambda_2}}_\pi-\angles{f|f}_\pi,
\qquad \mathcal{E}_{\lambda_2}(h_{\lambda_2})=\angles{h_{\lambda_2}|f}_\pi=\angles{f|h_{\lambda_2}}_\pi,
$$
and by Lemma \ref{lemma:inverse} in Appendix \ref{app:hilbert-spaces}
$$\mathrm{var}(P_{\lambda_1},f)=2\sup_{g\in L_0^2(\pi)}[2\angles{f|g}_\pi-\mathcal{E}_{\lambda_1}(g)]-\angles{f|f}_\pi\ge 2[2\angles{f|h_{\lambda_2}}_\pi-\mathcal{E}_{\lambda_1}(h_{\lambda_2})]-\angles{f|f}_\pi.$$
Therefore,
\begin{align*}
&\lambda_1\cdot  \mathrm{var}(P_{\lambda_1},f)-\lambda_2 \cdot \mathrm{var}(P_{\lambda_2},f) \\
&\quad \ge \lambda_1 2[2\angles{f|h_{\lambda_2}}_\pi-\mathcal{E}_{\lambda_1}(h_{\lambda_2})]-\lambda_2 2 \angles{f|h_{\lambda_2}}_\pi+(\lambda_2-\lambda_1) \angles{f|f}_\pi \\
&\quad =\lambda_1 2[2\mathcal{E}_{\lambda_2}(h_{\lambda_2})-(\lambda_1-\lambda_2)\mathcal{E}_{2}(h_{\lambda_2})-\mathcal{E}_{\lambda_2}(h_{\lambda_2})]-2\lambda_2 \mathcal{E}_{\lambda_2}(h_{\lambda_2})+(\lambda_2-\lambda_1) \angles{f|f}_\pi \\
&\quad =\lambda_1 2(\lambda_2-\lambda_1)\mathcal{E}_{2}(h_{\lambda_2})-2(\lambda_2-\lambda_1) \mathcal{E}_{\lambda_2}(h_{\lambda_2})+(\lambda_2-\lambda_1) \angles{f|f}_\pi \\
&\quad \ge 2(\lambda_2-\lambda_1)[\mathcal{E}_{2}(h_{\lambda_2})-\mathcal{E}_{\lambda_2}(h_{\lambda_2})]+(\lambda_2-\lambda_1) \angles{f|f}_\pi \\
&\quad \ge (\lambda_2-\lambda_1) \angles{f|f}_\pi \\
&\quad>0,
\end{align*}
where the second last inequality follows from Theorem \ref{theorem:i-sir-convexity} and the fact that $\lambda_1>1$.
\end{proof}

For Algorithm \ref{alg:continuous-isir}, derivatives of holding probability and $\psi$ (given in \eqref{function:problem-solution}) are defined for non-integer $\lambda>1$,
\begin{align*}
\frac{\partial}{\partial \lambda}P_\lambda(x,\braces{x})&=P_{\floor{\lambda+1}}(x,\braces{x})-P_{\floor{\lambda}}(x,\braces{x}), \qquad \psi'(\lambda)=\psi(\floor{\lambda+1})-\psi(\floor{\lambda}).
\end{align*}

\begin{lemma}\label{lemma:cost-function-with-approximate}
Let $g,h:(1,\infty)\to \R$,
$$g(\lambda)=\lambda\frac{1+\epsilon(\lambda)}{1-\epsilon(\lambda)}, \qquad h(\lambda)=\lambda\frac{1+\psi(\lambda)}{1-\psi(\lambda)},$$
then $g$ is strictly decreasing and strictly convex when $\lambda\in (1,2]$, same holds for $h$ if $\#(\Sp)>1$.
\end{lemma}
\begin{proof}
We have for $\lambda>1$ that $0<\epsilon(\lambda)<1$, and $\epsilon$ is decreasing by Theorem \ref{theorem:i-sir-convexity-2}. Same properties hold for $\psi$ when $\#(\Sp)>1$ by Lemma \ref{lemma:convexity-of-the-estimate-2} in Appendix \ref{app:approximate-properties}.
We have for $\lambda\in (1,2)$ that
$$g'(\lambda)=\frac{2\lambda \epsilon'(\lambda)+1-\epsilon(\lambda)^2}{(1-\epsilon(\lambda))^2},$$
where
\begin{align*}
2\lambda\epsilon'(\lambda)+1-\epsilon(\lambda)^2
&=2\lambda\big(\epsilon(2)-1\big)+1-\big((2-\lambda)+(\lambda-1) \epsilon(2)\big)^2 \\
&=2\lambda\big(\epsilon(2)-1\big)+2(\lambda-1)\big(1-\epsilon(2)\big)-(\lambda-1)^2\big(1-\epsilon(2)\big)^2 \\
&=-2\big(1-\epsilon(2)\big)-(\lambda-1)^2\big(1-\epsilon(2)\big)^2 \\
&<0,
\end{align*}
and
\begin{align*}
g''(\lambda)&=\frac{(1-\epsilon(\lambda))[2\epsilon'(\lambda)-2\epsilon(\lambda)\epsilon'(\lambda)]+2\epsilon'(\lambda)[2\lambda \epsilon'(\lambda)+1-\epsilon(\lambda)^2]}{(1-\epsilon(\lambda))^3} \\
&=-\frac{2\epsilon'(\lambda)\big[-(1-\epsilon(\lambda))^2+2\big(1-\epsilon(2)\big)+(\lambda-1)^2\big(1-\epsilon(2)\big)^2\big]}{(1-\epsilon(\lambda))^3} \\
&\ge-\frac{2\epsilon'(\lambda)\big[\big(1-\epsilon(2)\big)+(\lambda-1)^2\big(1-\epsilon(2)\big)^2\big]}{(1-\epsilon(\lambda))^3} \\
&>0.
\end{align*}
Since $g$ is continuous, we get the claim. The proof is similar for $h$.
\end{proof}

When we combine the affine cost function with functions in \eqref{eq:v-g-h}, we obtain the following proposition:

\begin{proposition}\label{propostion:cost-function-with-asymptotic-variance-1}
Let $c:(1,\infty)\to \R$, $c(\lambda)=a+b\lambda$, where $a\ge 0$ and $b>0$. Let Assumption \ref{a:bounded-weights} hold and $f\in L_{0,+}^2(\pi)$, then $\lambda\mapsto c(\lambda) V_f(\lambda)$, $\lambda\mapsto c(\lambda) G_f(\lambda)$ and $\lambda\mapsto c(\lambda) H_f(\lambda)$ are strictly decreasing when $\lambda\in (1,2]$.
\end{proposition}
\begin{proof}
Follows from Theorem \ref{theorem:i-sir-asymptotic-variance}, Lemma \ref{lemma:convexity-of-the-estimate} and \ref{lemma:convexity-of-the-estimate-2} in Appendix \ref{app:approximate-properties} combined with Lemma \ref{lemma:cost-function-with-asymptotic-variance} and \ref{lemma:cost-function-with-approximate}.
\end{proof}

\section{Finite discrete case for i-SIR}\label{app:discrete-case-isir}
Estimating the asymptotic variance of i-SIR can be challenging. In the finite discrete case asymptotic variance can be solved via the formula in Lemma \ref{lemma:asymptotic-variance-of-finite-i-SIR} in Appendix \ref{app:asymptotic-variance-results}, since the transition probability matrix $P_\lambda$ as well as the eigenvalues and eigenvectors can be solved, even though this can be extremely slow for large $\lambda$ when $\Sp$ has many states. However, in the finite case transition probability matrix $P_N$ can be expressed with the help of multinomial distribution:

\begin{proposition}\label{proposition:discrete-isir-matrix-and-rejection}
Let us have Algorithm \ref{alg:isir} with finite discrete $\pi$, $\Sp=\braces{s_1,s_2,\dots,s_n}$, $n\in \N_+$ and let us denote $U=\mathrm{supp}(q)\setminus \Sp$.
\begin{enumerate}[(i)]
\item If $U=\emptyset$, let $Q=(q(\braces{s_1}),\dots,q(\braces{s_n}))$. \label{enum:case-1-condition}
\item If $U\neq \emptyset$, let $Q=(q(\braces{s_1}),\dots,q(\braces{s_n}),q(U))$. \label{enum:case-2-condition}
\end{enumerate}
For both cases \eqref{enum:case-1-condition} and \eqref{enum:case-2-condition}, we have for $N\in \N$, $N\ge 2$ and $i,j\in \braces{1,\dots,n}$ that
\begin{align}
P_N(s_i,\braces{s_j})&=\E\Big[(\charfun{i=j}+Z_j)\frac{w(s_j)}{w(s_i)+\sum_{k=1}^{n} Z_kw(s_k)}\Big], \label{eq:transition-probability-as-expectation} \\
\epsilon(N,s_i)&=\E\Big[\frac{w(s_i)}{w(s_i)+\sum_{k=1}^{n} Z_kw(s_k)}\Big], \label{eq:rejection-probability-as-expectation}\\
\epsilon(N)&=\E\Big[\frac{w(Y)}{w(Y)+\sum_{k=1}^{n} Z_kw(s_k)}\Big], \nonumber
\end{align}
where $Z\sim \mathrm{Multinom}(N-1,Q)$ and $Y\sim \pi$ are independent.
\end{proposition}
\begin{proof}
Let us first take $U=\emptyset$ and $Q=(q(\braces{s_1}),\dots,q(\braces{s_n}))$. For Algorithm \ref{alg:isir}, we have $Y_1^{1}=X_{0}$, $Y_1^{2:N}\overset{\text{i.i.d.}}{\sim}q$ and $I_1 \sim \mathrm{Categorical}(W_1^{1:N})$. Therefore,
\begin{equation*}
Z_1=\Big(\sum_{j=2}^N\charfun{Y_1^j=s_1},\dots, \sum_{j=2}^N\charfun{Y_1^j=s_n}\Big)\sim \mathrm{Multinom}(N-1,Q).
\end{equation*}
Let $v\in \R^n$ and let us denote
\begin{equation*}
A_1^{v}=\braces{Z_1=v}, \qquad B=\Bigbraces{u\in \N^n:\sum_{k=1}^n u_k=N-1}.
\end{equation*}
Then, it holds for $i,j\in \braces{1,\dots,n}$ that
\begin{align*}
\P(X_1=s_j\mid X_0=s_i)&=\sum_{v\in B}\P(X_1=s_j,A_1^v \mid X_0=s_i) \\
&=\sum_{v\in B}\P(X_1=s_j \mid X_0=s_i,A_1^v)\P(A_1^v)\\
&=\sum_{v\in B}\sum_{k=1}^N\P(X_1=s_j,I_1=k \mid X_0=s_i,A_1^v)\P(A_1^v) \\
&=\begin{cases}
\sum_{v\in B}(1+v_i)\frac{w(s_i)}{w(s_i)+\sum_{k=1}^nv_kw(s_k)}\P(A_1^v), & i=j \\
\sum_{v\in B}v_j\frac{w(s_j)}{w(s_i)+\sum_{k=1}^nv_kw(s_k)}\P(A_1^v), &i\neq j.
\end{cases}
\end{align*}
For the rejection
\begin{align*}
\P(I_1=1 \mid X_0=s_i)&=\sum_{v\in B}\P(I_1=1,A_1^v \mid X_0=s_i) \\
&=\sum_{v\in B}\P(I_1=1 \mid X_0=s_i,A_1^v)\P(A_1^v) \\
&=\sum_{v\in B}\frac{w(s_i)}{w(s_i)+\sum_{k=1}^nv_kw(s_k)}\P(A_1^v).
\end{align*}
Since $(X_i)$ is Markov chain, we obtain the case for $U=\emptyset$. The case $U\neq \emptyset$ follows similarly when we have instead $Q=(q(\braces{s_1}),\dots,q(\braces{s_n}),q(U))$, $v\in \R^{n+1}$,
\begin{align*}
Z_1&=\Big(\sum_{j=2}^N\charfun{Y_1^j=s_1},\dots, \sum_{j=2}^N\charfun{Y_1^j=s_n}, \sum_{j=2}^N\charfun{Y_1^j\notin \Sp}\Big)\sim \mathrm{Multinom}(N-1,Q), \\
A_1^{v}&=\braces{Z_1=v}, \qquad B=\Bigbraces{u\in \N^{n+1}:\sum_{k=1}^{n+1} u_k=N-1},
\end{align*}
since $w(s)=0$ for all $s\notin \Sp$.
\end{proof}
Proposition \ref{proposition:discrete-isir-matrix-and-rejection} can be easily generalised to Algorithm \ref{alg:continuous-isir} and gives possibilities to estimate the probability transition matrix, holding probability and rejection probabilities with multinomial distribution. We have also the following special case:

\begin{proposition}\label{proposition:two-state-isir}
If $\#(\Sp)=2$, that is, $\Sp=\braces{s_1,s_2}$, then it holds for $f\in L_{0,+}^2(\pi)$ that
$$H_{f}(\lambda)=\mathrm{var}(P_\lambda,f)=\var_\pi(f)\frac{P_\lambda(s_1,\braces{s_1})+P_\lambda(s_2,\braces{s_2})}{2-(P_\lambda(s_1,\braces{s_1})+P_\lambda(s_2,\braces{s_2}))}.$$
\end{proposition}
\begin{proof}
We have $P_\lambda$ is $\pi$-reversible by Lemma \ref{lemma:i-sir-properties}. By direct calculations, we can find eigenvalues $\lambda_1=1$, $\lambda_2=P_\lambda(s_1,\braces{s_1})+P_\lambda(s_2,\braces{s_2})-1$ and eigenvectors
$v_1=(1,1)$,
$v_2=(-\sqrt{\frac{\pi(\braces{s_2})}{\pi(\braces{s_1})}},\sqrt{\frac{\pi(\braces{s_1})}{\pi(\braces{s_2})}})$
to match conditions in Lemma \ref{lemma:asymptotic-variance-of-finite-i-SIR} in Appendix \ref{app:asymptotic-variance-results} and thus
\begin{align*}
\mathrm{var}(P_\lambda,f)
&=\angles{v_2|f}_\pi^2\frac{P_\lambda(s_1,\braces{s_1})+P_\lambda(s_2,\braces{s_2})}{2-(P_\lambda(s_1,\braces{s_1})+P_\lambda(s_2,\braces{s_2}))} \\
&=\var_\pi(f)\frac{P_\lambda(s_1,\braces{s_1})+P_\lambda(s_2,\braces{s_2})}{2-[P_\lambda(s_1,\braces{s_1})+P_\lambda(s_2,\braces{s_2})]},
\end{align*}
since
\begin{align*}
\angles{v_2|f}_\pi^2&=\bigg(-\sqrt{\frac{\pi(\braces{s_2})}{\pi(\braces{s_1})}}f(s_1)\pi(\braces{s_1})+\sqrt{\frac{\pi(\braces{s_1})}{\pi(\braces{s_2})}}f(s_2)\pi(\braces{s_2})\bigg)^2 \\
&=\pi(\braces{s_2})\pi(\braces{s_1})\big(f(s_1)-f(s_2)\big)^2,
\end{align*}
$\E_\pi[f]=0$ and
\begin{align*}
\var_\pi(f)
&=\pi(\braces{s_1})[\pi(\braces{s_1})+\pi(\braces{s_2})]f(s_1)^2+\pi(\braces{s_2})[\pi(\braces{s_1})+\pi(\braces{s_2})]f(s_2)^2 \\
&=2\pi(\braces{s_1})\pi(\braces{s_2})f(s_1)f(s_2)+\pi(\braces{s_1})\pi(\braces{s_2})[f(s_1)-f(s_2)]^2+\sum_{i=1}^2\pi(\braces{s_i})^2f(s_i)^2 \\
&=\pi(\braces{s_1})\pi(\braces{s_2})[f(s_1)-f(s_2)]^2.
\end{align*}
We have
$1-\pi(\braces{s_1})^2-\pi(\braces{s_2})^2 
=2\pi(\braces{s_1})\pi(\braces{s_2})$ and by reversibility
\begin{align*}
&-1+\sum_{i=1}^2 \pi(\braces{s_i})P_\lambda(s_i,\braces{s_i}) \\
&\quad=-\pi(\braces{s_1})P_\lambda(s_1,\braces{s_2})-\pi(\braces{s_2})P_\lambda(s_2,\braces{s_1}) \\
&\quad=-[\pi(\braces{s_1})+\pi(\braces{s_2})][\pi(\braces{s_1})P_\lambda(s_1,\braces{s_2})+\pi(\braces{s_2})P_\lambda(s_2,\braces{s_1})] \\
&\quad=-2\pi(\braces{s_1})\pi(\braces{s_2})[P_\lambda(s_1,\braces{s_2})+P_\lambda(s_2,\braces{s_1})].
\end{align*}
Therefore,
\begin{align*}
H_{f}(\lambda)&=\frac{1-2\pi(\braces{s_1})^2-2\pi(\braces{s_2})^2+\sum_{i=1}^2 \pi(\braces{s_i})P_\lambda(s_i,\braces{s_i})]}{1-\pi(\braces{s_1})P_\lambda(s_1,\braces{s_1})-\pi(\braces{s_2})P_\lambda(s_2,\braces{s_2})}\mathrm{var}_\pi(f) \\
&=\frac{2-P_\lambda(s_1,\braces{s_2})-P_\lambda(s_2,\braces{s_1})]}{P_\lambda(s_1,\braces{s_2})+P_\lambda(s_2,\braces{s_1})}\mathrm{var}_\pi(f) \\
&=\frac{P_\lambda(s_1,\braces{s_1})+P_\lambda(s_2,\braces{s_2})}{2-P_\lambda(s_1,\braces{s_1})-P_\lambda(s_2,\braces{s_2})]}\mathrm{var}_\pi(f)
\end{align*}
and we obtain the claim.
\end{proof}

\section{Some properties on Hilbert spaces and reversible Markov chains}\label{app:hilbert-spaces}

The results presented in this section are presented in many functional analysis books in the case of a complex Hilbert space. 
As mentioned in \citep{bogachev-smolyanov}, many of such results also hold in the real case for self-adjoint operators by complexification. For the sake of completeness, we detail the complexification arguments here.
We shall denote by $\|\cdot\|$ the norm induced by the inner product $\angles{\uarg|\uarg}$ of a Hilbert space $(\mathcal{H},\angles{\uarg|\uarg})$ that is either real or complex.

When $\mathcal{H}$ is real Hilbert space:
We define complex Hilbert space $(\mathcal{H}_\mathbb{C},\angles{\uarg|\uarg}_\mathbb{C})$, where $\mathcal{H}_\mathbb{C}=\braces{f+ig:f,g\in \mathcal{H}}$,
$$\angles{f+ig|u+iv}_\mathbb{C}=\angles{f|u}-i\angles{f|v}+i\angles{g|u}+\angles{g|v}, \qquad f,g,u,v\in \mathcal{H},$$
and for bounded operator $A$ on $\mathcal{H}$, we define the complexification $A_\mathbb{C}(f+ig)=Af+iAg$, where $f,g\in \mathcal{H}$. Then, $A_\mathbb{C}$ is bounded operator on $\mathcal{H}_\mathbb{C}$. Additionally, if $A$ is self-adjoint operator on a real Hilbert space $\mathcal{H}$, we have
\begin{align}\label{eq:complexification-inner-product}
\begin{split}
\angles{g+ih|A_{\mathbb{C}}(g+ih)}_{\mathbb{C}}&=\angles{g|Ag}-i\angles{g|Ah}+i\angles{h|Ag}+\angles{h|Ah} \\
&=\angles{g|Ag}-i\angles{Ag|h}+i\angles{h|Ag}+\angles{h|Ah} \\
&=\angles{g|Ag}+\angles{h|Ah}.
\end{split}
\end{align}

The following lemma is combination of Theorems from \citep[p.~330--331]{rudin}:
\begin{lemma}\label{lemma:positive-square-root}
Let $A$ be a bounded operator on a complex Hilbert space $\mathcal{H}$ satisfying $\angles{f|Af}\ge 0$ for all $f\in \mathcal{H}$. Then, it holds that $A^*=A$ and there exists square root $A^{1/2}$ for which $\angles{f|A^{1/2}f}\ge 0$ for all $f\in \mathcal{H}$. Additionally, if $A$ is invertible, so is $A^{1/2}$.
\end{lemma}

\begin{lemma}\label{lemma:positive-definite-included}
Let $A$ be a bounded self-adjoint operator on a Hilbert space $\mathcal{H}$, satisfying $\angles{f|Af}\ge 0$ for all $f\in \mathcal{H}$ and such that the inverse $A^{-1}$ exists. Then, it also holds that $\angles{f|Af}=0$ if and only if $\|f\|=0$.
\end{lemma}
\begin{proof}
It is clear by Cauchy--Schwarz that $\|f\|=0$ implies $\angles{f|Af}=0$. For the reverse it is enough to show that $\|f\|>0$ implies $\angles{f|Af}>0$.

If $\mathcal{H}$ is complex Hilbert space, then by Theorem \ref{lemma:positive-square-root} there exists invertible square root $A^{1/2}$ for which $(A^{1/2})^*=A^{1/2}$ and thus it holds for $f\in \mathcal{H}$, $\|f\|>0$ that
$$\angles{f|Af}=\angles{f|A^{1/2}A^{1/2}f}=\angles{A^{1/2}f|A^{1/2}f}>0,$$
where the strict inequality follows since $A^{1/2}$ is invertible.

In the case of real Hilbert space $\mathcal{H}$, we use complexification. Since $A$ is self-adjoint and $\angles{f|Af}\ge 0$ for all $f\in \mathcal{H}$, we have by \eqref{eq:complexification-inner-product} that $\angles{g|A_\mathbb{C}g}_{\mathbb{C}}\ge 0$ for all $g\in \mathcal{H}_\mathbb{C}$. Clearly, $(A^{-1})_\mathbb{C}$ is the inverse operator of $A_\mathbb{C}$ on $\mathcal{H}_\mathbb{C}$. Therefore, by the previous complex result, we have $\angles{g| A_\mathbb{C}g}_\mathbb{C}>0$ for all $g\in \mathcal{H}_\mathbb{C}$ and thus it holds for all $f\in \mathcal{H}$, $\|f\|>0$ that
\begin{equation*}
\angles{f|Af}=\angles{f+i0|A_\mathbb{C}(f+i0)}_\mathbb{C}>0.	\qedhere
\end{equation*}
\end{proof}

\begin{lemma}\label{lemma:invertible-requirement}
Let $A$ be a bounded self-adjoint operator on a Hilbert space $\mathcal{H}$, satisfying $\angles{f|Af}\ge 0$ for all $f\in \mathcal{H}$, then $A$ is invertible if there exists such $c>0$ that
$$\|Af\|\ge c\|f\|$$
for all $f\in \mathcal{H}$.
\end{lemma}
\begin{proof}
Result holds for complex Hilbert space \citep[e.g.][Theorem 12.12]{rudin}.

In the case of real Hilbert space $\mathcal{H}$, we use complexification. We have by \eqref{eq:complexification-inner-product} that $\angles{f|A_\mathbb{C}f}_\mathbb{C}\ge 0$ for all $f\in \mathcal{H}_\mathbb{C}$, and by Lemma \ref{lemma:positive-square-root}, it holds that $A_\mathbb{C}$ is self-adjoint operator on $\mathcal{H}_\mathbb{C}$. Additionally, for $f,g\in \mathcal{H}$, it holds that
\begin{align*}
\angles{f+ig|f+ig}_\mathbb{C}
=\angles{f|f}-i\angles{f|g}+i\angles{g|f}+\angles{g|g}
=\angles{f|f}+\angles{g|g},
\end{align*}
hence
\begin{align*}
\angles{A_\mathbb{C}(f+ig)|A_\mathbb{C}(f+ig)}_\mathbb{C}&=\angles{Af|Af}+\angles{Ag|Ag}\ge c^2\angles{f|f}+c^2\angles{g|g}=c^2\angles{f+ig|f+ig}_\mathbb{C}.
\end{align*}
By the complex case result, $A_\mathbb{C}:\mathcal{H}_\mathbb{C}\to \mathcal{H}_\mathbb{C}$ is invertible and therefore $A:\mathcal{H}\to \mathcal{H}$ is invertible.
\end{proof}

\begin{lemma}\label{lemma:inverse}
Let $A$ be a bounded self-adjoint operator on a Hilbert space $\mathcal{H}$, satisfying $\angles{f|Af}\ge 0$ for all $f\in \mathcal{H}$ and such that the inverse $A^{-1}$ exists. Then,
$$\angles{f|A^{-1}f}=\sup_{g\in \mathcal{H}}[2\angles{f|g}-\angles{g|Ag}],$$ 
where the supremum is attained uniquely with $g=A^{-1}f$ up to $\mathcal{H}$-equivalence class.
\end{lemma}
\begin{proof}
As in \cite[Lemma 16]{andrieu-vihola-order} equalities follow from noticing
\begin{equation}\label{equation:inverse-formula}
\angles{f|A^{-1}f}-[2\angles{f|g}-\angles{g|Ag}]=\angles{A^{-1}f-g,A(A^{-1}f-g)}\ge 0,
\end{equation}
where the equality holds for $g=A^{-1}f$. By Lemma \ref{lemma:positive-definite-included}, equation \eqref{equation:inverse-formula} is equal to zero if and only if $\|A^{-1}f-g\|=0$.
\end{proof}

\begin{assumption}\label{a:operator-injective}
Let $J\subset \R$ and $(A_\lambda)_{\lambda\in J}$ be collection of bounded self-adjoint operators on Hilbert space $\mathcal{H}$ for which $\angles{f|A_\lambda f}\ge 0$ for all $f\in \mathcal{H}$, the inverse $A_\lambda^{-1}$ exists and $\lambda \mapsto \angles{f|A_\lambda f}$ is injective for $\|f\|>0$ when $\lambda\in J$.
\end{assumption}

\begin{lemma}\label{lemma:not-equal-inverse-general}
If \ref{a:operator-injective} holds, $x,y\in J$, $x\neq y$ and $f\in \mathcal{H}$, then $\|A_{x}^{-1}f-A_{y}^{-1}f\|>0$ if $\|f\|>0$.
\end{lemma}
\begin{proof}
Let us make an assumption that $\|A_{x}^{-1}f-A_{y}^{-1}f\|=0$. Since $A_\lambda$ is self-adjoint for all $\lambda\in J$, we have
\begin{align*}
\angles{A_y^{-1}f,A_yA_y^{-1}f}
&=\angles{f,A_y^{-1}f}
=\angles{A_x^{-1}f,A_xA_y^{-1}f} 
=\angles{A_y^{-1}f,A_xA_y^{-1}f},
\end{align*}
where the assumption is used on the last equality. We have contradiction, since $\|A_y^{-1}f\|>0$ and $\lambda \mapsto \angles{A_y^{-1}f|A_\lambda A_y^{-1}f}$ is injective by \ref{a:operator-injective}.
\end{proof}

\begin{lemma}\label{lemma:not-equal-inverse-general-corollary}
If \ref{a:operator-injective} holds and $x,y\in J$, $x\neq y$, $f\in \mathcal{H}$, $\|f\|>0$, then there is no such $g\in \mathcal{H}$ that both equations
$$\sup_{h\in \mathcal{H}}[2\angles{f|h}-\angles{h|A_xh}]=2\angles{f|g}-\angles{g|A_xg}$$
and
$$\sup_{h\in \mathcal{H}}[2\angles{f|h}-\angles{h|A_yh}]=2\angles{f|g}-\angles{g|A_yg}$$
hold.
\end{lemma}
\begin{proof}
Follows from Lemma \ref{lemma:inverse} and \ref{lemma:not-equal-inverse-general}.
\end{proof}

The results above are used to obtain the following result for a collection of $\pi$-reversible Markov kernels:

\begin{theorem}\label{theorem:asymptotic-variance-results}
Let $J\subset \R$ and $(Q_\lambda)_{\lambda\in J}$ be collection of $\pi$-reversible Markov kernels and the inverse operator $(I-Q_\lambda)^{-1}$ exists for $(I-Q_\lambda):L_0^2(\pi)\to L_0^2(\pi)$. Let us define for fixed $h\in L_0^2(\pi)$ functions $u_h,v_h:J\to \R$, $u_h(\lambda)=\angles{h|Q_\lambda h}_\pi$, $v_h(\lambda)=\mathrm{var}(Q_\lambda,h)$. 

Let us state the following assumptions:
\begin{enumerate}[(a)]
\item $u_g$ is decreasing (increasing) function for all $g\in L_{0,+}^2(\pi)$. \label{enum-asymptotic-variance-condition-1}
\item $u_g$ is strictly decreasing (increasing) function for all $g\in L_{0,+}^2(\pi)$. \label{enum-asymptotic-variance-condition-2}
\item $u_g$ is convex function for all $g\in L_{0,+}^2(\pi)$, where $J$ is interval. \label{enum-asymptotic-variance-condition-3}
\item $u_g$ is strictly convex function for all $g\in L_{0,+}^2(\pi)$, where $J$ is interval. \label{enum-asymptotic-variance-condition-4}
\end{enumerate}

For $f\in L_{0,+}^2(\pi)$, it holds that
\begin{enumerate}[(i)]
\item $v_f$ is decreasing (increasing) function if \eqref{enum-asymptotic-variance-condition-1} holds.
\item $v_f$ is strictly decreasing (increasing) function if \eqref{enum-asymptotic-variance-condition-2} holds.
\item $v_f$ is convex function if \eqref{enum-asymptotic-variance-condition-3} holds.
\item $v_f$ is strictly convex function if \eqref{enum-asymptotic-variance-condition-4} holds.
\item $v_f$ is strictly convex and strictly decreasing (increasing) function if \eqref{enum-asymptotic-variance-condition-2} and \eqref{enum-asymptotic-variance-condition-3} hold. \label{enum-asymptotic-variance-result-5}
\end{enumerate}
\end{theorem}
\begin{proof}
Since $Q_\lambda$ is $\pi$-reversible, we have $I-Q_\lambda:L_0^2(\pi)\to L_0^2(\pi)$ is self-adjoint.
For $g\in L_0^2(\pi)$, let us denote the Dirichlet form by 
$$\mathcal{E}_\lambda(g)=\angles{g|(I-Q_\lambda)g}_\pi=\angles{g|g}_\pi-u_g(\lambda),$$
which is non-negative for the Markov kernel, and the function $\lambda\mapsto -\mathcal{E}_\lambda(g)$ preserves conditions in \eqref{enum-asymptotic-variance-condition-1}--\eqref{enum-asymptotic-variance-condition-4}.
Let us denote $h_\lambda=(I-Q_\lambda)^{-1}f$. By Lemma \ref{lemma:inverse},
$$2\angles{f|h_\lambda}_\pi-\mathcal{E}_\lambda(h_\lambda)=\sup_{h\in L_0^2(\pi)}[2\angles{f|h}_\pi-\mathcal{E}_\lambda(h)]\ge 2\angles{f|g}_\pi-\mathcal{E}_\lambda(g)$$
for all $g\in L_0^2(\pi)$ and
$$\mathrm{var}(Q_\lambda,f)=2\angles{f|(I-Q_\lambda)^{-1}f}_\pi-\angles{f|f}_\pi=2\sup_{h\in L_0^2(\pi)}[2\angles{f|h}_\pi-\mathcal{E}_{\lambda}(h)]-\angles{f|f}_\pi.$$
Let $x,y\in J$. When \eqref{enum-asymptotic-variance-condition-1} holds with the decreasing case, and $x<y$ (or \eqref{enum-asymptotic-variance-condition-1} holds with the increasing case, and $x>y$), then
\begin{align*}
\mathrm{var}(Q_{x},f)-\mathrm{var}(Q_{y},f)
&=2\Big(\sup_{g\in L_0^2(\pi)}[2\angles{f|g}_\pi-\mathcal{E}_{x}(g)]-\sup_{h\in L_0^2(\pi)}[2\angles{f|h}_\pi-\mathcal{E}_{y}(h)]\Big) \\
&\ge2\Big(\big(2\angles{f|h_{y}}_\pi-\mathcal{E}_{x}(h_{y})\big)-\big(2\angles{f|h_{y}}_\pi-\mathcal{E}_{y}(h_{y})\big)\Big) \\
&=2\big(-\mathcal{E}_{x}(h_{y})+\mathcal{E}_{y}(h_{y})\big) \\
&\ge 0,
\end{align*}
where the last inequality is strict if \eqref{enum-asymptotic-variance-condition-2} holds.
When \eqref{enum-asymptotic-variance-condition-3} holds, $x<y$ and $t\in (0,1)$, then
\begin{align*}
&\mathrm{var}(Q_{tx+(1-t)y},f) \\
&=2\Big(2\angles{f|h_{tx+(1-t)y}}_\pi-\mathcal{E}_{tx+(1-t)y}(h_{tx+(1-t)y})\Big) -\angles{f|f}_\pi\\
&\le2\Big(t[2\angles{f|h_{tx+(1-t)y}}_\pi-\mathcal{E}_{x}(h_{tx+(1-t)y})]+(1-t)[2\angles{f|h_{tx+(1-t)y}}_\pi-\mathcal{E}_{y}(h_{tx+(1-t)y})]\Big)-\angles{f|f}_\pi \\
&\le2\sup_{g\in L_0^2(\pi)}\Big(t[2\angles{f|g}_\pi-\mathcal{E}_{x}(g)]+(1-t)[2\angles{f|g}_\pi-\mathcal{E}_{y}(g)]\Big)-\angles{f|f}_\pi \\
&\le2\Big(t\sup_{g\in L_0^2(\pi)}[2\angles{f|g}_\pi-\mathcal{E}_{x}(g)]+(1-t)\sup_{h\in L_0^2(\pi)}[2\angles{f|h}_\pi-\mathcal{E}_{y}(h)]\Big)-\angles{f|f}_\pi \\
&=t \cdot \mathrm{var}(Q_x,f)+(1-t)\cdot \mathrm{var}(Q_y,f),
\end{align*}
where the first inequality is strict if \eqref{enum-asymptotic-variance-condition-4} holds. 
However, if \eqref{enum-asymptotic-variance-condition-2} and \eqref{enum-asymptotic-variance-condition-3} hold, then also \ref{a:operator-injective} holds and
\begin{align*}
&\mathrm{var}(Q_{tx+(1-t)y},f)  \\
&\le 2\Big(t[2\angles{f|h_{tx+(1-t)y}}_\pi-\mathcal{E}_{x}(h_{tx+(1-t)y})]+(1-t)[2\angles{f|h_{tx+(1-t)y}}_\pi-\mathcal{E}_{y}(h_{tx+(1-t)y})]\Big)-\angles{f|f}_\pi \\
&< 2\Big(t\sup_{g\in L_0^2(\pi)}[2\angles{f|g}_\pi-\mathcal{E}_{x}(g)]+(1-t)\sup_{h\in L_0^2(\pi)}[2\angles{f|h}_\pi-\mathcal{E}_{y}(h)]\Big)-\angles{f|f}_\pi \\
&=t \cdot \mathrm{var}(Q_x,f)+(1-t)\cdot \mathrm{var}(Q_y,f),
\end{align*}
where the strict inequality follows from Lemma \ref{lemma:not-equal-inverse-general-corollary}.
\end{proof}

\end{document}